\documentclass[11pt]{amsart} 
\usepackage{amscd,amssymb,amsxtra}
\usepackage[mathscr]{eucal}
\usepackage{mathabx}
\usepackage{comment}
\usepackage{color}
\usepackage{enumitem} 
\usepackage{mathtools}
\usepackage{booktabs}

\makeatletter
\let\@secnumfont\bfseries
\def\section{\@startsection{section}{1}%
  \z@{4\linespacing\@plus\linespacing}{\linespacing}%
  {\bfseries\centering}}
\def\introsection{\@startsection{section}{1}%
  \z@{3\linespacing\@plus\linespacing}{\linespacing}%
  {\bfseries\centering}}
\def\subsection{\@startsection{subsection}{2}%
   \z@{1.25\linespacing\@plus.7\linespacing}{.5\linespacing}%
   {\normalfont\bfseries}}
\def\subsectionsinline{\def\subsection{\@startsection{subsection}{2}%
  \z@{1\linespacing\@plus.7\linespacing}{-.5em}%
  {\normalfont\bfseries}}}

\newenvironment{qedequation}{%
   \pushQED{\qed}%
   \incr@eqnum
   \mathdisplay@push
   \st@rredfalse \global\@eqnswtrue
   \mathdisplay{equation}%
}{%
   \endmathdisplay{equation}%
   \mathdisplay@pop
   \ignorespacesafterend
   \popQED\@endpefalse
}

\makeatother

\theoremstyle{definition}
\newtheorem{definition}[equation]{Definition}

\newtheorem*{definition*}{Definition}
\newtheorem*{example*}{Example}
\newtheorem*{problem*}{\color{blue}Problem}
\newtheorem*{exercise*}{Exercise}
\newtheorem*{question*}{\color{blue}Question}
\newtheorem*{project*}{\color{blue}Project}
\newtheorem*{construction*}{Construction}

\theoremstyle{remark}

\newtheorem{remark}[equation]{Remark}

\newtheorem*{note*}{Note}
\newtheorem*{notation*}{Notation}
\newtheorem*{remark*}{Remark}
\newtheorem*{data*}{Data}

\theoremstyle{plain}
\newtheorem{theorem}[equation]{Theorem}
\newtheorem{corollary}[equation]{Corollary}
\newtheorem{lemma}[equation]{Lemma}
\newtheorem{proposition}[equation]{Proposition}
\newtheorem{conjecture}[equation]{Conjecture}

\newtheorem*{theorem*}{Theorem}
\newtheorem*{corollary*}{Corollary}
\newtheorem*{lemma*}{Lemma}
\newtheorem*{proposition*}{Proposition}
\newtheorem*{conjecture*}{Conjecture}
\newtheorem*{claim*}{Claim}
\newtheorem*{proposal*}{Proposal}
\newtheorem*{conclusion*}{Conclusion}
\newtheorem*{hypothesis*}{Hypothesis}
\newtheorem*{assumption*}{Assumption}

\newenvironment{proof*}[1][\proofname]{
  \begin{proof}[#1]}{  \renewcommand\qedsymbol{\relax}
\end{proof}}

\numberwithin{equation}{section}

\definecolor{refkey}{rgb}{0,.6,.4}

\renewcommand{\:}{\colon}
\renewcommand{\AA}{{\mathbb A}}
\newcommand{\Ahat}{{\hat A}}

\newcommand{\CC}{{\mathbb C}}
\newcommand{\CP}{{\mathbb C\mathbb P}}

\newcommand{\EE}{\mathbb E}

\DeclareMathOperator{\End}{End}

\newcommand{\HH}{{\mathbb H}}
\DeclareMathOperator{\Hom}{Hom}
\DeclareMathOperator{\id}{id}

\newcommand{\PP}{{\mathbb P}}
\DeclareMathOperator{\pt}{pt}

\newcommand{\RP}{{\mathbb R\mathbb P}}
\newcommand{\RR}{{\mathbb R}}
\newcommand{\TT}{\mathbb T}
\DeclareMathOperator{\Spin}{Spin}

\DeclareMathOperator{\tr}{tr}

\newcommand{\ZZ}{{\mathbb Z}}
\DeclareMathOperator{\ch}{ch}
\newcommand{\chiup}{\raise.5ex\hbox{$\chi$}}
\newcommand{\cir}{S^1}

\DeclareMathOperator{\ind}{ind}
\newcommand{\inv}{^{-1}}
\DeclareRobustCommand{\mstrut}{^{\vphantom{1*\prime y\vee M}}}
\newcommand{\mlstrut}{_{\vphantom{1*\prime y}}}

\newcommand{\temsquare}{\raise3.5pt\hbox{\boxed{ }}}

\newcommand{\zmod}[1]{\ZZ/#1\ZZ}

\newcommand{\zt}{\zmod2}

\DeclareMathOperator{\SO}{SO}

\usepackage[all,2cell,cmtip]{xy}\renewcommand{\cir}{\ensuremath{S^1}}
\usepackage{graphicx}
\usepackage[colorlinks,backref=page,citecolor=refkey]{hyperref}

\definecolor{refkey}{rgb}{0,.8,.2}\definecolor{labelkey}{rgb}{1,0,0} 

\DeclareMathOperator{\OO}{O}

\DeclareMathOperator{\Cliff}{Cliff}
\DeclareMathOperator{\Det}{Det}

\DeclareMathOperator{\Pfaff}{Pfaff}
\DeclareMathOperator{\Pin}{Pin}

\DeclareMathOperator{\ext}{Ext}

\DeclareMathOperator{\inde}{index}
\DeclareMathOperator{\pfaff}{pfaff}
\DeclareMathOperator{\pinc}{pin^{c}}
\DeclareMathOperator{\pin}{pin}
\DeclareMathOperator{\sign}{sign}
\DeclareMathOperator{\spec}{spec}
\DeclareMathOperator{\spinc}{spin^{c}}
\DeclareMathOperator{\spin}{Spin}
\DeclareMathOperator{\sq}{Sq}
\DeclareMathOperator{\thom}{Thom}
\newcommand{\BPinp}{B\!\Pin^+}
\newcommand{\BSpin}{B\!\Spin}
\newcommand{\CZ}{\CC/\ZZ}
\newcommand{\Clm}[1]{\Cliff_{-#1}}
\newcommand{\Clp}[1]{\Cliff_{+#1}}
\newcommand{\Cx}{\CC^\times }
 
\newcommand{\Enl}{E^0_\lambda }
\newcommand{\Epl}[1]{{\sE'_{\lambda }}^{\!#1}}
\newcommand{\Eppm}[1]{{\sE''_{\mu ^2}}^{\!\!\!#1}}
\newcommand{\Epp}[1]{{\sE''}^#1}
\newcommand{\Ep}[1]{{\sE'}^#1}

\newcommand{\FSO}{\mathcal{B}_{SO}}
\newcommand{\FSpinc}{\mathcal{B}_{\textnormal{Spin}^c}}

\newcommand{\HP}{\HH\PP}

\newcommand{\ICx}{I\Cx}

\newcommand{\Lchar}{L_{\textnormal{char}}}
\newcommand{\Lt}{L\otimes \zt}
\newcommand{\MM}{\mathbb{M}}

\newcommand{\MSpin}{M\!\Spin}
\newcommand{\MTPp}{MT\!\Pp}
\newcommand{\MTP}{MT\!\Pin}

\newcommand{\Pp}{\Pin^{+}}

\newcommand{\RS}{RS}
\newcommand{\RZ}{\RR/\ZZ}
\newcommand{\R}{\RR}

\newcommand{\SC}{\SS_{\CC}}

\newcommand{\Sq}[1]{Sq^{#1}}

\newcommand{\Zlt}{\ZZ _2}

\newcommand{\Z}{\mathbb{Z}}
\newcommand{\aS}{\alpha \mstrut _{\SS}}
\newcommand{\aW}{\tau \mstrut _W}
\newcommand{\ac}{\alpha \mstrut _{C}}
\newcommand{\af}{\alpha \mstrut _{RS}}
\newcommand{\asv}{\alpha \mstrut _{\SS\otimes V}}

\newcommand{\bone}{\bold{1}}

\newcommand{\bp}{\bar{p}}
\newcommand{\bw}{\bar\omega }
\newcommand{\cH}{\widecheck{H}}
\newcommand{\cO}{\widecheck{\Omega}}

\newcommand{\cbar}{\bar{c}}
\newcommand{\chat}{\hat{c}}

\newcommand{\cp}{\widecheck{p}}

\newcommand{\hG}{\Gamma\mstrut _{\!\CC}}

\newcommand{\hM}{\widehat{M}}
\newcommand{\hO}{\widehat{\Omega }}

\newcommand{\hW}{\widehat{W}}
\newcommand{\hX}{\widehat{X}}
\newcommand{\hac}{\hat\alpha \mstrut _{C}}
\newcommand{\haf}{\hat\alpha \mstrut _{RS}}

\newcommand{\kW}{\kappa \mstrut _W}
\newcommand{\mcable}{\mc}
\newcommand{\mc}{\mathfrak{m}_c}
\newcommand{\oB}{\omega \mstrut _B}

\newcommand{\phat}{\hat{p}}

\newcommand{\pinm}{\pin^{-}}
\newcommand{\pinp}{\pin^{+}}

\newcommand{\pmo}{\{\pm1\}}

\newcommand{\sB}{{M}\kern-.1em\mathfrak m_{c}}
\newcommand{\sBs}{\MSpin\langle\beta w_4\rangle}

\newcommand{\sE}{\mathscr{E}}

\newcommand{\sH}{\mathscr{H}}

\newcommand{\slot}{\,-\,}

\newcommand{\sqmo}{\sqrt{-1}}

\newcommand{\tKO}{\widetilde{KO}}
\newcommand{\tR}{\widetilde{\RR}}

\newcommand{\tZ}{\widetilde{\ZZ}}
\newcommand{\tar}{\Sigma^{12}\ICx}
\newcommand{\tb}{\tilde\beta }
\newcommand{\tc}{\tilde{c}}

\newcommand{\torsn}{\textnormal{torsion}}
\newcommand{\tors}{\mstrut _{\textnormal{tor}}}

\newcommand{\ttar}{\Sigma^{11}\ICx}
\newcommand{\tx}{\tilde{x}}
\newcommand{\und}[1]{\underline{#1}}

\newcommand{\xbar}{\bar{x}}
\newcommand{\xhat}{\hat{x}}
\newcommand{\ybar}{\bar{y}}

\renewcommand{\Re}{\textnormal{Re}}
\renewcommand{\SS}{\mathbb{S}}

\newtheorem{lem}[equation]{Lemma}
\newtheorem{prop}[equation]{Proposition}
\newtheorem*{thm*}{Theorem}
\newtheorem*{cor*}{Corollary}
\newtheorem*{lem*}{Lemma}
\newtheorem*{prop*}{Proposition}

\ifx\undefined\theoremstyle
	\newtheorem{definition}[equation]{Definition}
	\newenvironment{defin}{\begin{definition}\rm}{\end{definition}}
	\newtheorem{conj}[equation]{Conjecture}
	\newtheorem*{rem*}{Remark}
	\newtheorem{rems}[equation]{Remarks}

\else
	\theoremstyle{definition}

	\newtheorem{defin}[equation]{Definition}
	\newtheorem{conj}[equation]{Conjecture}

	\theoremstyle{remark}

	\newtheorem*{rem*}{Remark}

\ifx\undefined\prem	\newtheorem{prem}{\bf Working Remark} \fi

\fi
\newtheorem{eg}[equation]{Example}

\ifx\undefined\pf
\newenvironment{pf}{\bigskip{\em Proof:\/}}{\qed\medskip}
\newenvironment{pf*}[1]{\bigskip{\em #1:\/}}{\qed\medskip}
\fi

\ifx\undefined\numberwithin
\def\numberwithin#1#2{\makeatletter\@ifundefined{c@#1}{\@nocnterrr}{%
  \@ifundefined{c@#2}{\@nocnterr}{%
  \@addtoreset{#1}{#2}%
  \toks@\expandafter\expandafter\expandafter{\csname the#1\endcsname}%
  \expandafter\xdef\csname the#1\endcsname
    {\expandafter\noexpand\csname the#2\endcsname
     .\the\toks@}}}\makeatother}\fi

\newcommand{\Q}{{\mathbb Q}}
\newcommand{\C}{\mathbb C}

\ifx\undefined\eqref\newcommand{\eqref}[1]{\rm (\ref{#1})}\fi

\ifx\undefined\slot
\newcommand{\slot}{\,-\,}
\fi

\DeclareMathOperator{\bstring}{BString}
\DeclareMathOperator{\pinplus}{Pin^{+}}

\newcommand{\mspin}{\MSpin}
\newcommand{\bspin}{\BSpin}
\newcommand{\bpinplus}{B\!\Pin^+}
\newcommand{\mtPinPlus}{\MTPp}
\DeclareMathOperator{\String}{String}

\newcommand{\pinPlus}{\pin^{+}}
\newcommand{\hp}{\HP}
\newcommand{\emcee}{\mathfrak m_{c}}
\newcommand{\twistedB}{{B}\kern-.04em\mathfrak m_{c}}
\newcommand{\twistedThom}{{M}\kern-.1em\mathfrak m_{c}}

\newcommand{\LL}{[\mkern-2mu[}
\newcommand{\RRR}{]\mkern-2mu]}

\newcommand{\ahat}{\hat{A}}

\begin{document}

\abovedisplayskip18pt plus4.5pt minus9pt
\belowdisplayskip \abovedisplayskip
\abovedisplayshortskip0pt plus4.5pt
\belowdisplayshortskip10.5pt plus4.5pt minus6pt
\baselineskip=15 truept
\marginparwidth=55pt

\makeatletter
\renewcommand{\tocsection}[3]{%
  \indentlabel{\@ifempty{#2}{\hskip1.5em}{\ignorespaces#1 #2.\;\;}}#3}
\renewcommand{\tocsubsection}[3]{%
  \indentlabel{\@ifempty{#2}{\hskip 2.5em}{\hskip 2.5em\ignorespaces#1%
    #2.\;\;}}#3} 
\renewcommand\thepart{\Roman{part}} 
\def\l@part{\@tocline{-1}{12pt plus2pt}{0pt}{}{\bfseries}}
\renewcommand{\part}{\@startsection{part}{0}%
  \z@{3\linespacing\@plus\linespacing}{\linespacing}%
  {\normalfont\bfseries\centering}}

\makeatother

\setcounter{tocdepth}{2}

 \title[M-Theory Anomaly Cancellation]{Consistency of M-Theory on Nonorientable Manifolds} 
 \author[D. S. Freed]{Daniel S.~Freed}
 \address{Department of Mathematics \\ University of Texas \\ Austin, TX
78712} 
 \email{dafr@math.utexas.edu}
 \author[M. J. Hopkins]{Michael J.~Hopkins}
 \address{Department of Mathematics \\ Harvard University \\ Cambridge, MA
02138} 
 \email{mjh@math.harvard.edu}
 \thanks{This material is based upon work supported by the National Science
Foundation under Grant Numbers DMS-1158983, DMS-1160461, DMS-1510417, and
DMS-1611957.  We thank the Aspen Center for Physics for hospitality while
most of this work was carried out.}
 \date{\today}
 \dedicatory{To Sir Michael}
 \begin{abstract} 
 We prove that there is no parity anomaly in M-theory in the low-energy field
theory approximation.  Our approach is computational.  We determine
generators for the 12-dimensional bordism group of pin manifolds with a
$w_1$-twisted integer lift of~$w_4$; these are the manifolds on which
Wick-rotated M-theory exists.  The anomaly cancellation comes down to
computing a specific $\eta $-invariant and cubic form on these manifolds.  Of
interest beyond this specific problem are our expositions of: computational
techniques for $\eta $-invariants, the algebraic theory of cubic forms, Adams
spectral sequence techniques, and anomalies for spinor fields and
Rarita-Schwinger fields.
 \end{abstract}
\maketitle

{\small
\def\reftext{References}
\renewcommand{\tocsection}[3]{%
  \begingroup 
   \def\tmp{#3}%
   \ifx\tmp\reftext
   #3%
  \else\indentlabel{\ignorespaces#1 #2.\;\;}#3%
  \fi\endgroup}
\tableofcontents}

   \section{Introduction}\label{sec:1}

Time-reversal symmetry is a topic of renewed interest, in part because of its
prevalence in condensed matter models.  Under Wick rotation time-reversal is
connected to reflections in Euclidean space, and time-reversal symmetric
theories may be formulated on unoriented manifolds.  The obstruction to doing
so is often termed a ``parity anomaly'', though `parity' is not synonymous
with `time-reversal symmetry'.\footnote{`Parity' typically refers to a
spatial reflection through a point in Minkowski spacetime, relative to a
splitting into time cross space.  As this is orientation-preserving in even
space dimensions, more relevant is reflection in a timelike hyperplane, which
is always orientation-reversing.  A time-reversal is reflection in a
spacelike hyperplane.}  Witten~\cite{W2} recently showed that there is no
anomaly for this symmetry on an M2-brane in M-theory.  He suggested that we
investigate the analogous issue in the bulk on suitable 11-manifolds.  We do
so here and prove that there is no time-reversal anomaly in M-theory.
 
We work in the low-energy field theory approximation to M-theory, which is
classical 11-dimensional supergravity with a gravitational correction
term~\cite{DLM,VW}.  The theory includes a fermionic field, and so
$X$~carries a pin structure---the appropriate choice is a $\pinp$ structure,
as opposed to a $\pinm$ structure---on the \emph{tangent} bundle.  The
$C$-field in M-theory, which is odd under time-reversal symmetry, induces an
additional topological structure on~$X$: a $w_1$-twisted integer lift of the
fourth Stiefel-Whitney class~$w_4(X)$; see~\cite[\S2.3]{W3}.  A
$\pinp$-manifold with a $w_1$-twisted integer lift of~$w_4$ is called an
\emph{$\mc$-manifold}.  There are two sources of anomalies.  The first is the
standard fermion anomaly, though there are subtleties: the fermion field is a
Rarita-Schwinger field, rather than a spinor field, and the background is a
pin manifold, rather than a spin manifold.  The second anomaly is
nonstandard, due to the cubic form for the $C$-field.  In the spin case
Witten~\cite[\S4]{W3} represents the $C$-field as a connection on a principal
$E_8$-bundle, and he uses this to prove that these two anomalies cancel.  In
the pin case this argument is not available, so we resort to a computational
approach.  Each anomaly is encoded in an invertible unitary topological
12-dimensional field theory, hence is determined by its partition function.
Furthermore, the partition function is a bordism invariant, so it suffices to
check that the partition functions of the two theories agree on a set of
generators for the appropriate bordism group.  We use the Adams spectral
sequence, together with computer assistance and geometric arguments, to
compute a set of generators for the relevant bordism group.  We deploy a mix
of topological and geometric techniques to compute the partition functions on
these generators, and so prove anomaly cancellation.
 
To define M-theory we must not only prove that anomalies cancel, but provide
data which performs the anomaly cancellation.  In the spin case, ignoring
time-reversal symmetry, this ``setting of the quantum integrand'' can be
achieved using Witten's $E_8$-bundle technique~\cite{FM}.  We do not know a
canonical setting in the pin case, and indeed isomorphism classes of settings
form a torsor over isomorphism classes of 11-dimensional invertible field
theories on the same class of manifolds.  The latter group is isomorphic
to~$\zt$, at least conjecturally, as we explain in~\S\ref{sec:8}.  Since
this group is nonzero, the problem of setting the quantum integrand remains
open.

Now we give a more detailed summary.  We begin in the
expository~\S\ref{sec:2} by reviewing the Wick-rotated setting for M-theory
as a theory on a certain geometric bordism category.  We recall that the
anomaly of an 11-dimensional theory is an invertible 12-dimensional
theory~$\alpha $, and that invertible topological theories may be represented
as maps of spectra in stable homotopy theory.  The aim of this paper is to
prove Theorem~\ref{thm:1}: the tensor product~$\af\otimes \ac$ of the anomaly
theories arising from the Rarita-Schwinger and $C$-fields of M-theory is
trivializable.  In~\S\ref{sec:3} we define~$\af$.  As with the anomaly of any
fermionic field its partition function is the exponential of an
Atiyah-Patodi-Singer $\eta $-invariant.  We elucidate some aspects of the
general theory in Appendix~\ref{sec:7}, and in~\S\ref{subsec:3.2} we give a
general formula for the anomaly theory of a Rarita-Schwinger field;
see~\eqref{eq:34}.  In our situation the anomaly partition function is
independent of the Riemannian metric, so is a topological invariant.  It
turns out to be~$\pm1$ on $\emcee$-manifolds, though at this stage the only
apparent statement is that it is a root of unity.  Indeed, on a general
$\pinp$ manifold it does not necessarily have order~2.  We develop formulas
to compute it, following work of Donnelly, Stolz, and Zhang.  Of particular
interest is a topological formula which, as far as we know, has only an
analytic proof in the literature~\cite{Z}.\footnote{We thank Jonathan
Campbell for pointing us to Zhang's paper.}  The partition function of the
anomaly theory~$\ac$ is an inhomogeneous cubic polynomial in the $C$-field.
It too is topological and by definition is equal to~$\pm1$.  In~\S\ref{sec:4}
we develop an algebraic theory of the cubic form, imitating the standard
algebraic theory of quadratic forms, and then define~$\ac$.  We also review
Witten's proof that $\af\otimes \ac$ is trivializable when restricted to spin
manifolds.  Section~\ref{sec:5} is a geometric interlude to review some basic
spin and $\pinp$~manifolds and their topological invariants.  We also
introduce more complicated manifolds used as representative elements of
bordism groups.  Our main computational result, whose proof we sketch
in~\S\ref{sec:9}, is stated as Theorem~\ref{thm:31}.  We specify generators
of the relevant 12-dimensional bordism group, which we represent by specific
12-dimensional $\mc$-manifolds.  For each of these we compute that the
partition function of~$\af\otimes \ac$ vanishes, which suffices to
demonstrate the anomaly cancellation.  We employ a potpourri of techniques to
make the computations.  The aforementioned ambiguity in the definition of
M-theory is discussed in Section~\ref{sec:8}.  Section~\ref{sec:9} contains a
computation of the low dimensional bordism groups of $\mc$-manifolds.  In
particular, we provide a proof of Theorem~\ref{thm:31}.  For the computation
we need a set of generators of 12-dimensional spin bordism (localized at~2),
which we produce in Appendix~\ref{sec:11}, based on the work of
Anderson-Brown-Peterson.  Appendix~\ref{sec:10} details the mod~2 cohomology
of the Thom spectrum of $\mc$-manifolds, a key input into the Adams spectral
sequence computation.  A more detailed computer-free version of these
computations will appear in~\cite{GH}.

Aspects of this paper have interest beyond our proof that M-theory is
time-reversal invariant.  This includes the algebraic theory of cubic forms
in~\S\ref{sec:4}; our techniques to compute $\eta $-invariants of pin
manifolds; the Adams spectral sequence techniques in~\S\ref{sec:9} and the
cohomology computations in Appendix~\ref{sec:10}; the discussion of spinor
field anomalies in Appendix~\ref{sec:7}; and the interplay between invertible
unitary topological field theories and stable homotopy theory, which is
developed and plays a key role in an application to condensed matter physics
in~\cite{FH}.

The authors take this opportunity to express our deep sense of gratitude and
indebtedness to Michael Atiyah for his mentoring, encouragement, and support.
Michael's enthusiasm for mathematics and for its interaction with physics has
long been an inspiration.  We appreciate his unfailing sense of what
constitutes an enlightening and ``correct'' proof, and we join him in
lamenting the lack of such a proof for this anomaly cancellation.

We thank Rob Bruner, Jonathan Campbell, Stephan Stolz, and Edward Witten for
useful conversations and correspondence.  The anonymous referee gave an
earlier version a very close reading, and the resulting comments greatly
improved the manuscript, for which we are thankful.

  \section{Time-reversal, anomalies, and bordism}\label{sec:2}

A relativistic quantum field theory on $n$-dimensional Minkowski
spacetime~$\MM^n$ has a symmetry group~$\sH_{1,n-1}$, equipped with a
homomorphism to the group of isometries of~$\MM^n$.  (See~\cite[\S2]{FH} for
an account of symmetry groups in quantum field theory.)  Divide by
translations and Wick rotate to Euclidean signature to obtain a compact Lie
group~$H_n$ of vector symmetries, equipped with a homomorphism $\rho
_n\:H_n\to O_n$ whose image is (i)~$SO_n$ in the absence of time-reversal
symmetry, or (ii)~$O_n$ if the theory has time-reversal symmetry.
Eleven-dimensional M-theory has both time-reversal symmetry and fermionic
fields, and no additional global symmetries,\footnote{One could regard the
$C$-field as the background field for a higher symmetry, but as the primary
objects of interest are the background fields we do not pursue this point of
view.} so the Wick-rotated symmetry group is one of the two Pin groups.
Because time-reversal squares in Minkowski spacetime to~$(-1)^F$, the
appropriate group is $H\mstrut _{11}=\Pp_{11}$; see~\cite[Appendix~A]{FH}.
We consider M-theory on curved compact 11-dimensional Riemannian
manifolds~$X$, and so we require that $X$~have a \emph{tangential} $\pinp$
structure.\footnote{Equivalently, the stable normal bundle of~$X$ has a
$\pinm$ structure.}  There is an additional topological structure, first
identified in~\cite[\S2.3]{W3}.  The $C$-field is an abelian gauge field,
thus obeys a Dirac quantization condition.  The correct condition is that the
de Rham cohomology class of its field strength, a closed 4-form twisted by
the orientation bundle,\footnote{On the orientation double cover of~$X$ it
lifts to a closed 4-form which is odd under the deck transformation.}
refines to a $w_1$-twisted integer cohomology class $c\in H^4(X;\tZ)$ whose
mod~2 reduction is the fourth Stiefel-Whitney class~$w_4(X)$.  Here $\tZ$~is
the local coefficient system induced from the orientation double cover
of~$X$.  This motivates the following.

  \begin{definition}[]\label{thm:19}
 Let $M$~be a $\pinp$ manifold.  An \emph{$\mc$ structure}\footnote{The name
is taken from~\cite[\S4.3]{W2}, where it is introduced by analogy with a
$\spinc$ structure.} on~$M$ is a $w_1$-twisted integer lift of~$w_4(M)$.  We say
$M$~is an \emph{$\mc$-manifold} if $M$~is equipped with an $\mc$ structure.
  \end{definition}

\noindent
 A necessary and sufficient condition to be $\mcable$ is $\tb w_4(M)=0$,
where $\tb$~is the Bockstein map into $w_1$-twisted integral cohomology;
see~\eqref{eq:70}.  The Wick rotation of M-theory is defined on a geometric
bordism category of $\mc$-manifolds.
 
Once an $n$-dimensional field theory is formulated on compact Riemannian
manifolds, then there is the possibility of an anomaly: the partition
function may not be well-defined as a complex number, but rather may be an
element of a complex line.  These complex lines depend \emph{locally} on the
Riemannian manifold, which is expressed by saying that they are the quantum
state spaces of a field theory~$\alpha $.  The theory~$\alpha $ is called a
\emph{gravitational anomaly}.  In addition to the coupling to a gravitational
background, if the kernel of $\rho _n\:H_n\to O_n$ is nontrivial then there
is also a coupling to a background gauge field,\footnote{which may be twisted
by the tangent bundle, as in a $\spin^c$~structure, for example.} in which
case we have a mixed gravitational and gauge anomaly.  In most examples the
anomaly theory~$\alpha $ extends to an $(n+1)$-dimensional theory which has a
partition function on closed $(n+1)$-manifolds.  That is so in this paper.
The 11-dimensional M-theory is not rigorously defined, but nonetheless we do
define the 12-dimensional anomaly theory that is our main focus.  Anomalies
are very special among field theories: they are \emph{invertible}.  Recall
that field theories have a composition law of tensor product, and there is a
trivial theory~$\bone$ which is an identity for this composition law.  So a
field theory~$\alpha $ is invertible if there exists a theory~$\beta $ such
that $\alpha \otimes \beta $ is isomorphic to~$\bone$.  An invertible field
theory has nonzero partition functions, one-dimensional state spaces, etc.
We refer to~\cite{F1} and the references therein for exposition on this point
of view about anomalies.
 
Recall that M-theory has two bosonic fields---a metric and $C$-field---and a
single fermionic field---the Rarita-Schwinger field~$\psi $.  To analyze
anomalies we work in the effective theory after integrating out~$\psi $; the
metric and $C$-field are treated as background fields.  One source of
anomalies is the fermionic integration of~$\psi $, which we review
in~\S\ref{sec:3}.  Let~$\af$ denote that 12-dimensional anomaly theory.  The
other source of anomalies is the ``Chern-Simons coupling'' of the $C$-field,
which is an inhomogeneous cubic form we review in~\S\ref{sec:4}.  Let~$\ac$
denote that 12-dimensional anomaly theory.  Our main result is the
cancellation of these anomalies.

  \begin{theorem}[]\label{thm:1}
 The total anomaly theory $\af\otimes \ac$ is trivializable. 
  \end{theorem}

\noindent
 That is, $\af\otimes \ac\cong \bone$.  This implies that M-theory should
exist as an ``absolute'' quantum field theory whose partition functions are
complex numbers (not merely elements of an abstract complex line), whose
state spaces are vector spaces (not merely well-defined as projective
spaces), etc.  In other words, M-theory is anomaly-free.  As explained
in~\cite[\S1]{W2} this is a strong form of the vanishing of the ``parity
anomaly''.
 
An important feature is that both~$\af$ and~$\ac$ are \emph{topological}
field theories.  That means they are each independent of the metric and
$\ac$~only depends on the $C$-field through its topology.\footnote{A
$C$-field on~$X$ represents a class in the twisted differential cohomology
group $\cH^4(X;\tZ)$ (\S\ref{subsec:4.3}); the statement is that the anomaly
theory~$\ac$ only depends on the representative of its image under
$\cH^4(X;\tZ)\to H^4(X;\tZ)$.}  Furthermore, as already stated these theories
are \emph{invertible}.  Finally, due to their physical
origins\footnote{Invertible topological theories (not necessarily unitary)
have domain Madsen-Tillmann spectra; see~\cite{FHT}, \cite{S-P}.} these
theories are \emph{unitary}, or equivalently in the Wick-rotated version they
satisfy reflection positivity.  The main theorem in~\cite{FH} asserts that,
assuming reasonable ans\"atze, reflection positive invertible topological
field theories live in the world of stable homotopy theory: they are spectrum
maps from a Thom spectrum to a universal target, the shifted Pontrjagin dual
to the sphere spectrum.  This result uses a strong form of locality---a
\emph{fully extended} field theory---and also a companion strong form of
reflection positivity for invertible topological theories.  Thus the anomaly
theories are maps
  \begin{equation}\label{eq:1}
     \af,\ac\: \sB\longrightarrow \tar.
  \end{equation}
Here $\sB$ is the Thom spectrum of $\mc$-manifolds: manifolds with a stable
tangential $\pinp$ structure and a $w_1$-twisted integer lift of~$w_4$.  We
construct~$\sB$ in~\S\ref{subsec:9.1}.  Also, $\ICx$~is the character dual to the
sphere spectrum, closely related to the Brown-Comenetz dual~\cite{BC}.  The
universal property which characterizes~$\ICx$ (see~\cite[\S5.3]{FH}) implies
that the group of homotopy classes of maps~\eqref{eq:1} is isomorphic to the
group~$\Hom(\pi _{12}\sB,\Cx)$ of characters of~$\pi _{12}\sB$.  In other
words, the maps~ \eqref{eq:1} are determined up to homotopy---and the
corresponding topological field theories up to isomorphism---by abelian group
homomorphisms
  \begin{equation}\label{eq:2}
     \haf,\hac\:\pi _{12}\sB\longrightarrow \Cx. 
  \end{equation}
These homomorphisms encode the partition functions of the respective anomaly
theories.  We prove Theorem~\ref{thm:1} by demonstrating that the product 
  \begin{equation}\label{eq:25}
      \haf\cdot \hac\:\pi _{12}\sB\longrightarrow \Cx
  \end{equation}
of partition functions is identically one.  Both $\haf$~ and $\hac$~take
values in the group $\TT\subset \Cx$ of unit norm complex numbers.  From its
definition~\eqref{eq:83}, the homomorphism $\hac$~takes values in
$\pmo\subset \Cx$, and so the field theory~$\ac$ has order two: its square is
isomorphic to the trivial theory.  It emerges from our computations that
$\af$~also has order two.\footnote{$\af$~is pulled back from an invertible
theory defined on~$\MTP^+$ which has order~$2^8$.}

Theorem~\ref{thm:1} asserts that the total anomaly is trivializable but does
not specify a trivialization.  (For further discussion, see~\cite{FM} where
the trivialization is called a ``setting of the quantum integrand''.)
Homotopy classes of trivializations form a torsor over the group of
invertible 11-dimensional reflection positive topological theories.  That is,
given one trivialization, and so in principle one realization of M-theory,
any other one differs by inserting a ``topological term'' in the
11-dimensional theory.  In~\S\ref{sec:8}, based on computations to appear
in~\cite{GH}, we discuss the following.   

  \begin{conjecture}[]\label{thm:2}
 The group of homotopy classes of spectrum maps $\sB\to\ttar$ is isomorphic
to~$(\zt)$.  The partition function of the nontrivial theory is the mod~2
index of the Dirac operator.
  \end{conjecture}

   \section{The Rarita-Schwinger operator and $\eta $-invariants}\label{sec:3}

The reader may want to refer to recent expositions of fermions and anomalies
in~\cite{W1,W2}.  We recall the relation between fermion integrals and
pfaffians in~\S\ref{subsec:3.1}.  In~\S\ref{subsec:3.2} we indicate the
modifications engendered by a Rarita-Schwinger field, as opposed to a spinor
field, and then specialize to M-theory and define the Rarita-Schwinger
anomaly theory~$\af$.  This all relies on material in Appendix~\ref{sec:7}.
In~\S\ref{subsec:3.4} we recall and prove some properties of and formulas for
$\eta $-invariants on $\pinp$ 12-manifolds that we use in our subsequent
computations (\S\ref{sec:6}).  This exponentiated $\eta $-invariant is
topological---independent of the Riemannian metric---and there is a
topological formula (Theorem~\ref{thm:37}) for its value.

  \subsection{Brief recollection of free fermionic path integrals}\label{subsec:3.1}

The material in this section also appears in~\cite[\S11]{F3} as part of a
broader discussion of anomalies.

Suppose $W$~is a finite dimensional complex vector space and 
  \begin{equation}\label{eq:26}
     B\:W\times W\longrightarrow \CC 
  \end{equation}
a skew-symmetric bilinear form.  We identify~ $B$ as a skew-symmetric map
$W\to W^*$, and so an element $\oB\in {\textstyle\bigwedge} ^2W^*$. The
natural integral on the exterior algebra is the linear map\footnote{The odd
vector space~$\Pi W$, the parity-reversal of~$W$, has as its ring of
functions the $\zt$-graded exterior algebra~${\textstyle\bigwedge} ^{\bullet
}W^*$.  The fermionic integration~\eqref{eq:27} is purely algebraic---there
is no measure---and it is defined on functions rather than forms or
densities.}
  \begin{equation}\label{eq:27}
     \int_{\Pi W}\:{\textstyle\bigwedge} ^{\bullet }W^*\longrightarrow \Det
     W^* 
  \end{equation}
which projects a form of mixed degree to its highest degree component in
$\Det W^*={\textstyle\bigwedge} ^{\textnormal{max}}W^*$.  If $\dim W=2m$ is
even, then 
  \begin{equation}\label{eq:28}
     \int_{\Pi W}e^{\oB}=\frac{\omega _B^m}{m!} = \pfaff B\in \Det W^* 
  \end{equation}
is the \emph{pfaffian} of~$B$; if $\dim W$~is odd, then the integral
vanishes.  It is natural to regard~$\Det W^*$ as $\zt$-graded by the parity
of~$\dim W$, which is equal to the parity of the dimension of the null space
~$\ker B$.  There is an infinite dimensional version of the pfaffian for
$W$~a Hilbert space and $B$~a \emph{Fredholm form}: $B$~is Fredholm if $\ker
B$~ is a closed, finite dimensional subspace.  Then the $\zt$-graded line
$\Pfaff B$ depends on~$B$ and forms a nontrivial complex line bundle over the
space\footnote{That space has two components distinguished by the parity
of~$\dim \ker B$, the mod~2 index.  Over each component the Pfaffian line
bundle represents a generator of~$H^2(-;\ZZ)$.  The pfaffian section vanishes
if and only if $\ker B\neq 0$.} of Fredholm skew forms, and the pfaffian
elements
  \begin{equation}\label{eq:29}
     \pfaff B\in \Pfaff B 
  \end{equation}
form a section of the Pfaffian bundle.  See~\cite{Q}, \cite[Appendix~B]{Seg}
for the case of the Fredholm determinant (Remark~\ref{thm:5} below).

  \begin{remark}[]\label{thm:4}
 There are real and quaternionic refinements; the latter applies to M-theory
on spin manifolds~\cite[\S1]{FM}. 
  \end{remark}

  \begin{remark}[]\label{thm:5}
 Suppose $T\:U\to V$ is a linear map between complex vector spaces.  Set
$W=V^*\oplus U$ and 
  \begin{equation}\label{eq:30}
     B\bigl((v_1^*,u\mstrut _1),(v_2^*,u\mstrut _2) \bigr) = \langle
     v^*_1,Tu\mstrut _2 \rangle - \langle v^*_2,Tu\mstrut _1 \rangle . 
  \end{equation}
Then the Pfaffian line of~$B$ is canonically isomorphic to the determinant
line $\Det V\otimes (\Det U)^*$ of~$T$, and under that isomorphism $\pfaff
B=\det T$ as an element of $\Det(U)^*\otimes \Det(V)$.
  \end{remark}

A spinor field in an $n$-dimensional relativistic field theory on Minkowski
spacetime (without time-reversal symmetry) is specified by a real spinor
representation~$\SS$ of $\Spin_{1,n-1}$ together with a symmetric nonnegative
$\Spin_{1,n-1}$-invariant bilinear form $\Gamma \:\SS\times
\SS\to\RR^{1,n-1}$; see~\S\ref{subsec:7.1} for details.  The
complexification~$\SC$ is a representation of the compact spin
group~$\Spin_n$.  On a closed Riemannian spin $n$-manifold~$X$ there is an
associated complex vector bundle whose sections are spinor fields~$\psi $.
Define the complex skew-symmetric form 
  \begin{equation}\label{eq:31}
     B_X(\psi _1,\psi _2)=\int_{X}\hG (\psi _1,\nabla \psi _2)\,|dx|, 
  \end{equation}
where $\nabla $~is induced from the Levi-Civita covariant derivative,
$\hG$~is the complexification of~$\Gamma $, and $|dx|$~is the Riemannian
measure.  On appropriate function spaces $B_X$~is Fredholm.  The Feynman path
integral over~$\psi $ is the formal analog of~\eqref{eq:28}, and we define
the result to be~\eqref{eq:29}, the pfaffian element of the Pfaffian line.
In particular, the fermionic path integral is anomalous.
 
As explained in~\S\ref{sec:2} the Pfaffian line of the Dirac form~$B_X$ is
the quantum state space on~$X$ of an invertible $(n+1)$-dimensional field
theory~$\alpha $, called the \emph{anomaly theory}.  To define its partition
function we must use the data~$(\SS,\Gamma )$ to define a Dirac operator on a
Riemannian spin $(n+1)$-dimensional manifold~$W$.  That construction is
carried out in Appendix~\ref{sec:7}; the partition function~$\alpha (W)$ for
$W$~closed is the exponentiated $\eta $-invariant~\eqref{eq:12}.

  \subsection{The Rarita-Schwinger anomaly}\label{subsec:3.2}

The Rarita-Schwinger field occurs in theories of supergravity; it is the
super-partner to the metric.  In $n$~spacetime dimensions there is an
associated anomaly theory, which is an $(n+1)$-dimensional invertible field
theory, just as for a spinor field.  Here we explain the modifications to the
discussions in the previous section and~\S\ref{subsec:7.1} required to
specify the anomaly theory.  More information may be found in~\cite{FJ},
\cite[Appendix~A]{FM} and the references therein.
 
Suppose given data~$(\SS,\Gamma )$ for a spinor field in $n$-dimensional
Minkowski spacetime~$M^n$, as above.  Let $V$~be the standard $n$-dimensional
real representation of~$\Spin_{1,n-1}$.  The Rarita-Schwinger field is a
function $\chi \:M^n\to \SS\otimes V$.  (More precisely, we should view
$\SS\otimes V$ as an \emph{odd} super vector space.)  There is a
correspondence between free fields and particles, and under this
correspondence the Rarita-Schwinger field gives rise to four particles, as
recounted in~\cite[\S A.2]{FJ}: a single spin~3/2 particle, which is the
desired gravitino, as well as three spurious spin~1/2 particles.  Two of the
spin~1/2 particles are associated to spinor fields with values in~$\SS$ and
the third to a spinor field with values in~$\SS^*$.  Wick rotation of~$\chi $
proceeds by complexification, and one obtains a skew form~$B'_X$ analogous
to~\eqref{eq:31}, now built on sections of a spinor bundle tensored with the
tangent bundle.  To eliminate the extra spin~1/2 particles, we divide the
pfaffian of~$B'_X$ by the product of the pfaffians of the forms~$B\mstrut _X$
associated to the three Wick rotated spin~1/2 fields~\cite[\S A.5]{FJ}.

We now determine the associated anomaly theory, which is an invertible
$(n+1)$-dimensional field theory of Riemannian spin manifolds.  In
Appendix~\ref{sec:7} we define the anomaly theory~$\aS$ associated to spinor
data~$(\SS,\Gamma )$.  There is a variation which gives the anomaly of the
pfaffian of~$B'_X$.  Motivated by~\eqref{eq:8} define\footnote{The real
Clifford algebra~$\Cliff_{p,q}$ has $p$~generators squaring to~$+1$ and
$q$~generators squaring to~$-1$.  Set $\Clp p=\Cliff_{p,0}$ and $\Clm
q=\Cliff_{0,q}$.}
  \begin{equation}\label{eq:36}
     \EE' = \Clp2\,\otimes \,(\SS\oplus \SS^*)\,\otimes \,\RR^{n+1} .
  \end{equation}
Let $\Spin_{n+1}$~act as after~\eqref{eq:8} on $\EE=\Clp2\,\otimes
\,(\SS\oplus\SS^*)$ and tensor with the usual vector representation
on~$\RR^{n+1}$.  There is a commuting $\Clm3$~action, as after~\eqref{eq:8},
and the Dirac operator~\eqref{eq:10} and exponentiated $\eta
$-invariant~\eqref{eq:12} are defined.  Denote the resulting
$(n+1)$-dimensional theory as~$\asv$.  Since $\Spin_n$~acts reducibly
on~$\RR^{n+1}=\RR\oplus \RR^n$, the specialization to a product manifold
$\RR\times X$ gives the Dirac operator coupled to the tangent bundle plus an
extra copy of the Dirac operator on spinor fields.  Therefore, the anomaly
theory associated to the pfaffian of~$B'_X$ is
  \begin{equation}\label{eq:37}
     \asv\otimes \alpha _{\SS}^{\otimes (-1)}. 
  \end{equation}
Put~\eqref{eq:37} together with the anomalies of the spurious spin~1/2 fields
to obtain the total anomaly 
  \begin{equation}\label{eq:33}
     \alpha \mstrut _{\SS\otimes V}\otimes \alpha _{\SS}^{\otimes
     (-3)}\otimes \alpha _{\SS^*}^{\otimes (-1)}. 
  \end{equation}
Proposition~\ref{thm:6} implies that the product $\alpha \mstrut
_{\SS}\otimes \alpha \mstrut _{\SS^*}$ is trivializable.  Hence
\eqref{eq:33}~is isomorphic to
  \begin{equation}\label{eq:34}
     \alpha \mstrut _{\SS\otimes V}\otimes \alpha _{\SS}^{\otimes (-2)}. 
  \end{equation}
This is a general formula for the anomaly theory of a Rarita-Schwinger field
in $n$~dimensions built from the spinor data~$(\SS,\Gamma )$.

Now we specialize to~$n=11$ and M-theory on $\pinp$ manifolds.  The spinor
data~$(\SS,\Gamma )$ is specified in~\S\ref{subsec:7.2}.  Let $W$~be a closed
Riemannian $\pinp$ 12-manifold.  We compute the partition function~$\haf(W)$
of the total anomaly theory~\eqref{eq:34}.  The (s)pinor bundle on~$W$ is a
rank~32 quaternionic bundle, and the appropriate Dirac operator~\eqref{eq:10}
is a self-adjoint operator on its sections.  The partition
function~\eqref{eq:12} is, in this case, a ratio of exponentiated $\eta
$-invariants, which we write as 
  \begin{equation}\label{eq:38}
     \haf(W) = \exp\left( 2\pi i\;\frac{\eta (TW-2)}{4} \right) . 
  \end{equation}
Here $\eta (TW-2)$ is the difference of the $\eta $-invariant of the Dirac
operator coupled to the tangent bundle and twice the $\eta $-invariant of the
pure Dirac operator.

  \begin{proposition}[]\label{thm:8}
 The Rarita-Schwinger partition function~$\haf(W)$ is
\textnormal{(i)}~independent of the metric on~$W$, \textnormal{(ii)}~a
$\pinp$~bordism invariant, and \textnormal{(iii)}~a root of unity.
  \end{proposition}

  \begin{proof*}
 Similar assertions for even dimensional $\pinc$ manifolds are proved
in~\cite[\S1]{G}, so we will be brief.  The space of Riemannian metrics is
connected (better: contractible), so it suffices to compute the derivative
with respect to the metric.  That variation formula~\cite{APS2} is an
integral of a $w_1$-twisted 13-form over~$W$; the result is a 1-form on the space
of metrics.  The integrand is a local invariant of the geometry by a general
theory developed by Seeley~\cite{S}, and here it vanishes for parity reasons;
see~\cite[Lemma~1.5]{G}.  This proves~(i).  Now suppose $W$~is the boundary
of a compact $\pinp$ 13-manifold~$Z$.  Then \cite[Theorem~3.10]{APS1}
computes~\eqref{eq:38} as the exponential of the integral of the same $w_1$-twisted
13-form over~$Z$.  But as above, this 13-form vanishes identically, and so
(ii)~holds.  For~(iii) we need only use that the relevant bordism group is
finite, in fact~\cite{KT1}
  \begin{qedequation}\label{eq:39}
     \pi _{12} \MTPp \cong \zmod{2^8}\;\oplus \; \zmod{2^4}\;\oplus\;
     \zmod{2^2}. \qedhere
  \end{qedequation}
  \end{proof*}

  \begin{corollary}[]\label{thm:9}
 The Rarita-Schwinger partition function factors through a homomorphism 
  \begin{equation}\label{eq:40}
     \haf\:\pi _{12}\MTPp\longrightarrow \Cx. 
  \end{equation}
  \end{corollary}
 
As reviewed in~\S\ref{sec:2}, the homomorphism~\eqref{eq:40} determines an
invertible unitary topological field theory 
  \begin{equation}\label{eq:41}
     \af\:\MTPp\longrightarrow \Sigma ^{12}\ICx 
  \end{equation}
up to isomorphism.\footnote{The maps~\eqref{eq:1}, \eqref{eq:2} which
correspond to~\eqref{eq:40}, \eqref{eq:41} are lifted to the bordism
spectrum~$\sB$ of manifolds with a $\pinp$ structure \emph{and} a $w_1$-twisted
integer lift of~$w_4$.}  We stretch notation slightly and use the notation~
`$\af(W)$' for the partition function of a closed $\pinp$ 12-manifold~$W$.

  \subsection{Properties of the $\eta $-invariant}\label{subsec:3.4}

On a \emph{spin} manifold~$W$ the partition function~$\af(W)$ has a natural
logarithm defined using an integer index. 

  \begin{proposition}[]\label{thm:18}
Let $W$~be a closed spin 12-manifold.  Then $ \af(W) = (-1)^{\RS(W)}$, where 
  \begin{equation}\label{eq:78}
     \RS(W) = \frac 12\left\langle\Ahat(W)\ch(TW-2)\,,\,[W]\right\rangle. 
  \end{equation}
  \end{proposition}

  \begin{proof}
 As remarked after~\eqref{eq:24}, the $\pinp$ spinor data restricts to
spin spinor data, so on a spin manifold we apply the discussion
in~\S\ref{subsec:7.1}.  The crucial point, which holds in general for
$n$~odd, is that the action of the volume form $\omega =\delta ^1\delta
^2\gamma ^1\cdots\gamma ^{10}\subset \Clp{12}$ commutes with the action
of~$\Spin_{12}$ and so anticommutes with the Dirac operator~$D^0$.  This
implies that the spectrum of~$D^0$ is invariant under $\lambda
\mapsto-\lambda $.  Choose~$a$ in~\eqref{eq:11} to be negative and greater
than the first negative eigenvalue of~$D^0$.  Then $\eta _a(s)=\dim\ker D^0$
for all~$s$---the nonzero eigenvalues in the sum cancel---and so
from~\eqref{eq:12} we have
  \begin{equation}\label{eq:90}
     \af(W)=\exp(2\pi i\dim_{\CC}\ker D^0/4)=(-1)^{\dim_{\HH}\ker
     D^0}.   
  \end{equation}
The quaternionic dimension of the kernel is congruent mod~2 to the
quaternionic index of the chiral Dirac operator, which maps the
$+1$-eigenspace of~$\omega $ to its $-1$-eigenspace.  The Atiyah-Singer index
formula completes the proof.
  \end{proof}

  \begin{remark}[]\label{thm:36}
 The expansion of~\eqref{eq:78} in terms of Pontrjagin numbers of~$W$ is
  \begin{equation}\label{eq:142}
     RS(W) = \left\langle \frac{97p_1^3 - 788 p_1p_2 +
     3952p_3}{967680}\,,\,[W] \right\rangle .
  \end{equation}
  \end{remark}

  \begin{remark}[]\label{thm:20}
 In~\S\ref{subsec:4.2} we encounter a shift of~\eqref{eq:78} by an integer,
namely  
  \begin{equation}\label{eq:81}
     \RS'(W) = \frac 12\left\langle\Ahat(W)\ch(TW-4)\,,\,[W]\right\rangle. 
  \end{equation}
It has the same mod~2 reduction as~$RS(W)$.
  \end{remark}

For a real vector bundle $V\to W$ over a closed $\pinp$ 12-manifold set
  \begin{equation}\label{eq:94}
     \tau \mstrut _{W}(V) = \exp\left( 2\pi i\; \frac{\eta \mstrut
     _W(V)}{4}\right). 
  \end{equation}
The $\eta $-invariant~\eqref{eq:11} depends on a parameter~$a\in \RR$, as we
defined it, but the exponential~\eqref{eq:94} is independent of~$a$.
Comparing with~\eqref{eq:38} our notation is $\af(W)=\aW(TW-2)$.

  \begin{proposition}[]\label{thm:28}
 Let $V^0,V^1\to W$ be real vector bundles over a closed $\pinp$
12-manifold~$W$.  Then the ratio $\aW(V^0)/\aW(V^1)$ of exponentiated $\eta
$-invariants depends only on the class of the virtual bundle $[V^0]-[V^1]\in
KO^0(W)$.
  \end{proposition}

\noindent
 In particular, it does not depend on the choices of covariant derivative. 

  \begin{proof}
 The independence from the covariant derivative follows from the variation
formula, as in the proof of Proposition~\ref{thm:8}.  Then simply observe
that $\aW(V)$ is multiplicative: $\aW(V\oplus V')=\aW(V)\aW(V')$.
  \end{proof}

Proposition~\ref{thm:8} and Proposition~\ref{thm:28} suggest that there is a
$KO$-theory formula for~$\aW(V)$.  Indeed, such a formula was recently proved
by Zhang~\cite{Z}, based on an analytic theorem of Bismut-Zhang~\cite{BZ}
telling the behavior of $\eta $-invariants under immersions.  While the other
formulas and techniques in this section suffice for most of the computations
in~\S\ref{subsec:6.2}, we were only able to compute~$\af$
in~\S\ref{subsubsec:6.2.6} using this topological formula. 

  \begin{theorem}[Zhang]\label{thm:37}
 Let $V\to W$ be a real vector bundle over a closed $\pinp$ 12-manifold~$W$.
Let $L\to W$ be the orientation real line bundle, $H\to\RP^{20}$ the
tautological line bundle, and $\gamma \:W\to\RP^{20}$ a map such that $\gamma
^*H\cong L$.  Then 
  \begin{equation}\label{eq:146}
     \gamma _*\bigl([V] \bigr)=2^{11}\,\frac{\eta \mstrut _W(V)}{4}\,\bigl(1-[H]
     \bigr)\qquad \textnormal{in }\widetilde{KO}^0(\RP^{20}). 
  \end{equation}
  \end{theorem}

\noindent
 In this formula\footnote{We express~\eqref{eq:146} in a different, but
equivalent, form than~\cite{Z}, and we have used the $\pinp$ variant of his
$\pinm$ theorem (which he remarks holds in the $\pinp$ case).}  $[V]\in
KO^0(W)$ is the $KO$-class of $V\to W$; the map~$\gamma $ has a spin
structure induced from the $\pinp$ structures on $W$ and $\RP^{20}$ together
with a choice of isomorphism $\gamma ^*H\cong L$; and $\gamma _*$~is the
induced pushforward on $KO$-theory, after multiplication by the Bott class.
The group~$\widetilde{KO}^{0}(\RP^{20})$ is cyclic of order~$2^{11}$ with
generator $1-[H]$.

A $\pinp$ structure on a smooth manifold~$M$ has an \emph{opposite}, obtained
by tensoring with the orientation double cover. 

  \begin{proposition}[]\label{thm:29}
 Let $V\to W$ be a real vector bundle over a closed $\pinp$ 12-manifold~$W$,
and let $L\to W$ be the real line bundle associated to the orientation double
cover $\pi \:\hW\to W$.  Then
  \begin{equation}\label{eq:92}
     \aW(V\otimes L) = \aW (V)\inv . 
  \end{equation}
  \end{proposition}

  \begin{proof}
 Let $\sigma $~be the deck transformation of the double cover~$\pi $.  Then
$V$-valued spinor fields on~$W$ lift to $\sigma $-invariant $\pi ^*V$-valued
spinor fields on~$\hW$ and $V\otimes L$-valued spinor fields on~$W$ lift to
$\sigma $-anti-invariant $\pi ^*V$-valued spinor fields on~$\hW$.  Hence
$\eta \mstrut _W(V)+\eta \mstrut _W(V\otimes L)=\eta _{\hW}(\pi ^*V)$.  The
pullback of the $\pinp$ structure on~$W$ combines with the orientation of~
$\hW$ to produce a spin structure on~$\hW$, so $\tau _{\hW}(\pi ^*V)$ is
computed using the mod~2 reduction of 
  \begin{equation}\label{eq:93}
     \frac 12\left\langle\pi
     ^*\left[\Ahat(W)\ch(V)\right]\,,\,[\hW]\right\rangle, 
  \end{equation}
as in Proposition~\ref{thm:18}.  Since $\sigma $~is an orientation-reversing
involution, it follows that the integer~\eqref{eq:93} equals its negative,
hence vanishes.
  \end{proof}

  \begin{proposition}[]\label{thm:30}
 Suppose $W=W'\times W''$ is the product of a $\pinp$ 4-manifold~$W'$ and a
spin 8-manifold~$W''$.  Let $V'\to W'$ and $V''\to W''$ be real vector
bundles.  Then  
  \begin{equation}\label{eq:95}
     \aW(V'\otimes V'') = \tau _{W'}(V')^{\ind D_{W''}(V'')}, 
  \end{equation}
where the exponent is the index of the Dirac operator coupled to~$V''$. 
  \end{proposition}

  \begin{proof}
 This follows directly from the topological index formula~\eqref{eq:146}, but
there is a straightforward analytic proof which we outline here.  Use the
setup of Appendix~\ref{sec:7}.  Let $\Ep i, \Epp i$, $i=0,1$, denote the
spaces of spinor fields on~$W',W''$, and $D',D''$ the Dirac operators.  Let
$\omega =\gamma ^0\delta ^1\delta ^2$ denote the volume form of the commuting
$\Clm3$.  Then the space of spinor fields on~$W$ is\footnote{As we eventually
compute using an orthogonal decomposition into finite dimensional
eigenspaces, we do not worry about the topology in these tensor products.}
$\Ep0\otimes \Epp0\;\oplus \; \Ep1\otimes \Epp1$ and the Dirac operator
on~$W$ is $D^0_W=\omega D'\otimes \id\;+\;\omega \otimes D''$.  Write
spectral decompositions
  \begin{equation}\label{eq:135}
     \begin{aligned} \Ep0 &= \;\bigoplus\limits_{\lambda \in \spec \omega D'}
      \;\Epl0 \\ \Epp0 &= \bigoplus\limits_{\mu \in \spec (D'')^2}
      \Eppm0 \\ \Epp1 &= \bigoplus\limits_{\mu \in \spec (D'')^2}
      \Eppm1 \\ \end{aligned} 
  \end{equation}
If $\mu ^2\neq 0$ then $D^0_W$~acts on $\Epl0\otimes \Eppm0 \;\oplus
\; \Epl1\otimes \Eppm1$ with trace~0: we compute 
  \begin{equation}\label{eq:136}
     D^0_W(\psi '\otimes \psi ''\;\pm\;\mu \inv \omega \psi '\otimes D''\psi
     '') \;=\;(\lambda \mp\mu )(\psi '\otimes \psi ''\;\mp\;\mu \inv \omega
     \psi '\otimes D''\psi '') 
  \end{equation}
and let $\psi ',\psi ''$ run over orthonormal bases of~$\Epl0,\Eppm0$,
respectively.  Hence the only contributions\footnote{Choose~$a$
in~\eqref{eq:11} to be less than zero and greater than the first negative
eigenvalue of~$D_W$.} to the $\eta $-invariant of~$D^0_W$ come from
$\Ep0\otimes (\ker D'')^0$ and $\Ep1\otimes (\ker D'')^1$.  If $\psi '\in
\Epl0$ and ${\psi ''}^i\in (\ker D'')^i$, $i=0,1$, then since $\omega
D'=-D'\omega $ we compute
  \begin{equation}\label{eq:137}
     \begin{aligned} D^0_W(\psi '\otimes {\psi ''}^0) &= \phantom{-}\lambda\,
      (\psi '\otimes {\psi ''}^0) \\ D^0_W(\omega \psi '\otimes {\psi ''}^1) &=
      -\lambda\, (\omega \psi '\otimes {\psi ''}^1) \\ \end{aligned} 
  \end{equation}
and \eqref{eq:95}~quickly follows. 
  \end{proof}

The next result is inspired by techniques in~\cite{APS2}.  Suppose $W$~is a
closed $\pinp$ 12-manifold and $\pi \:\hW\to W$ its orientation double cover.
Let $\sigma \:\hW\to \hW$ be the canonical orientation-reversing free
involution.  If $P\to W$ is the principal $\Pp_{12}$-bundle of frames, then
$\sigma $~lifts canonically to an involution of $\pi ^*P\to\hW$ which
reverses the spin structure on~$\hW$.  Suppose that $Z$~is a compact spin
13-manifold with boundary~$\partial Z=\hW$ and $\sigma '$~an
orientation-reversing involution of~$Z$ which extends~$\sigma $ and is
equipped with a lift to a spin-reversing involution of the $\Pp_{13}$-bundle
of frames.  Let $F\subset Z$ denote the fixed point set of~$\sigma '$.  At an
isolated fixed point~ $f\in F$ the action of~$\sigma '$ on~$T_fZ$ is by~$-1$,
so its lift to the $\pinp$ frames acts by~$\pm\omega $, where $\omega =\gamma
^1\gamma ^2\cdots\gamma ^{13}$ is the volume form.  Let $i_f=\pm1$ denote the
sign.  If $V\to W$ is a real vector bundle, assume $\pi ^*V\to\hW$ extends
over~$Z$ and the involution~$\sigma '$ lifts, extending the lift of~$\sigma $
on the boundary.  Let $\tau _f$~denote the trace of the lifted action at a
fixed point~$f\in F$.

  \begin{proposition}[]\label{thm:33}
 If $F$~consists of isolated points, then 
  \begin{equation}\label{eq:138}
     \tau _W(V) = \exp\left( 2\pi i\sum\limits_{f\in F}\frac{i_f\tau
     _f}{2^8} \right) . 
  \end{equation}
  \end{proposition}

\noindent
 If $V$~is the trivial real line bundle, then this is \cite[Proposition
5.3]{St}, which is based on the general equivariant index theorem
\cite[Theorem 1.2]{Do} for manifolds with boundary.  Donnelly's theorem
identifies the contribution at a fixed point in terms of an asymptotic
expansion of a heat kernel.  The general cohomological expression for that
contribution appears in \cite[(3.9)]{AS} in the context of the general
Lefschetz theorem, and it applies to fixed point manifolds of positive
dimension as well as isolated fixed points.  That this is the correct fixed
point contribution in Donnelly's theorem is proved in~\cite{DP} for the
signature operator.  We use it for the Dirac operator and an
orientation-reversing isometry in~\S\ref{subsubsec:6.2.5}.

   \section{Cubic forms and the $C$-field}\label{sec:4}

  \subsection{Motivation: $\spinc$ manifolds}\label{subsec:4.4}

Recall that the compact Lie
group~$\Spin^c_n$ is a group extension 
  \begin{equation}\label{eq:75}
     1\longrightarrow \TT\longrightarrow \Spin^c_n\longrightarrow
     SO_n\longrightarrow 1 
  \end{equation}
where $\TT$~is the circle group of complex numbers of unit norm; it is
defined as the quotient $(\Spin_n\times \TT)/\pmo$.  Let $M$~be
an $n$-dimensional $\spinc$ manifold.  A \emph{$\spinc$ structure} on~$M$ is
a principal $\Spin^c_n$-bundle $\FSpinc(M)\to M$ together with an isomorphism
$\FSpinc(M)/\TT\cong \FSO(M)$ with the principal $SO_n$-bundle of oriented
orthonormal frames.  The $\TT$-bundle over~$M$ associated to the homomorphism
$\textnormal{Spin}^c_n\to\TT$ is called the {\it characteristic
bundle\/},\footnote{The associated line bundle is often called the
determinant line bundle of the $\spinc$ structure.} and its first Chern
class~$c\in H^2(M;\ZZ)$ is an integer lift of the second Stiefel-Whitney class:
  \begin{equation}\label{eq:42}
     c\equiv w_2(M)\pmod2. 
  \end{equation}
Furthermore, any other $\spinc$ structure is obtained by ``tensoring'' with a
circle bundle $Q\to M$ using the homomorphism $\Spin^c_n\times \TT\to
\Spin^c_n$; the characteristic class of the new $\spinc$ structure is~$c+2x$,
where $x=c_1(Q)$.  Finally, there is an involution on $\spinc$ structures
which inverts the characteristic bundle and so changes the sign of~$c$.
 
Suppose $n=\dim M$~is even and $M$~is compact without boundary.  The
($\zt$-graded) complex spin representation of~$\textnormal{Spin}^c_n$ gives
rise to a Dirac operator~$D_M$ whose index is a topological invariant.  It is
computed by the Atiyah-Singer formula
  \begin{equation}\label{eq:43}
     \inde D_M = \langle\Ahat(M)e^{c/2}\,,\,[M]\rangle, 
  \end{equation}
where $\Ahat(M) = 1 - p_1(M)/24 + \dots $ and $[M]$~is the fundamental class
of~$M$.  As a function~$\kappa (c)$ of the characteristic class~$c$ it is a
polynomial, which for~$n=4,6$ may be written   
  \begin{align}
     \kappa _2(c) &= \frac{c^2 - \sigma (M)}{8}\;, \label{eq:44}\\ \kappa _3(c)
     &= \frac{c^3 - p_1(M)c}{48}\;.  \label{eq:45}
  \end{align} 
The subscript indicates the degree of the polynomial, $\sigma(M)$~is the
signature of the 4-manifold~$M$, and we omit evaluation on~$[M]$ from the
notation for convenience.  One may continue to~$n=8,10,\dots $ to obtain
polynomials of higher degree.  These polynomials satisfy a symmetry property:
  \begin{equation}\label{eq:46}
     \kappa _2(-c) = \kappa _2(c),\qquad \qquad \kappa _3(-c) = -\kappa
     _3(c). 
  \end{equation}

For a fixed characteristic element~$c$ define $q^c\:H^2(M;\ZZ)\to\ZZ$ as  
  \begin{equation}\label{eq:47}
     q^c(x) = \kappa (c+2x) - \kappa (c),\qquad x\in H^2(M;\ZZ).
  \end{equation}
For~$n=4,6$ we find   
  \begin{align}
     q_2^c(x) &= \frac 12(x^2 + cx) \label{eq:50}\\ q_3^c(x) &= \frac
      1{24}\bigl(p_1(M)x + 4x^3 + 6cx^2 + 3c^2x \bigr)\label{eq:48} \\ &=
     \,\,\frac16\,x^3 +\dots\nonumber
  \end{align} 
Note $q_2^c$~is a quadratic refinement of the intersection pairing
on~$H^2(M^4;\ZZ)/\torsn$, and $q_3^c$~is a cubic refinement of the symmetric
trilinear form on~$H^2(M^6;\ZZ)/\torsn$.
      
The general mathematical problem suggested here is: Replace~$c$ by a
cohomology class of arbitrary even degree and extend the topological
invariants~\eqref{eq:44}, \eqref{eq:45}.  Of course, one may pose this as
well for the higher degree polynomials of~$c$ deduced from the index
formula~\eqref{eq:43} in higher dimensions.  In the quadratic case we have
$n=4k$~for some~$k\in \ZZ^{>0}$ and~$c\in H^{2k}(M;\ZZ)$ lies in the middle
degree.  The associated topological invariant was investigated by
Brown~\cite{Bro} and Browder~\cite{Brd}.  In this instance $c$~is an integer
lift of the middle Wu class~$\nu _{2k}\in H^{2k}(M;\zt)$, which may or may
not exist.  Corresponding geometric invariants were constructed in~\cite{HS}.
We take up the next interesting case---the cubic form for~$n=12$ and $\deg
c=4$---which appears in the action of the $C$-field in M-theory.

  \subsection{Algebraic theory of cubic forms}\label{subsec:4.1}

We begin with a review of the algebraic theory of quadratic forms.  Let
$L$~be a finitely generated free abelian group and $\langle \cdot ,\cdot
\rangle\:L\times L\to\ZZ$ a nondegenerate (i.e., unimodular) symmetric
bilinear form.  The nondegeneracy implies the existence of a unique element
$\cbar\in \Lt$ such that
  \begin{equation}\label{eq:51}
     \langle \xbar,\xbar \rangle\equiv \langle \cbar,\xbar
     \rangle\pmod2,\qquad \xbar\in \Lt, 
  \end{equation}
since the left hand side is linear in~$\xbar$.  An element~$c\in L$
with~$c\equiv \cbar\pmod2$ is called {\it characteristic\/}.  The
set~$\Lchar\subset L$ of characteristic elements is a torsor for~$L$:
if~$c\in \Lchar$ and~$x\in L$ then $c+2x\in \Lchar$.  Now an easy check shows
that $\langle c,c \rangle\pmod8$~is independent of~$c\in \Lchar$, so for any
integer lift~$\sigma \in \ZZ$ of $\langle c,c  \rangle\pmod8$,
  \begin{equation}\label{eq:52}
     \kappa _2(c) = \frac{\langle c,c \rangle-\sigma }{8} 
  \end{equation}
is an integer.  It is a standard result~\cite[Chapter~5]{Se} that $\sigma
$~may be chosen to be the signature of~$\langle \cdot ,\cdot \rangle$,
defined by extending the form to the real vector space~$L\otimes \RR$.  This
is the algebraic theory which underlies~\eqref{eq:44}.  Note $\kappa _2(-c) =
\kappa _2(c)$.
 
We develop a similar theory for the cubic~\eqref{eq:45}.  Consider the triple
$(L,\langle \cdot ,\cdot ,\cdot \rangle,\cbar)$ where $L$~is a finitely
generated free abelian group, $\langle \cdot ,\cdot ,\cdot \rangle\:L\times
L\times L\to\ZZ$ is a symmetric trilinear form, and $\cbar\in \Lt$ is assumed
to satisfy\footnote{Equation~\eqref{eq:53} for trilinear forms appears in
Postnikov's study~\cite{Po} of the mod~2 cohomology ring of a closed
3-manifold, for example.}
  \begin{equation}\label{eq:53}
     \langle \cbar,\xbar,\ybar \rangle \equiv \langle \xbar,\xbar,\ybar
     \rangle + \langle \xbar,\ybar,\ybar \rangle\pmod2,\qquad \xbar,\ybar\in
     \Lt. 
  \end{equation}
As we do not know a notion of nondegeneracy for trilinear forms which
guarantees the existence of~$\cbar$, we postulate its existence.  Define the
torsor~$\Lchar\subset L$ of characteristic elements as above.  Let
$L^*=\Hom(L,\ZZ)$ and for convenience write the trilinear form as a simple
product.

  \begin{lemma}[]\label{thm:10}
 There exists a unique $\phat\in L^*\otimes \ZZ/24\ZZ$ such that  
  \begin{equation}\label{eq:54}
     \phat\cdot \xhat \equiv 4\xhat^3 + 6\chat\xhat^2 + 3\chat^2\xhat
     \pmod{24} 
  \end{equation}
for all $\xhat\in L\otimes \ZZ/24\ZZ$ and mod~24 reductions~$\chat$ of
characteristic elements $c\in \Lchar$.
  \end{lemma}

  \begin{proof}
 Use~\eqref{eq:53} to check that, as a function of~$\xhat$, the right hand
side of~\eqref{eq:54} defines a homomorphism $L\otimes
\zmod{24}\to\zmod{24}$.
  \end{proof}

  \begin{lemma}[]\label{thm:11}
 Let $p\in L^*$ satisfy $p\equiv \phat\pmod{24}$.  Then  
  \begin{equation}\label{eq:56}
     \frac{c^3 - p\cdot c}{24}\pmod2 
  \end{equation}
lies in~$\zt$ and is independent of~$c\in \Lchar$.  Furthermore, there exist
lifts~$p\in L^*$ of~$\phat$ such that this invariant vanishes, in which case
  \begin{equation}\label{eq:55}
     \kappa _3(c) = \frac{c^3 - p\cdot c}{48} 
  \end{equation}
is an integer.  Also, $\kappa _3(-c) = -\kappa _3(c)$.
  \end{lemma}

  \begin{proof}
 To check the independence of~$c\in \Lchar$, replace $c$~in \eqref{eq:56}
with~$c+2x$ for $x\in L$ and use the fact that $cx^2$~is even, which follows
from~\eqref{eq:53}.  To see that the fraction in~ \eqref{eq:56} is an
integer, use~\eqref{eq:54} and the fact that $c^3$~is even, which also
follows from~\eqref{eq:53}.  To find the lift~$p$, if $\cbar=0$, then any~$p$
works since we can compute~\eqref{eq:56} using~$c=0$.  If $\cbar\neq 0$, and
if for a chosen lift~$p$ the invariant~\eqref{eq:56} is nonzero, choose
$x^*\in L^*$ such that $x^*\cdot c$~is odd for any characteristic~$c$ and
replace~$p$ with $p+x^*$.
  \end{proof}

  \subsection{The cubic form on spin 12-manifolds}\label{subsec:4.2}

In~\eqref{eq:45} we gave an example of the cubic form~\eqref{eq:55} for a
closed oriented 6-manifold~$M^6$, where $L=H^2(M;\ZZ)/\torsn$, $\langle x,y,z
\rangle = \langle x\smile y\smile z\,,\,[M]\rangle$, $\cbar=w_2(M)$,
and~$p=p_1(M)$.  We now consider a closed {\it spin\/} 12-manifold~$W^{12}$
and set
  \begin{equation}\label{eq:57}
     \begin{aligned} L &= H^4(W;\ZZ)/\torsn \\ \langle x,y,z \rangle &=
     \langle x\smile 
      y\smile z\,,\,[W]\rangle \\ \cbar&=w_4(W)\end{aligned} 
  \end{equation}

  \begin{remark}[]\label{thm:12}
 Let $T^4\subset H^4(W;\ZZ)$ denote the torsion subgroup, which fits into the
exact sequence 
  \begin{equation}\label{eq:76}
     0\longrightarrow T^4\longrightarrow H^4(W;\ZZ)\longrightarrow
     L\longrightarrow 0 
  \end{equation}
Tensoring with~$\zt$ defines a homomorphism $H^4(W;\ZZ)\longrightarrow
H^4(W;\zt)$, and the precise definition of~$\cbar$ is the image of~$w_4(W)$
under the quotient map
  \begin{equation}\label{eq:67}
     H^4(W;\zt)\longrightarrow  H^4(W;\zt)/(T^4\otimes \zt).
  \end{equation}
In the classifying space~$\BSpin$ there is a characteristic class~$\lambda
\in H^4(\BSpin;\ZZ)$ such that (i)~$2\lambda =p_1$ and (ii)~ the image of~
$\lambda$ under $H^4(\BSpin;\ZZ)\to H^4(\BSpin;\zt)$ is~$w_4$.  A spin
manifold~$W$ has a corresponding integer characteristic class~$\lambda (W)$.
The existence of this integer lift of~$w_4(W)$ implies that the image
of~$w_4(W)$ under~\eqref{eq:67} lies in the subgroup
$\bigl(H^4(W;\ZZ)/T^4\bigr)\otimes \zt= L\otimes \zt$.  The computations
below are written in~$H^{\bullet }(W;\zt)$, but the results should be
interpreted in terms of this subquotient.  (To do so, use the fact that
torsion integer cohomology classes evaluate trivially on the fundamental
class.)
  \end{remark}

  \begin{lemma}[]\label{thm:13}
 The Stiefel-Whitney class~$\cbar=w_4(W)$ of a closed spin 12-manifold~$W$
satisfies~\eqref{eq:53}.
  \end{lemma}

The proof uses the Cartan formula and Adem relations for Steenrod squares, as
well as the Wu formula, which states that on a closed $n$-manifold~$M$ there
is a class $\nu _i(M)\in H^i(M;\zt)$ such that squaring to the top,
  \begin{equation}\label{eq:68}
     Sq^i\:H^{n-i}(M;\zt)\longrightarrow H^n(M;\zt) ,
  \end{equation}
is cup product with~$\nu _i(M)$.  In low degrees we have 
  \begin{equation}\label{eq:69}
     \begin{aligned} \nu _1&=w_1\\ \nu _2&=w_1^2+w_2\\ \nu _3&=w_1w_2\\
     \nu_4&=w_4 + w_1w_3 + w_2^2 + w_1^4 \\ \end{aligned} 
  \end{equation}
in terms of the Stiefel-Whitney classes of the \emph{tangent} bundle.
(See~\cite[\S E.1.1]{HS} for characteristic properties of Wu classes from
which \eqref{eq:69}~may be computed.)  The Bockstein~$\beta $ is defined as
the connecting homomorphism induced from the coefficient sequence
$0\to\ZZ\xrightarrow{\;2\;}\ZZ\to\zt\to0$:
  \begin{equation}\label{eq:70}
     \cdots \longrightarrow H^i(M;\ZZ)\xrightarrow{\;\;r\;\;}
     H^i(M;\zt)\xrightarrow{\;\;\beta \;\;}H^{i+1}(M;\ZZ)
     \xrightarrow{\;\;2\;\;} H^{i+1}(M;\ZZ)\longrightarrow \cdots 
  \end{equation}
Also, $Sq^1=r\circ \beta $.  It follows that if $x$~is an integer cohomology
class, $Sq^1$~vanishes on its mod~2 reduction $\xbar=r(x)$.

  \begin{proof*}
 If $x,y\in H^4(W;\ZZ)$ and $\xbar,\ybar\in H^4(W;\zt)$ are their mod~2
reductions, then \eqref{eq:69} with $w_1=w_2=0$ implies
  \begin{equation}\label{eq:58}
     \begin{split} w_4(W)\xbar\ybar & = \Sq4(\xbar\ybar) \\ &=(\Sq4\xbar)\ybar
      + \Sq2\xbar\;\Sq2\ybar + \xbar\,\Sq4\ybar \\ &= \xbar\xbar\ybar +
      \Sq2\xbar\;\Sq2\ybar + \xbar\ybar\ybar,\end{split} 
  \end{equation}
since $\Sq1$~vanishes on reductions of integer classes.  Then
from~\eqref{eq:69} again
  \begin{equation}\label{eq:59}
     \begin{split} 0&=w_2(W)\xbar\,\Sq2\ybar \\&= \Sq2(\xbar\;\Sq2\ybar) \\ &=
      \Sq2\xbar\;\Sq2\ybar + \xbar\,\Sq2\Sq2\ybar \\ &= \Sq2\xbar\;\Sq2\ybar
      + \xbar\,\Sq3\Sq1\ybar \\ &=
     \Sq2\xbar\;\Sq2\ybar.\rlap{\qquad\qquad\qquad\qquad\qquad
     \;\qquad\qquad\qquad\qquad\qquad\qedsymbol}\end{split} 
  \end{equation}
  \end{proof*}

  \begin{proposition}[]\label{thm:23}
 In~$\BSpin$ there is a unique characteristic class~$p\in H^8(\BSpin;\ZZ)$ with 
  \begin{equation}\label{eq:66}
     2p = p_2 - \lambda ^2. 
  \end{equation}
Furthermore, $p\equiv w_8\pmod2$.   
  \end{proposition}

\noindent
 For a smooth manifold~$M$ we obtain a characteristic class $p(M)\in
H^8(M;\ZZ)$, and we use the same symbol to denote its reduction modulo
torsion.

  \begin{proof}
 $H^8(\BSpin;\ZZ)$ is torsionfree and $p_2\equiv \lambda ^2\equiv
w_4^2\pmod2$, which proves the existence and uniqueness of~$p$.  To compute
its reduction mod~2 we restrict to the classifying space of a maximal torus
of~$\Spin_N$ for $N\ge8$.  The computation is carried out in~\cite[\S3]{BW},
where $p$~ is the class called~`$-q_2$'.  Its reduction mod~2 equals the
reduction of the class called~`$c_4$', which is the Stiefel-Whitney
class~$w_8$. 
  \end{proof}

  \begin{remark}[]\label{thm:45}
 Hopkins-Singer~\cite[Appendix~E]{HS} define spin Wu classes in~$H^{\bullet
}(\BSpin;\ZZ)$, in terms of which we have $p=\nu ^{\textnormal{Spin}}_8 -
2p_2 + p_1^2$.  
  \end{remark}

  \begin{proposition}[]\label{thm:15}
 On a closed spin 12-manifold~$W$ the mod~24 reduction of~$p(W)$
satisfies~\eqref{eq:54}.
  \end{proposition}

  \begin{proof}
 We follow Witten's argument in~\cite[\S4]{W3}.  Namely, a principal
$E_8$-bundle over a 12-manifold is determined up to isomorphism by an
element~$x\in H^4(W;\ZZ)$.  Let $V(x)$~denote the (real) adjoint vector
bundle to the principal $E_8$-bundle with characteristic class~$x$, and set
$c=\lambda (W)+ 2x$.  The Chern character of~$V(x)\to W$ is\footnote{\emph{A
priori} the Chern character is a cubic polynomial in~$x$, so we need only
determine the coefficients.  The restriction of the adjoint representation
of~$E_8$ to $\Spin_{16}\subset E_8$ is the sum of a half-spin representation
and the adjoint representation of~$\Spin_{16}$.  The restriction of its
complexification to $\Spin_3\subset \Spin_{16}$ is $78V_1\oplus 64V_2\oplus
14V_3$, where $V_n$~is the $n$-dimensional irreducible representation
of~$\Spin_3\cong SU_2$; the Chern character of this representation is easily
computed.  Finally, the generator of~$H^4(BE_8;\ZZ)$ restricts to minus twice
the generator of~$H^4(BSU_2;\ZZ)$.  (The generator of~$H^4(BSO_{16};\ZZ)$
restricts to the generator of~$H^4(BSO_3;\ZZ)$.  The former pulls back to
twice the generator of~$H^4(\BSpin_{16};\ZZ)$, whereas the latter pulls back
to minus four times the generator of~$H^4(BSU_2;\ZZ)$.)}
  \begin{equation}\label{eq:79}
     \ch V(x) = 248 - 60x + 6x^2 - \frac 13x^3. 
  \end{equation}
Then a long computation verifies the following identity:
  \begin{equation}\label{eq:60}
     \left\langle\frac{c^3 - 
     pc}{48} \;+\; \frac 12\Ahat(W)\ch V(x) \;+\;
     \frac14\Ahat(W)\ch(TW-4)\,,\,[W]\right\rangle =0. 
  \end{equation}
The second term is an integer; it is the $KO$-theory direct image of the real
bundle~$V(x)$, defined using the spin structure, which by the Atiyah-Singer
index theorem is the index of the Dirac operator coupled to~$V(x)$.
Similarly, the last term is a half-integer, hence so is the cubic
expression.  Replace the denominator in the cubic expression by~24 to obtain
an integer, and now subtract the integers for arbitrary~$x$ and~$x=0$ to
establish the congruence
  \begin{equation}\label{eq:61}
     \bigl(p_2(W) - \lambda (W)^2 \bigr)x\equiv 8x^3 + 12\lambda (W)x^2 +
     6\lambda (W)^2x\pmod{24}, 
  \end{equation}
where we omit evaluation on~$[W]$ from the notation for convenience.  If
necessary, use the last argument in Lemma~\ref{thm:11} to replace~$p(W)$
by~$p'=p(W)+24a$ for~$a\in H^8(W;\ZZ)/\torsn$ so that
  \begin{equation}\label{eq:62}
     \frac{c^3 - p'c}{48}\in \ZZ,\qquad c=\lambda (W)+2x,\quad x\in H^4(W;\ZZ), 
  \end{equation}
and so deduce the desired mod~24 congruence.
  \end{proof}

Note that $p(W)$~is {\it not\/} necessarily a distinguished lift of~$\phat$
described in Lemma~\ref{thm:11}; rather we need to add the constant
term~$1/4\,\RS'(W)$ (see~\eqref{eq:81}) in~\eqref{eq:60} is needed to obtain
integrality.  Define the integer-valued cubic form 
  \begin{equation}\label{eq:63}
     \kW(c) = \frac{c^3 - p(W)c}{48} + \frac 12 \RS'(W)
  \end{equation}
on characteristic elements; it satisfies a shifted version of the
symmetry~\eqref{eq:46}:
  \begin{equation}\label{eq:64}
     \kW(-c) = \RS'(W) - \kW(c). 
  \end{equation}

  \subsection{The cubic form on $\pinp\!$ 12-manifolds}\label{subsec:4.5}

Any manifold~$M$ has a canonical \emph{orientation double cover} $\hM\to M$:
the fiber at~$m\in M$ is the set of orientations on~$T_mM$.  There results a
canonical local system $\tZ\to M$ of coefficients; we call $H^{\bullet
}(M;\tZ)$ the \emph{$w_1$-twisted cohomology}.  An orientation is a
trivialization of $\tZ\to M$, and on an oriented manifold $w_1$-twisted
integer cohomology reduces to untwisted integer cohomology.  The fundamental
class~$[M]$ of a closed manifold~$M$ lives in $w_1$-twisted integer homology,
so we can integrate $w_1$-twisted cohomology classes.

  \begin{lemma}[]\label{thm:34}
 Let $M$~be a closed $n$-manifold with no orientable components and $\pi
\:\hM\to M$ the orientation double cover.  Then the image of 
  \begin{equation}\label{eq:139}
     \pi ^*\:H^n(M;\tZ)\longrightarrow H^n(\hM;\ZZ) 
  \end{equation}
is $2H^n(\hM;\ZZ)$, and if $\bw\in H^n(M;\tZ)$, then 
  \begin{equation}\label{eq:140}
     \langle \bw,[M] \rangle = \langle \frac 12\pi ^*\bw,[\hM] \rangle. 
  \end{equation} 
  \end{lemma}

\noindent
 As the domain and codomain of~\eqref{eq:139} are torsionfree, we can prove
Lemma~\ref{thm:34} using de Rham theory, a task we leave to the reader.

Let $W$~be a closed $\pinp$ 12-manifold~$W$.  The existence of a $\pinp$
structure on~$W$ is equivalent to $w_2(W)=0$, but in general $w_1(W)\neq 0$.
Also, $w_{3}(W)=0$ since $w_3=\Sq1w_2+w_1w_2$ (Wu formula).  Note then that
the Wu classes~\eqref{eq:69} simplify to $\nu _2(W)=w_1(W)^2$ and $\nu
_4(W)=w_4(W)+ w_1(W)^4$.  There is a short exact sequence of coefficients
$0\to\tZ\xrightarrow{\;2\;}\tZ\to\zt\to0$, and the connecting homomorphism in
the resulting long exact sequence---\eqref{eq:70}~with twisted
coefficients---is the twisted Bockstein~$\tb$.  In this case $r\circ
\tb=Sq^1+w_1$, so that if $x$~is a $w_1$-twisted integer class then
  \begin{equation}\label{eq:71}
     Sq^1\xbar = w_1(M)\smile x. 
  \end{equation}

For a closed $\mc$ 12-manifold~$W$ we modify~\eqref{eq:57} to
  \begin{equation}\label{eq:65}
     \begin{aligned} L &= H^{4}(W;\tZ)/\torsn \\ \langle x,y,z \rangle &=
      (x\smile y\smile z)[W] \\ \cbar&=w_4(W)\end{aligned} 
  \end{equation}
Remark~\ref{thm:12} applies if we replace integer cohomology with
$w_1$-twisted integer cohomology and assume $W$~is an $\mc$-manifold.  The
dual lattice $L^*=H^8(W;\ZZ)/\torsn$ is untwisted integer cohomology as in
the spin case.
 
  \begin{proposition}[]\label{thm:24}
 In $\BPinp$ there is a unique characteristic class $\bp\in
H^8(\BPinp;\ZZ)/\torsn$ whose restriction to~$\BSpin$ is the class~$p$ of
Proposition~\ref{thm:23}.   
  \end{proposition}

  \begin{proof}
 Let $\{E^{p,q}_r\}$ denote the Leray-Serre spectral sequence for the fibration 
  \begin{equation}\label{eq:86}
      \BSpin\longrightarrow \BPinp\xrightarrow{\;\;w_1\;\;}\RP^{\infty}. 
  \end{equation}
Then $E_2^{0,8}\cong H^8(\BSpin;\ZZ)$ and $E_\infty ^{0,8}\cong
H^8(\BPinp;\ZZ)/\torsn\cong \ker(d_2\:E_2^{0,8}\to E_2^{2,7})$.  Note that
$E_2^{2,7}\cong \zt$, since $H^7(\BSpin;\ZZ)$~is cyclic of order~2, generated
by the integer Bockstein of~$w_6$, and $H^2\bigl(\RP^{\infty};\zt \bigr)\cong
\zt$.  The proposition follows from $d_2(p)=0$, which in turn follows since
$d_2(p)$~is detectable mod~2 and $p\pmod2=w_8$ survives the differentials.
  \end{proof}

  \begin{remark}[]\label{thm:40}
 There is a (homotopy) splitting of the map~$w_1$ in~\eqref{eq:86}, namely
the classifying map $\RP^{\infty}\to BO$ of the reduced tautological bundle
$H\to\RP^{\infty}$, which lifts since $H$~has a
$\pinp$~structure.\footnote{Introduce an inner product on $H\to \RP^{\infty}$
and use a splitting of the homomorphism $\Pin^+_1\to O_1$.}  Then the
product map 
  \begin{equation}\label{eq:169}
     \BSpin\times \RP^{\infty}\to \BPinp 
  \end{equation}
is a homotopy equivalence, since it induces an isomorphism on homotopy
groups.  This yields an isomorphism $H^{\bullet }(\BPinp;\ZZ)/\torsn\to
H^{\bullet }(\BSpin;\ZZ)/\torsn$, which re-proves Proposition~\ref{thm:24}.
  \end{remark}

  \begin{proposition}[]\label{thm:16}
 Let $W$~be a closed $\mcable$ 12-manifold.  Then $\cbar=w_4(W)$
satisfies equation~\eqref{eq:53}.  Furthermore, the mod~24 reduction
of~$\bp(W)$, viewed as a class in~$L^*=H^8(W;\ZZ)/\torsn$, satisfies the
condition in Lemma~\ref{thm:10}.
  \end{proposition}

  \begin{proof}
 We modify the proof of Lemma~\ref{thm:13}.  So \eqref{eq:58}~becomes
  \begin{equation}\label{eq:72}
     \bigl(w_4(W) + w_1^4(W) \bigr)\xbar\ybar = \xbar\xbar\ybar +
     \xbar\ybar\ybar + 
     w_1(W)(\xbar\,\Sq3\ybar + \ybar\,\Sq3\xbar) + \Sq2\xbar\;\Sq2\ybar
  \end{equation}
and \eqref{eq:59} becomes 
  \begin{equation}\label{eq:73}
     w_1^2(W)\xbar\,\Sq2\ybar = \Sq2\xbar\;\Sq2\ybar + w_1(W)\xbar\Sq3\ybar +
     \xbar w_1^2(W)\Sq2\ybar + \xbar w_1(W)\Sq3\ybar,
  \end{equation}
which implies $\Sq2\xbar\;\Sq2\ybar=0$.  Then using $\nu _3(W)=0$
from~\eqref{eq:69}, we find 
  \begin{equation}\label{eq:74}
  \begin{split}
     0 = \Sq3(w_1(W)\xbar\ybar) &= w_1^2(W)\Sq2(\xbar\ybar) \;+\;
     w_1(W)\Sq3(\xbar\ybar) \\&= w_1^4(W)\xbar\ybar \;+\;
     w_1(W)\bigl(\xbar\Sq3\ybar + \ybar\Sq3\xbar\bigr). 
  \end{split}
  \end{equation}
Combine these equations to complete the proof that
$\cbar$~satisfies~\eqref{eq:53}.  
 
For the last statement in the proposition we observe that
Proposition~\ref{thm:24} implies $\bp(W)=p(W)$ if $W$~is spin, and also if
$\pi \:\hW\to W$ is the orientation double cover then $\pi ^*\bp(W)=p(\hW)$.
The last statement reduces to Proposition~\ref{thm:15} on orientable
components of~$W$, and on nonorientable components we use Lemma~\ref{thm:34}
to reduce to Proposition~\ref{thm:15} on the orientation double cover.
  \end{proof}

  \begin{lemma}[]\label{thm:21}
 Let $W$~be a closed $\mcable$ 12-manifold and $\tc\in H^4(W;\tZ)$ a
$w_1$-twisted integer lift of~$w_4(W)$.  Then 
  \begin{equation}\label{eq:82}
     \frac{\tc^3 - \bp(W)\tc}{48}\pmod\ZZ
  \end{equation}
lies in $\frac 12\ZZ/\ZZ$, is independent of the choice of~$\tc$, and is a
bordism invariant of $\mcable$-manifolds.  It is additive under
disjoint union.
  \end{lemma}

  \begin{proof}
 That the fraction in~\eqref{eq:82} is a half-integer follows from
Lemma~\ref{thm:11} in the algebraic theory of cubic forms.  Any $w_1$-twisted
integer lift of~$w_4(M)$ has the form~$\tc + 2\tx$ for some~$\tx\in
H^4(W;\tZ)$, and an easy check from~\eqref{eq:54} proves that
\eqref{eq:82}~is unchanged by the replacement.  If $W=\partial Z$ is the
boundary of a compact $\mcable$ 13-manifold~$Z$, then $Z$~has a fundamental
class in relative homology and the usual adjunction (integer Stokes'
theorem) argument implies that \eqref{eq:82}~vanishes, even before reducing
modulo~$\ZZ$.
  \end{proof}

Define 
  \begin{equation}\label{eq:83}
     \hac(W) = \exp\left(2\pi i\;\frac{\tc^3 - \bp(W)\tc}{48}\right). 
  \end{equation}
Recall that $\sB$~is the bordism spectrum of $\pinp$ manifolds with an $\mc$
structure. 

  \begin{corollary}[]\label{thm:22}
 The exponential of the cubic form factors through a homomorphism 
  \begin{equation}\label{eq:84}
     \hac\:\pi _{12}\sB\longrightarrow \Cx 
  \end{equation}
which takes values in~$\pmo\subset \Cx$.
  \end{corollary}
As discussed in~\S\ref{sec:2}, the homomorphism~\eqref{eq:40} determines an
invertible topological field theory  
  \begin{equation}\label{eq:85}
     \ac\:\sB\longrightarrow \Sigma ^{12}\ICx 
  \end{equation}
up to isomorphism.  The square~$\alpha _C^{\otimes 2}$ is isomorphic to the
trivial theory.

  \begin{remark}[]\label{thm:32}
 Let $\pi \:\hW\to W$ be the orientation double cover of an $\mc$-manifold
which has no orientable components, and suppose $\tc\in H^4(W;\tZ)$ is a
$w_1$-twisted integer lift of~$w_4(W)$.  Set $c=\pi ^*\tc\in H^4(\hW;\ZZ)$.  As in
the proof of Proposition~\ref{thm:16} we have $p(\hW)=\pi ^*\bp(W)$.  Apply
Lemma~\ref{thm:34} to evaluate the \emph{integer} cubic
form---twice~\eqref{eq:82}---on the orientation double cover:
  \begin{equation}\label{eq:132}
     \bigl\langle\frac{\tc^3 - \bp(W)\tc}{24}\;,\;[W]\bigr\rangle \;=\;
     \bigl\langle\frac{c^3 - p(W)c}{48}\;,\;[\hW] \bigr\rangle. 
  \end{equation} 
  \end{remark}

  \subsection{The $C$-field and its anomaly; cancellation on spin manifolds}\label{subsec:4.3}
 
The $C$-field in M-theory is an example of an abelian gauge field.
Classically all information is captured by its field strength~$\Omega $,
which is a closed 4-form.  In the quantum theory Dirac's quantization of
charge applies: the de Rham cohomology class of~$\Omega $ is constrained to
lie in a full lattice in the degree~4 real cohomology.  There is more
information, as inspired by the Aharanov-Bohm effect in the case of ordinary
electromagnetism and the resulting refinement of the electromagnetic
field---a closed 2-form---to a connection on a principal $\TT$-bundle.  In
higher degrees a suitable language for quantum abelian gauge fields is
\emph{differential cohomology},\footnote{See~\cite[\S3]{F2} for a general
exposition of abelian gauge fields as differential cocycles.} which is
developed in~\cite{HS} in part to model the $C$-field; the focus there is on
the M5~brane and so on a quadratic form.  Here we work in the ``bulk'' on a
Wick-rotated spacetime which we take to be an 11-dimensional Riemannian
$\pinp$ manifold~$X$.  Dirac's quantization of charge for the $C$-field,
which is determined in~\cite{W3}, is encoded by positing the $C$-field as a
geometric representative of a $w_1$-twisted differential cohomology
class\footnote{If $X$~is spin, then there is a model~\cite{DFM} in terms of
$E_8$-bundles, as in the proof of Proposition~\ref{thm:15}.}  which
lifts~$w_4(X)$.  Locally $C$-fields exist but there is a global obstruction,
as explained after Definition~\ref{thm:19}.  In that spirit, a
$C$-field~$\cO$ is a \emph{differential $\mc$ structure} on~$X$; a precise
model is established in~\cite{HS}, where it is termed a \emph{differential
integral Wu structure}.  Its field strength~$\Omega $ is a closed
$w_1$-twisted 4-form\footnote{$\Omega $~lifts to a closed 4-form~$\hO$ on the
total space of the orientation double cover $\hX\to X$; then $\sigma
^*\hO=-\hO$, where $\sigma \:\hX\to\hX$ is the nontrivial deck
transformation.} whose de Rham cohomology class in~$H^4(X;\tR)$ is the real
image of a $w_1$-twisted integer lift~$\tc\in H^4(X;\tZ)$ of~$w_4(X)$.
 
The effective action of M-theory has a cubic term of the form 
  \begin{equation}\label{eq:87}
     \exp\left( 2\pi i\;\frac{\cO^3-\cp(X)\cO}{48} \right) , 
  \end{equation}
where $\cp(X)$~is a lift to differential cohomology of the class $\bp(X)\in
H^8(X;\ZZ)/\torsn$.  This differential cohomology version of the cubic form
is analogous to a Chern-Simons invariant.  We do not need its precise
definition, so will not elaborate further. 

  \begin{remark}[]\label{thm:26}
 The $\cO^3$~term in~\eqref{eq:87} is part of the classical 11-dimensional
supergravity action~\cite{CJS}.  The $\cp(X)\cO$~term is a quantum
correction, introduced in~\cite[(3.14)]{DLM} in the spin case, in part
inspired by~\cite[\S3]{VW} who introduce an analogous correction in the
Type~IIA superstring.  We do not know of any literature about this quantum
correction in the $\pinp$ case.
  \end{remark}

  \begin{remark}[]\label{thm:25}
 We have only defined the class~$\cp$ in~ \eqref{eq:87} up to an element
of~$H^7(\BPinp;\RZ)$, but we now argue that \eqref{eq:87}~is independent of
the lift.  First, $H^7(\BPinp;\RZ)\cong H^8(\BPinp;\ZZ)\tors$, since
$H^7(\BPinp;\RR)=0$.  ($A\tors$~is the torsion subgroup of the abelian
group~$A$.)  Recall from~\eqref{eq:169} that $\BPinp\simeq \BSpin\times
\RP^{\infty}$.  Then the main theorem in~\cite{Ko} implies that
$2H^8(\BPinp;\ZZ)\tors=0$.  Use the short exact sequence 
  \begin{equation}\label{eq:176}
     0\longrightarrow \frac 12\ZZ/\ZZ\longrightarrow
     \RR/\ZZ\xrightarrow{\;\;2\;\;}\RZ\longrightarrow 0 
  \end{equation}
to deduce that $H^7(\BPinp;\frac12\ZZ/\ZZ)\longrightarrow H^7(\BPinp;\RZ)$ is
surjective.  It follows that the ambiguity in~\eqref{eq:87} is expressed as a
characteristic number of mod~2 cohomology.  Our Adams spectral sequence
computation (see Figure~\ref{fig:mm1} in~\S\ref{sec:9}) shows that there is
no element of~$\pi \mstrut _{11}\sB$ in Adams filtration~0, and so every
mod~2 characteristic number vanishes on closed 11-dimensional
$\mc$-manifolds.

Note that there is an ambiguity in the M-theory action from a topological
term which is not mod~2 characteristic numbers but rather a mod~2
$KO$-characteristic number, a mod~2 index of a Dirac operator;
see~\S\ref{sec:8}.
  \end{remark}

Our focus is on anomalies, and here the crucial point is that only the square
of~\eqref{eq:87} is unambiguously defined as an element of~$\CC$.  This is
equivalent to the assertion that on a closed $\mc$ 12-manifold~$W$
the cubic form 
  \begin{equation}\label{eq:88}
     \frac{\tc^3-\bp(W)\tc}{24} 
  \end{equation}
is integral, but is not necessarily even.  Hence the square
root~\eqref{eq:87} is an element of a complex line~$\ac(X)$ whose
square~$\ac(X)^{\otimes 2}$ is trivialized.  As the notation suggests, this
line is the state space of the invertible 12-dimensional field theory~$\ac$.
The field theory~$\ac$ is topological: it does not depend on the Riemannian
metric or differential $\mc$~structure, only on the underlying topological
$\mc$~structure.

Witten's argument~\cite[\S4]{W3}, reproduced in the proof of
Proposition~\ref{thm:15}, proves the Anomaly Cancellation Theorem~\ref{thm:1}
on spin manifolds.  Let $\sBs$ denote the bordism spectrum of spin manifolds
with an integer lift of~$w_4$.  There is a map $\sBs\to\sB$, where $\sB$~is
the bordism spectrum of $\mc$-manifolds.

  \begin{theorem}[Witten]\label{thm:27}
 The lift $\af\otimes \ac\:\sBs\to\tar$ is trivializable.  
  \end{theorem}

\noindent
 The $E_8$-model for the $C$-field leads to a distinguished
trivialization~\cite{FM}.   

  \begin{proof}
 Because an invertible topological field theory is determined up to
isomorphism by its partition functions, to prove Theorem~\ref{thm:27} we show
that for any closed spin 12-manifold~$W$ with an $\mc$~structure we have
  \begin{equation}\label{eq:89}
     \haf(W)\hac(W)=1. 
  \end{equation}
This follows immediately from the integrality of~\eqref{eq:63}; see
Proposition~\ref{thm:18} and Remark~\ref{thm:20}.
  \end{proof}

   \section{Some spin and pin manifolds}\label{sec:5}

This section is a geometric interlude to review and introduce some special
manifolds and their topological invariants.  We use these manifolds as
building blocks for the closed $\pinp$ 12-manifolds we need in~\S\ref{sec:6},
where we also specify $\mc$~structures.
 
If $M$~is a smooth manifold, then we use the notations 
  \begin{equation}\label{eq:96}
     \begin{aligned} w(M) &= 1 + w_1(M) + w_2(M) + \cdots \\ p(M) &= 1 +
      \,p_1(M) + \,p_2(M) + \cdots \\ \end{aligned} 
  \end{equation}
for the total Stiefel-Whitney class and total Pontrjagin class, respectively.
The former satisfies the Whitney sum formula $w(M_1\times M_2)=w(M_1)w(M_2)$
for Cartesian products; the analogous equation for the total Pontrjagin class
is true modulo torsion.  Also, these characteristic classes are defined for
arbitrary real vector bundles, not just the tangent bundle, and are
\emph{stable} in the sense that they are unchanged by adding a trivial
bundle.  Recall also the characteristic class~$\lambda $ of a spin manifold,
or of a real vector bundle with a spin structure, characterized
after~\eqref{eq:67}; it satisfies $2\lambda =p_1$.

  \subsection{$K3$ surface}\label{subsec:5.1}

There is a moduli space of inequivalent complex K3 surfaces whose underlying
real 4-manifolds are all diffeomorphic.  For definiteness, then, we define
$K\subset \CP^3$ as the zero locus of the quartic
  \begin{equation}\label{eq:97}
     (z^0)^4 + (z^1)^4 + (z^2)^4 + (z^3)^4 =0, 
  \end{equation}
where $z^0,z^1,z^2,z^3$ are the standard homogeneous coordinates on~$\CP^3$.
It is a smooth closed real 4-manifold which is simply connected, and the
complex structure induces an orientation.  The Chern classes can be computed
from those of~$\CP^3$ and that of the normal bundle, which is the restriction
of $\mathcal{O}(4)\to\CP^3$ to~$K$, and from there we derive the
Stiefel-Whitney and Pontrjagin classes:
  \begin{equation}\label{eq:98}
     \begin{aligned} w(K) &= 1 \\ p(K) &= 1 -48k,\end{aligned} 
  \end{equation}
where $k\in H^4(K;\ZZ)\cong \ZZ$ is the positive generator.  In particular,
$w_2(K)=0$ and so $K$~admits a spin structure compatible with the
orientation, which is unique up to isomorphism since $K$~is simply connected.
Also,
  \begin{equation}\label{eq:99}
     \lambda (K) = -24k. 
  \end{equation}

  \subsection{Quaternionic projective plane}\label{subsec:5.2}

Let $\HP^2$~denote the space of one dimensional quaternionic subspaces of the
quaternionic vector space~$\HH^3$.  For definiteness we let the division
algebra~$\HH$ act on the right of~$\HH^3$.  In coordinates write
  \begin{equation}\label{eq:143}
     \HP^2=\left\{[q^0,q^1,q^2]: q^i\in \HH\right\}/\sim,\qquad
     [q^0,q^1,q^2]\sim [q^0h ,q^1h ,q^2h ],\quad h \in \HH^{\neq 0}.  
  \end{equation}
$\HP^2$ is a simply connected 8-manifold.  In fact, the filtration $*\subset
\HP^1\subset \HP^2$ provides a CW structure with a single 0-cell, 4-cell, and
8-cell.  The simple connectivity implies that up to isomorphism $\HP^2$~has a
unique spin structure compatible with a given orientation.
 
Let $L\to\HP^2$ be the tautological quaternionic line bundle; its fiber at a
point $\ell \in \HP^2$ is the quaternionic line~$\ell $.  There is a short
exact sequence 
  \begin{equation}\label{eq:100}
     0\longrightarrow L\longrightarrow \underline{\HH^3}\longrightarrow
     Q\longrightarrow 1 
  \end{equation}
of (right) quaternionic vector bundles; in the middle is the trivial bundle
with fiber~$\HH^3$ and the sequence defines the rank two quotient
bundle~$Q\to\HP^2$.  Note that the dual $L^*\cong \Hom\mstrut _{\HH}(L,\HH)$
is canonically a \emph{left} $\HH$-module.  The tangent bundle is identified
as the real vector bundle $\Hom\mstrut _{\HH}(L,Q)\cong Q\otimes \mstrut
_{\HH} L^*$, and it is the quotient in the short exact sequence of real
vector bundles
  \begin{equation}\label{eq:101}
     0\longrightarrow L\otimes \mstrut _{\HH}L^*\longrightarrow
     \underline{\HH^3}\otimes \mstrut _{\HH}L^*\longrightarrow 
     Q\otimes \mstrut _{\HH}L^*\longrightarrow 0, 
  \end{equation}
so its total Pontrjagin class is the quotient 
  \begin{equation}\label{eq:102}
     p(\HP^2) = \frac{p(\underline{\HH^3}\otimes \mstrut
     _{\HH}L^*)}{p(L\otimes \mstrut _{\HH}L^*)},
  \end{equation}
since $H^{\bullet }(\HP^2;\ZZ)$ is torsionfree.  The quaternionic line bundle
$L^*\to\HP^2$ is, by restriction of scalars to~$\CC\subset \HH$, a rank~2
complex vector bundle isomorphic to its complex conjugate, so its total Chern
class has the form $1-x$, where $x\in H^4(\HP^2;\ZZ)$; we call~$x$ the
\emph{quaternionic first Pontrjagin class}.  Restrict to $\HP^1\subset \HP^2$
and fix a nonzero quaternionic functional $\HH^2\to\HH$ to define a section
of $L^*\to\HP^1$ which vanishes transversely at a single point.  It follows
that $x$~generates $H^4(\HP^2;\ZZ)$.  Now $L\otimes \mstrut _{\HH}L^*$ splits
off a trivial real line bundle, and the orthogonal rank~3 bundle is the real
adjoint bundle of the complex 2-plane bundle underlying $L^*\to\HP^2$; the
first Pontrjagin class of the real adjoint bundle is~$4x$.  Therefore,
from~\eqref{eq:102}
  \begin{equation}\label{eq:103}
     p(\HP^2) = \frac{(1+x)^6}{(1+4x)} = 1 + 2x + 7x^2. 
  \end{equation}
(See \cite[\S15.5]{BH} for an alternative derivation.)  It follows that
  \begin{equation}\label{eq:104}
     \lambda (\HP^2)=x,\qquad  w_4(\HP^2)=\xbar,
  \end{equation}
where $\xbar\in H^4(\HP^2;\zt)$ is the mod~2 reduction of~$x$. 

  \begin{remark}[]\label{thm:35}
 As mentioned above, a quaternionic line bundle $L\to X$ has a quaternionic
first Pontrjagin class~$p_1^{\HH}(L)\in H^4(X;\ZZ)$ which equals minus the
second Chern class after restricting scalars to~$\CC\subset \HH$.  We can
also restrict scalars to~$\RR\subset \HH$ to obtain a rank~4 real vector
bundle $L_{\RR}\to X$, whose first Pontrjagin class satisfies
$p_1(L_{\RR})=2p_1^{\HH}(L)$.  The following general formula is useful, and
can be used in the derivation of~\eqref{eq:103}.  Suppose $R,L\to X$ are
right and left quaternionic line bundles with quaternionic first Pontrjagin
classes $r ,\ell\in H^4(X;\ZZ)$.  Then $R\otimes \mstrut _{\HH}L\to X$ is a
real 4-plane bundle with total Pontrjagin class $1+2(r +\ell) +(r-\ell )^2$.
  \end{remark}

Use $x^2\in H^8(\HP^2;\ZZ)\cong \ZZ$ to orient~$\HP^2$: choose the
fundamental class such that $\langle x^2,[\HP^2] \rangle=1$.

  \subsection{Bott manifold}\label{subsec:5.3}

The bordism group~$\pi _8\MSpin$ is free abelian of rank two: there is an
isomorphism  
  \begin{equation}\label{eq:121}
     \begin{aligned} \pi _8\MSpin&\longrightarrow \ZZ\oplus \ZZ \\
      [M]&\longmapsto \bigl(\Ahat(M),\sigma (M)\bigr)\end{aligned} 
  \end{equation}
which maps a closed spin 8-manifold to its $\Ahat$-genus and its signature.
The quaternionic projective plane has $\Ahat(\HP^2)=0$, $\sigma (\HP^2)=1$.
A closed spin manifold~$B$ with $\Ahat(B)=1$ is called a \emph{Bott
manifold}.  We needn't insist on vanishing signature, as that can always be
achieved by connected sum with copies of~$\HP^2$ or its orientation-reversal,
and indeed the Bott manifold we use has signature~$-224$.
 
We do not know of an elementary construction of a Bott manifold.  One
possibility is a Riemannian manifold~$B$ of special holonomy~$\Spin_7\subset
\Spin_8$, which necessarily satisfies $\Ahat(B)=1$ and is simply connected;
see~\cite[\S10.6]{J}.  Closed 8-manifolds with $\Spin_7$ holonomy were first
produced by Joyce.  A more topological approach leans on the work of Kervaire
and Milnor~\cite{MK}, \cite{KM}.  The Bott manifold~$B$ so constructed is
also simply connected, so admits a unique spin structure.  Briefly, plumb
together 8~copies of the disk bundle of the tangent bundle to~$S^4$ according
to the $E_8$~Dynkin diagram.  The resulting compact 8-manifold~$N$ has a
boundary which is an exotic 7-sphere.  The Kervaire-Milnor results imply that
a connect sum of 28~copies of the exotic sphere bounds a ball, hence we
define~ $B$ as the boundary connect sum of 28~copies of~$N$ and cap off with
a standard ball; see~\cite[\S6.5]{HBJ} for details.  The manifold~$B$ is
\emph{almost parallelizable}, i.e., admits a trivialization of the tangent
bundle away from a point.  This implies that $p_1(B)=0$, and from a
computation with the signature we deduce the total Pontrjagin class
  \begin{equation}\label{eq:105}
     p(B) = 1 - 1440b, 
  \end{equation}
where $b\in H^8(B;\ZZ)\cong \ZZ$ is the positive generator.  Note $\lambda
(B)=0$ and $w_4(B)=0$.  Then
  \begin{equation}\label{eq:106}
     \Ahat(B) =\left\langle \frac{7p_1^2-4p_2}{5760} ,[B] \right\rangle 
  \end{equation}
implies $\Ahat(B)=1$.  We use this Bott manifold in the sequel.

  \subsection{Real projective spaces}\label{subsec:5.4}

Let $L\to\RP^n$ be the tautological real line bundle.  Arguing as in the
second paragraph of~\S\ref{subsec:5.2} we deduce that the tangent bundle
to~$\RP^n$ is stably equivalent to 
  \begin{equation}\label{eq:119}
     (n+1)L-1. 
  \end{equation}
Then if $\alpha \in H^1(\RP^n;\zt)\cong \zt$ is the generator, we conclude
  \begin{equation}\label{eq:107}
     w(\RP^n) = (1+\alpha )^{n+1}. 
  \end{equation}
 
Projective 4-space~ $\RP^4$ has $w_2=0$, so admits a $\pinp$ structure, in
fact two distinct ones which are opposite in the sense of
Proposition~\ref{thm:29}.  Of course, $w_1(\RP^4)=\alpha $ so that $\RP^4$~is
not orientable, so not spin either.  Also, $w_4(\RP^4)=\alpha ^4$ is nonzero,
and we fix a $w_1$-twisted lift $\tc_{\RP^4}\in H^4(\RP^4;\tZ)\cong \ZZ$ which is a
generator.  
 
For~$n=12$ we compute from~\eqref{eq:107} that $\RP^{12}$~is not orientable;
is $\pinp$ with two opposite $\pinp$ structures; that $w_4(\RP^{12})=\alpha
^4$; and since $H^4(\RP^{12};\tZ)=0$ it does not admit an $\mc$~structure.

The $\eta $-invariants of~$\RP^4$ and~$\RP^{12}$ are computed
in~\cite[Corollary~5.4]{St}.  The results are reciprocal for the two opposite
$\pinp$ structures (Proposition~\ref{thm:29}), and we use the $\eta
$-invariant to pin down a choice.  Stolz's result follows from
Proposition~\ref{thm:33} (see~\eqref{eq:94} for notation\footnote{These
computations are for the plain Dirac operator; there is no external vector
bundle}):
  \begin{align}\label{eq:109}
    \tau \mstrut _{\RP^{4}} &= \exp\left(\frac{2\pi i}{2^4}\right)\\[1em]
    \tau \mstrut _{\RP^{12}} &= \exp\left(\frac{2\pi
    i}{2^8}\right)\label{eq:108} 
  \end{align}

For later use we quote from~\cite{KT1} the position of these real projective
spaces in $\pinp$ bordism.  In dimension~4 we have
  \begin{equation}\label{eq:122}
     \pi _4\MTPp \cong \zmod{2^4} 
  \end{equation}
and $\RP^4$~represents a generator.  In dimension~12 we have, as already
quoted in~\eqref{eq:39},
  \begin{equation}\label{eq:123}
     \pi _{12}\MTPp \cong \zmod{2^8}\;\oplus \;\zmod{2^4}\;\oplus \;\zmod{2^2} 
  \end{equation}
and $\RP^{12}$~represents a generator of the first factor.
Proposition~\ref{thm:mm3} below proves that
  \begin{equation}\label{eq:171}
     \pi _4\sB \cong \ZZ\;\oplus \;\zmod{2^3}. 
  \end{equation}
The pair $(\RP^4,\tc_{\RP^4})$ represents a generator of the infinite cyclic
summand.  The connected sum $\RP^4\#\RP^4$ has order~8 in~\eqref{eq:122} and
has vanishing~$w_4$ (since its value on the fundamental class is the mod~2
Euler number).    The pair $(\RP^4\#\RP^4,0)$ represents a generator of the
second summand in~\eqref{eq:171}.

  \begin{remark}[]\label{thm:38}
 Stolz uses Proposition~\ref{thm:33} to compute~\eqref{eq:109}
and~\eqref{eq:108}.  These results also follow\footnote{Zhang uses the $\eta
$-invariant of real projective spaces to prove Theorem~\ref{thm:37}, so
logically we are only illustrating the theorem here, not using it to
derive~\eqref{eq:109} and~\eqref{eq:108}.} from the topological
formula~\eqref{eq:146}, stated for $\pinp$ 12-manifolds but with an extension
to $\pinp$ $(8k+4)$-manifolds, $k\in \ZZ^{\ge0}$.  For example, let $\gamma
\:\RP^{12}\to\RP^{20}$ be a linear embedding, and $H\to\RP^{20}$ the
tautological real line bundle.  The normal bundle to~$\gamma $ is the
restriction of $H^{\oplus 8}\to\RP^{20}$ to~$\RP^{12}$, and $\RP^{12}$~is the
0-set of a section of $H^{\oplus 8}\to\RP^{20}$ (8~linearly independent
linear functions).  It follows that $\gamma _*(1)$~is the $KO$ Euler class of
$H^{\oplus 8}\to\HP^{20}$, which we compute\footnote{The Euler class is
associated to the difference of the half spin representations of~$\Spin_8$,
which restricted to the diagonal $\zt\subset \Spin_8$ is~$8(1-\epsilon )$,
where $\epsilon $~is the sign representation.} to be $8\bigl(1-[H] \bigr)$
after multiplication by the Bott class.  It remains to observe that
$8\bigl(1-[H] \bigr)$ has order~$2^8$ in~$\widetilde{KO}(\RP^{20})$.
  \end{remark}

  \subsection{Three special manifolds}\label{subsec:5.5}

We define three 12-dimensional manifolds $W_0',W_0'',W\mstrut _1$ which
appear in Theorem~\ref{thm:31} below.  Each of $W_0',W_0''$ is presented as
the quotient of its orientation double cover by a free involution.
 
\subsubsection{$W_0'$}\label{subsubsec:5.5.1} Set 
  \begin{equation}\label{eq:114}
     \widehat{W_0'} = S^4\times (\HP^2\#\HP^2), 
  \end{equation}
the Cartesian product of the 4-sphere and the connected sum of two
quaternionic projective planes.  As an explicit model of the connected sum,
fix a line through the origin in real affine space~$\AA^9$, remove two small
antipodal balls from~$S^8\subset \AA^9$ which are exchanged by the half-turn
about that line, and glue in two identical copies of $\HP^2\setminus B^8$.
Then \eqref{eq:114}~has a free orientation-reversing involution which is the
Cartesian product of the antipodal involution of~$S^4$ and the half-turn of
$\HP^2\#\HP^2$ with its two fixed points.  The quotient is the
manifold~$W_0'$.  Since $\widehat{W_0'}$~is simply connected, we have $\pi
_1W_0'\cong \zt$ and hence $H^1(W_0';\zt)\cong \zt$.  Since the involution is
free on~$S^4$, the manifold $W_0'$ fits into a fiber bundle
  \begin{equation}\label{eq:115}
     \HP^2\#\HP^2\longrightarrow W_0'\longrightarrow \RP^4. 
  \end{equation}
The simply connected manifold $\HP^2\#\HP^2$ has a unique spin structure, and
so the half-turn lifts to a spin automorphism.  Its square is either the
identity or the spin flip; we show it is the identity by computing at a fixed
point on~$S^8$.  The differential of the half-turn is the linear map~$-1$ on
the 8-dimensional tangent space.  The linear map~$-1$ lifts to the volume
form in~$\Spin_8$, which squares to~$+1$.  Therefore, the vertical tangent
bundle of~\eqref{eq:115} is spin, and so $w_i(W_0')$, $i=1,2$, are pulled
back from~$\RP^4$.  Using~\eqref{eq:107} we see that $W_0'$~is $\pinp$: it
admits two opposite $\pinp$ structures.

The cohomology ring of the connected sum is 
  \begin{equation}\label{eq:133}
     H^{\bullet }(\HP^2\#\HP^2;\ZZ)\cong
     \ZZ[x_1,x_2]/(x_1x_2,x_1^2-x_2^2),\qquad \deg x_1=\deg x_2=4. 
  \end{equation}
Let $t\in H^4(S^4;\ZZ)$ denote a positive generator.  Then under the
antipodal involution the class~$2t$ descends to the generator $\tc\in
H^4(\RP^4;\tZ)$; see Lemma~\ref{thm:34}.  Recalling~\eqref{eq:104} and the
fact that $w_4(\RP^4)=\alpha ^4$, as mentioned following~\eqref{eq:107}, we
deduce that $w_4(W_0')=\alpha ^4+\xbar_1+\xbar_2$, where $\xbar_i$~is the
mod~2 reduction of~$x_i$.  The class~$\xbar_1+\xbar_2$ is invariant under the
involution of~$\HP^2\#\HP^2$ and descends to $w_4$~of the vertical tangent
bundle in~\eqref{eq:115}.  Define the $w_1$-twisted integer lift~$\tc_0'\in
H^4(W_0';\tZ)$ of~$w_4(W_0')$ by
  \begin{equation}\label{eq:134}
     \pi ^*\tc_0'=2t+x_1-x_2, 
  \end{equation}
where $\pi \:\hW'_0\to W'_0$ is the orientation double cover.

\subsubsection{$W_0''$}\label{subsubsec:5.5.2} Let $K_{\RR }\to S^4=\HP^1$ be
the underlying real 4-plane bundle of the quaternionic line bundle $K\to S^4$
with $p_1^{\HH}(K)\in H^4(S^4;\ZZ)$ a positive generator.
Define~$W_0''=\PP(K_{\RR}^{\oplus 2}\oplus \underline{\RR})$ as the total
space of the real projective bundle
  \begin{equation}\label{eq:117}
     \RP^8\longrightarrow \PP(K_{\RR}^{\oplus 2}\oplus
     \underline{\RR})\xrightarrow{\;\;\rho \;\;} S^4.  
  \end{equation}
Let $L\to W_0''$ be the tautological real line bundle.  Since the stable
tangent bundle to~$S^4$ is trivial, the stable tangent bundle to~$W_0''$ is
the stable tangent bundle along the fibers, which is
  \begin{equation}\label{eq:118}
     (L-\underline{\RR}) + (\rho ^*K_{\RR}^{\oplus 2}\otimes L) .
  \end{equation}
This comes from the short exact sequence $0\to L\to \rho ^*(K_{\RR}^{\oplus
2}\oplus \underline{\RR})\to Q\to 0$ of real vector bundles over~$W_0''$
(compare~\eqref{eq:100}).  Using~\eqref{eq:118} we compute
$w_1(W_0'')=w_1(L)$ and $w_2(W_0'')=0$: it suffices to restrict to a fiber
of~\eqref{eq:117} since that restriction induces an isomorphism
on~$H^i(-;\zt)$, $i=1,2$.  The orientation double cover is an $S^8$-bundle
over~$S^4$, which is simply connected.  Hence $W_0''$~admits two opposite
$\pinp$ structures.
 
The bundle $\rho ^*K_{\RR}\otimes L$ has total Stiefel-Whitney class of the form
$1+w_4$, and it follows easily from the Whitney formula applied
to~\eqref{eq:118} that $w_4(W_0'')=0$.  
 
\subsubsection{$W_1$}\label{subsubsec:5.5.3} The projective group~$\PP
Sp_1\cong SO_3$ acts on~$\HP^2$ via (see~\eqref{eq:143} for notation)
  \begin{equation}\label{eq:144}
     \lambda \cdot [q^0,q^1,q^2] = [\lambda q^0,\lambda q^1\lambda
     q^2],\qquad \lambda \in Sp_1. 
  \end{equation}
So a principal $SO_3$-bundle has an associated fiber bundle with
fiber~$\HP^2$.  The action~\eqref{eq:144} lifts to the spin bundle of frames
of~$\HP^2$.  To see this, choose a basepoint~$[1,0,0]$ and write
$\HP^2=Sp_3/Sp_1\times Sp_2$.  The principal $\Spin_8$-bundle of frames is
associated to the principal $(Sp_1\times Sp_2)$-bundle $Sp_3\to\HP^2$ via the
representation
  \begin{equation}\label{eq:145}
     Sp_1\times Sp_2\cong \Spin_3\times \Spin_5\longrightarrow \Spin_8. 
  \end{equation}
The $\PP Sp_1$ action fixes the basepoint, and the ``diagonal'' map $Sp_1\to
Sp_1\times Sp_2\to \Spin_8$ descends to~$\PP Sp_1$.  The induced map $\PP
Sp_1\to\Spin_8$ gives the desired lift.  Define~$W_1$ as the fiber bundle
  \begin{equation}\label{eq:120}
     \HP^2\longrightarrow W_1\longrightarrow\CP^1\times \CP^1
  \end{equation}
obtained from the principal $SO_3$-bundle of oriented orthonormal frames of
the real 3-plane bundle $\mathcal{O}(1,1)_{\RR}\oplus
\underline{\RR}\to\CP^1\times \CP^1$, where $\mathcal{O}(1,1)\to\CP^1\times
\CP^1$ is the tensor product of the hyperplane line bundles on the factors.
The manifold~ $W_1$ is simply connected, hence orientable.  The stable
tangent bundle to~$\CP^1\times \CP^1$ is trivial, and the vertical tangent
bundle is spin, hence $W_1$~is spin with a unique spin structure refining
each orientation.

To compute the Pontrjagin classes of the vertical tangent bundle
of~\eqref{eq:120}, we use the $\PP Sp_1$~action to construct a fiber bundle 
  \begin{equation}\label{eq:147}
     \HP^2\longrightarrow E\longrightarrow \CP^{\infty} 
  \end{equation}
from the rank three real vector bundle $\mathcal{O}(1)_{\RR}\oplus
\underline{\RR}\to\CP^{\infty}$.  The squaring map $\TT\to\TT$ induces a
degree two map~$f$ on $B\TT=\CP^{\infty}$; the pullback
$f^*\bigl(\mathcal{O}(1)_{\RR}\oplus \underline{\RR} \bigr)\cong
\mathcal{O}(2)_{\RR}\oplus \und{\RR}$ is the adjoint bundle of a principal
$Sp_1\cong SU_2$-bundle we write as the quaternionic line bundle
$K\to\CP^{\infty}$.  Then the pullback of~\eqref{eq:147} under~$f$ is the
projectivization of the rank three quaternionic vector bundle $K^{\oplus
3}\to\CP^{\infty}$.  Let $a\in H^2(\CP^{\infty};\ZZ)$ be a generator; then
the quaternionic Pontrjagin class of $K\to\CP^{\infty}$ is
$p_1^{\HH}(K)=a^2$.  Let $L\to \PP(K^{\oplus 3})$ be the tautological
quaternionic line bundle.  Using the technique in~\S\ref{subsec:5.2},
including Remark~\ref{thm:35}, we compute the total Pontrjagin class of the
vertical tangent bundle to $E'=\PP(K^{\oplus 3})\to\CP^{\infty}$ as 
  \begin{equation}\label{eq:148}
     \begin{aligned} p(E'/\CP^{\infty}) &= \frac{\left[ 1+2(x+a^2) +
      (x-a^2)^2 \right] }{1+4x} \\[6pt] &= 1 \,+\, (2x+6a^2) \,+\,
     (7x^2-6a^2x) \,+\, (12a^2x^2) \,+\,\cdots\end{aligned} 
  \end{equation}
where $x=p_1^{\HH}(L)\in H^4(E';\ZZ)$.  Grothendieck's formula for projective
bundles\footnote{Let $V\to X$ be a vector bundle of rank~$r>0$ (over~$\RR$,
$\CC$, or~$\HH$), $\PP(V)\xrightarrow{\;p\;}X$ its projectivization, and
$L\to\PP(V)$ the tautological line bundle.  Over~$\PP(V)$ there is a short
exact sequence $0\to L\to p^*V\to Q\to 0$ of vector bundles, where
$Q\to\PP(V)$ has rank~$r-1$.  Grothendieck's formula expresses the
vanishing of its $r^{\textnormal{th}}$~Chern or Pontrjagin class.} implies
  \begin{equation}\label{eq:175}
     x^3=3a^2x^2 - 3a^4x + a^6. 
  \end{equation}
Use the pullback diagram
  \begin{equation}\label{eq:149}
     \begin{gathered} \xymatrix{E'\ar[r]^{\tilde f} \ar[d]_<<<<{\pi '} &
     E\ar[d]^<<<<{\pi } \\ \CP^{\infty}\ar[r]^{f} & \CP^{\infty}} 
     \end{gathered}
  \end{equation}
of fiber bundles to compute $\pi _*$~of the degree~12 Pontrjagin
classes~$p_3,p_1p_2,p_1^3$ of the vertical tangent bundle to
$E/\CP^{\infty}$; they pull back under~$\tilde f$ to the corresponding
Pontrjagin classes of~$E'/\CP^{\infty}$.  Note $f^*a=2a$.  Thus
  \begin{equation}\label{eq:150}
     f^*\pi _*p_3 = \pi '_*\tilde f^*p_3=\pi '_*(12a^2x^2) = 12a^2 
  \end{equation}
from which 
  \begin{align}\label{eq:151}
     \pi _*p_3&=3a^2.\\ 
\intertext{Similarly,} 
     \pi _*p_1p_2&=18a^2\\ 
     \pi _*p_1^3&=24a^2\\ 
\intertext{and hence}     
     \pi _*\lambda ^3&=3a^2.
  \end{align}
Finally, we pull back by the degree~$(1,1)$ map $\CP^1\times
\CP^1\to\CP^{\infty}$ to compute the corresponding quantities on~$W_1$.
After evaluating on the fundamental class~$[\CP^1\times \CP^1]$ we find 
  \begin{equation}\label{eq:152}
     \begin{aligned} \langle p_3(W_1)\,,\, [W_1] \rangle &= 6 \\ \langle
      p_1p_2(W_1)\,,\, [W_1] \rangle &= 36 \\ \langle p_1^3(W_1)\,,\, [W_1]
     \rangle &= 48 \\ \langle \lambda ^3(W_1), [W_1] \rangle &= 6 \,.
\end{aligned} 
  \end{equation}

   \section{The anomaly cancellation}\label{sec:6}
 
In this section we state the main computational result, Theorem~\ref{thm:31},
which provides generators for the dimension~12 bordism of manifolds which
occur in M-theory.  We give the proof in~\S\ref{sec:9}.  Here, in six
subsections, we use this bordism computation to prove the Anomaly
Cancellation Theorem~\ref{thm:1} by computing the invariants~$\af,\ac$ for a
generator of each factor.  We organize the presentation by Adams filtration
(see~\S\ref{sec:mmadams-spectr-sequ}).  Two of the six generators are
represented by spin manifolds, so in these cases the anomaly cancellation is
already proved in Theorem~\ref{thm:27}.  Nonetheless, we check directly by
computing the invariants.

Since our invariants take values in a finite abelian 2-group (see
Corollary~\ref{thm:9} and Corollary~\ref{thm:22}), it suffices to compute
after completing at the prime~2.  (The structure of the bordism group at odd
primes is much simpler, but we do not treat it here.)  Let $\Zlt$~denote the
2-adic numbers.

  \subsection{The bordism group}\label{subsec:6.1}

Recall that $\sB$~denotes the bordism spectrum of $\mc$-manifolds.  We use
the manifolds and cohomology classes defined in~\S\ref{sec:5}.  The
cohomology class~$\lambda $ of a spin manifold is the canonical integer lift
of~$w_4$.

  \begin{theorem}[]\label{thm:31}
 The following six $\mc$-manifolds generate the group $\pi _{12}\sB\otimes
\Zlt$:  
  \begin{equation}\label{eq:113}
   \begin{gathered}
   (W'_{0},\tilde c'_{0}),\quad
   (W''_{0},0),\quad
   (W_{1},\lambda) \\
   (K\times \hp^{2},\lambda),\quad
   (\RP^{4},\tilde c'_{\RP^{4}})\times B, \quad
   (\RP^{4}\#\RP^{4},0)\times B.
   \end{gathered}
  \end{equation}
  \end{theorem}

\noindent
 Note that $W_1$ and $K\times \HP^2$ are spin manifolds.

  \subsection{Computations}\label{subsec:6.2}

As explained in~\eqref{eq:25}, the Anomaly Cancellation Theorem~\ref{thm:1}
is a consequence of Theorem~\ref{thm:31} and the following.

  \begin{theorem}[]\label{thm:39}
 For each of the pairs~$(W,\tc)$ listed in~\eqref{eq:113} the anomaly
cancellation condition $\af(W)\ac(W)=1$ holds.  
  \end{theorem}

\noindent
 The proof of Theorem~\ref{thm:39} is divided into six parts, one for each
generator.  It occupies the remainder of~\S\ref{sec:6}.  Recall that~$\af$ is
defined in~\eqref{eq:38} and $\ac$ in~\eqref{eq:83}.  Strictly speaking, we
do not rely on the particular $w_1$-twisted integer lift~$\tc$ of~$w_4$ since
the mod~2 cubic invariant~$\ac$ is independent of the choice
(Lemma~\ref{thm:21}).  For convenience, we summarize the computations in the
following chart: 
 
  \begin{equation}\label{eq:113a}
     \begin{tabular}{ c@{\hspace{2.5em}} c@{\hspace{3.75em}} 
	c}
     \toprule $(W,\tc)$&$\af(W)$&$\ac(W)$ \\
     \midrule\\[-8pt] 
      $(W'_{0},\tilde c'_{0})$ &$+1$&$+1$ \\ [3pt]
      $(W''_{0},0)$ &$+1$&$+1$ \\ [3pt]
      $(W_{1},\lambda)$ &$+1$&$+1$ \\ [3pt]
      $(K\times \hp^{2},\lambda)$ &$-1$&$-1$ \\ [3pt]
      $(\RP^{4},\tilde c'_{\RP^{4}})$ &$+1$&$+1$  \\ [3pt]
      $(\RP^{4}\#\RP^{4},0)$ &$+1$&$+1$ \\ [3pt]
     \bottomrule \end{tabular} 
  \end{equation}

  \subsubsection{Adams filtration~4}\label{subsubsec:6.2.2}

To compute the Rarita-Schwinger anomaly partition function of $\RP^4\times
B$, we apply Proposition~\ref{thm:28} and Proposition~\ref{thm:30}:
  \begin{equation}\label{eq:126}
     \af(\RP^4\times B) = \tau _{\RP^4\times B}(T\RP^4 + TB -2)=\frac{\tau
     \mstrut _{\RP^4}(T\RP^4)\mlstrut ^{\ind D_B}\;\tau _{\RP^4}^{\ind
     D_B(TB)}}{\tau _{\RP^4}^{2\ind D_B}} . 
  \end{equation}
From~\eqref{eq:119} the stable tangent bundle to~$\RP^{4}$ is~$5L-1$.
Proposition~\ref{thm:29} implies that as far as $\eta $-invariants are
concerned, tensoring with~$L$ induces a change of sign.  Hence using
Proposition~\ref{thm:28} and~\eqref{eq:108}
  \begin{equation}\label{eq:124}
     \tau _{\RP^{4}}(T\RP^{4}) = \tau
     _{\RP^{4}}(5L-1)=\tau _{\RP^{4}}(-6)=\exp\left(-\frac{3\pi
     i}{4}\right).  
  \end{equation}
An alternative computation uses the 4-dimensional analog of
Proposition~\ref{thm:33} in which the denominator of~$2^8$ in~\eqref{eq:138}
is replaced by~$2^4$.  Bound the orientation double cover~$S^{4}$ by the
closed 5-ball~$D^{5}$ with its antipodal involution.  The pullback
of~$T\RP^{4}$ extends over~$D^{5}$ as $TD^{5}-1$.  The trace~$\tau _f$ at
the unique fixed point is $-5-1=-6$, and we recover~\eqref{eq:124}
from~\eqref{eq:138}.  Still another computation uses a variant of
Theorem~\ref{thm:37} for $\pinp$ 4-manifolds: replace~$\RP^{20}$
with~$\RP^{12}$ and $2^{11}$~in \eqref{eq:146} with~$2^7$.  Following the
notation in Remark~\ref{thm:38} with these replacements we compute
  \begin{equation}\label{eq:154}
     \begin{aligned} \gamma _*(T\RP^{4}) &= \gamma _*\bigl(\gamma
      ^*(5[H]-1) \bigr) \\ &= (5[H]-1)\,\gamma _*(1) \\
     &=(5[H]-1)\,8(1-[H]) 
      \\ &= -48(1-[H]),\end{aligned} 
  \end{equation}
and now \eqref{eq:124} follows from the adapted~\eqref{eq:146}.  For the
index computation on the Bott manifold we use the Atiyah-Singer index theorem
and~\eqref{eq:105}:
  \begin{equation}\label{eq:128}
     \ind D_B(TB) = \Ahat(B)\ch(TB)[B] = (1 -
     \frac{p_2}{1440})(8-\frac{p_2}{6})[B] = 248. 
  \end{equation}
Combining \eqref{eq:109}, \eqref{eq:126}, \eqref{eq:124}, \eqref{eq:128} and
$\ind D_B=1$ (see~\eqref{eq:106}) we conclude
  \begin{equation}\label{eq:170}
      \af(\RP^4\times B)=\frac{e^{-3\pi i/4}\;e^{\pi i}}{e^{\pi i/4}} = 1. 
  \end{equation}
 
The  $C$-field anomaly  partition function  is the  exponential of  the mod~2
reduction  of the  cubic  form~\eqref{eq:88}.   As in~\S\ref{subsec:6.2}  the
class~$\bp(\RP^4\times   B)=p(B)=-720b$.    Evaluating   on   the   generator
$\tc_{\RP^4}\in H^4(\RP^4,\tZ)$ we find
  \begin{equation}\label{eq:129}
     \left(\frac{\tc_{\RP^4}^3 - \bp\tc_{\RP^4}}{24}\right) [\RP^4\times B]= 30.
  \end{equation}
Since this is even, $\ac(\RP^4\times B)=1$.

  \subsubsection{Adams filtration~5}\label{subsubsec:6.2.1}  
 
The Rarita-Schwinger partition function only depends on the image
of~$(\RP^4\#\RP^4)\times B$ in $\pinp$ bordism.  The connected sum
$\RP^4\#\RP^4$ represents twice~$\RP^4$ in $\pinp$ bordism, and the same is
true after crossing with the Bott manifold.  (See~\S\ref{subsec:5.4}.)  Thus
we deduce from~\eqref{eq:170} that $\af\bigl((\RP^4\#\RP^4)\times B
\bigr)=\af(\RP^4\times B)^2=1$.  Since $w_4\bigl((\RP^4\#\RP^4)\times B
\bigr)=0$ we can take the $w_1$-twisted integer lift to vanish, and hence
$\ac\bigl((\RP^4\#\RP^4)\times B \bigr)=1$.

  \subsubsection{Adams filtration~3}\label{subsubsec:6.2.3}

The manifold $W_3=K\times \HP^2$ is spin, so by Proposition~\ref{thm:18} the
Rarita-Schwinger anomaly partition function is the mod~2 reduction of an
integer~$RS(W_3)$ defined in~\eqref{eq:78}.  Using~\eqref{eq:98}
and~\eqref{eq:103} we compute
  \begin{equation}\label{eq:130}
     \begin{aligned} \frac 12\Ahat(W_3)\ch(TW_3-2) &= \frac 12
      \Ahat(K)\Ahat(\HP^2)\bigl(\ch(TK)+\ch(T\HP^2)-2\bigr) \\ &=\frac 12
      (1+2k)(1-\frac{x}{12})\left( (4 - 48k) + (8 + 2x -\frac 56x^2)-2\right)
      \\ &= -kx^2 + \cdots\end{aligned} 
  \end{equation}
Therefore, $RS(W_3)=-1$ and $\af(W_3)=-1$. 
 
The $C$-field anomaly partition function is computed from the cubic form 
  \begin{equation}\label{eq:131}
     \frac{\lambda ^3 - p\lambda }{24}= \frac{(x-24k)^3
     -(3x^2-24xk)(x-24k)}{24} = kx^2 + \cdots
  \end{equation}
Evaluating on the fundamental class and exponentiating we deduce
$\ac(W_3)=-1$.

  \subsubsection{Adams filtration~1}\label{subsubsec:6.2.4}

The spin manifold~$W_1$ is defined in~\S\ref{subsubsec:5.5.3}; it is an
$\HP^2$-bundle over $\CP^1\times \CP^1$.  We plug~\eqref{eq:152}
into~\eqref{eq:142} to compute $\af(W_1)=1$ and into~\eqref{eq:88}
with~$\tc=\lambda $ to compute $\ac(W_1)=1$.

  \subsubsection{Adams filtration~0, part~1}\label{subsubsec:6.2.5}

The manifold~$W_0'$ is an $(\HP^2\#\HP^2)$-bundle over~$\RP^4$;
see~\eqref{eq:115}.  We claim that its Rarita-Schwinger anomaly partition
function is trivial: $\af(W_0')=1$.  To prove this we apply
Proposition~\ref{thm:33}.  The total space~\eqref{eq:114} of the orientation
double cover bounds $D^5\times (\HP^2\#\HP^2)$ with the antipodal involution
on~$D^5$ times the half-turn about an axis through~$S^8\subset \AA^9$ acting
on the connected sum.  There are two fixed points: the center of~$D^5$ times
antipodal points $p,p'\in S^8$.  In~\eqref{eq:138} the traces~$\tau _p=\tau
_{p'}$, and we claim $i_{p'}=-i_p$.  To prove this choose the center of~$S^8$
as the origin of~$\AA^9$, so identify the affine space~$\AA^9$ with the
vector space~$\RR^9$.  The half-turn is implemented on spinors by the element
$\omega =\gamma ^1\gamma ^2\cdots\gamma ^8$ in $\Clp9$, where we choose the
axis to be the last coordinate.  But the standard basis vectors
$e_1,e_2,\dots ,e_8$ form an \emph{oriented} basis of exactly one
of~$T_pS^8$, $T_{p'}S^8$; it is negatively oriented at the other point.  So
the action of frames at the other point is by the conjugate $e_1\omega
e_1\inv =-\omega $.  Multiply by the volume element of~$\Clp5$, which gives
the action of the involution on $\pinp$ frames at the center of~$D^5$.
 
For the $C$-field anomaly partition function we compute the cubic form on the
orientation double cover using~\eqref{eq:132}.  Use the $w_1$-twisted integer
lift~$\tc$ specified in~\eqref{eq:134} via its lift~$c$ to~$\hW'_0$.
Use~\eqref{eq:103} to compute that $p(\hW'_0)=3(x_1^2 + x_2^2)$.  Thus
  \begin{equation}\label{eq:141}
     \frac{c^3-pc}{48}= \frac{12tx^2 - 12tx^2}{48}=0,
  \end{equation}
since $x^2=x_1^2=x_2^2$ in~$H^8(\HP^2\#\HP^2)$; see~\eqref{eq:133}.
Therefore, $\ac(W_0')=1$.

  \subsubsection{Adams filtration~0, part~2}\label{subsubsec:6.2.6}
 
The manifold $W_0''=\PP(2K_{\RR}\oplus \und{\RR})$, defined in~\eqref{eq:117}
as an $\RP^8$-bundle over~$S^4$.  Recall that $K\to S^4$ is the quaternionic
line bundle with $p_1^{\HH}(K)\in H^4(S^4;\ZZ)$ the positive generator.  We
use Theorem~\ref{thm:37} to compute $\af(W_0'')=1$.
 
As a preliminary define $\mu \in \tKO^4(S^4)$ as $\iota _*(1)$ for $\iota
\:\pt\to S^4$, and $\lambda \in \tKO^0(S^4)$ as the $KO$-class $[K_{\RR}]-4$,
where $K_{\RR}\to S^4$ is the real 4-plane bundle underlying $K\to S^4$.
Identify $\tKO^4(S^4)$ as the Grothendieck group of quaternionic vector
bundles over~$S^4$ of virtual rank zero; then $\mu =[K]-[\und\HH]$.  Let $\pi
\:S^4\to \pt$ be the unique map.  Then we claim
  \begin{align}
     \lambda \mu &=0 \label{eq:155}\\
     \lambda [\und\HH]&=4\mu \label{eq:156}\\
     \pi _*\mu &=1 \label{eq:157}\\
     \pi _*[\und\HH]&=0. \label{eq:158}
  \end{align}
Since we can trivialize $K_{\RR}\to S^4$ away from a point, we can arrange
representatives of~$\lambda ,\mu $ with disjoint support, from which
\eqref{eq:155} follows.  For~\eqref{eq:156} we observe that if $M$~is any
quaternionic line, then there is a natural isomorphism 
  \begin{equation}\label{eq:159}
     \begin{aligned} M^{\oplus 4}&\longrightarrow M_{\RR}\otimes _{\RR}\HH \\
      (\xi _1,\xi _2,\xi _3,\xi _4) &\longmapsto \sum\limits_{s=1}^4-\xi
      _s\otimes 1 \,+\, \xi _si\otimes i\,+\,\xi _sj\otimes j \,+\, \xi
     _sk\otimes k\end{aligned} 
  \end{equation}
of quaternionic vector spaces.  Equation~\eqref{eq:157} is immediate: $\pi
_*\mu =(\pi \circ \iota )_*(1)=1$.  Finally, $\und\HH\to S^4$ is pulled back
from $\und\HH\to\pt$, so $\pi _*[\und\HH]=[\und\HH]\pi _*(1)=0$ since $1\in
KO^0(S^4)$ extends over the 5-ball.
 
Another preliminary: If $M_{\RR}\to Y$ is the real 4-plane bundle underlying
a quaternionic line bundle $M\to Y$, then its $KO$-theory Euler class is 
  \begin{equation}\label{eq:160}
     [\und\HH] - [M]\in \tKO^4(Y). 
  \end{equation}
Proof: $M\to Y$~ is associated to a principal $Sp_1$-bundle via (i)~the
embedding $Sp_1\hookrightarrow \Spin_4\cong Sp_1\times Sp_1$ onto the second
factor and (ii)~ the action of~$Sp_1\times Sp_1$ on~$\HH$ in which the first
factor acts trivially and the second by right multiplication.  Then the
$KO$-Euler class is pulled back from the vector bundle associated to the
difference of the quaternionic half-spin representations.
 
Let $J\to S^4$ be the quaternionic line bundle with
$p_1^{\HH}(J)=-2p_1^{\HH}(K)$.  Then $K_{\RR}^{\oplus 2}\oplus J_{\RR}\to
S^4$ is trivializable.  Define 
  \begin{equation}\label{eq:161}
     \gamma \:W_0''=\PP(K_{\RR}^{\oplus 2}\oplus\und\RR)
     \xrightarrow{\;\;i\;\;}\PP(K_{\RR}^{\oplus 2}\oplus\und\RR\oplus
     J_{\RR})\cong S^4\times \RP^{12}\xrightarrow{\;\;\pi
     \;\;}\RP^{12}\xrightarrow{\;\;j\;\;}\RP^{20}, 
  \end{equation}
where $\pi $~is the projection onto the second factor and $j$~is a linear
embedding as in Remark~\ref{thm:38}.  Let $L\to \PP(K_{\RR}^{\oplus 2}\oplus
\und\RR\oplus J_{\RR})$ be the tautological real line bundle.  Then $L\cong
\pi ^*H=\pi ^*j^*H$ for $H\to\RP^{20}$ the tautological line bundle, and
$i^*L$~is also isomorphic to the tautological line bundle.  We identify
$L^*\cong L$.  The normal bundle to~$i$ is the quotient of tangent bundles
(see~\eqref{eq:118}):
  \begin{equation}\label{eq:162}
    \begin{aligned}
     \bigl[(i^*L-\und\RR)\;\oplus\; K_\RR^{\oplus 2}\otimes i^*L \;\oplus\;
     J_{\RR}\otimes i^*L \bigr] \bigm/ \bigl[(i^*L-\und\RR)\;\oplus\; K_\RR^{\oplus
     2}\otimes i^*L \bigr] &\cong  J_{\RR}\otimes i^*L \\&\cong 
     i^*(J_{\RR}\otimes L) . 
   \end{aligned}
  \end{equation}
There is a canonical section of 
  \begin{equation}\label{eq:163}
     J_{\RR}\otimes L\cong \Hom(L,J_{\RR})\longrightarrow \PP(K_{\RR}^{\oplus
     2}\oplus\und\RR\oplus J_{\RR}) 
  \end{equation}
given by projection $K_{\RR}^{\oplus 2}\oplus\und\RR\oplus J_{\RR}\to
J_{\RR}$, and its zero set is the image of~$i$.  It follows that $i_*(1)$~is
the $KO$-Euler class of~\eqref{eq:163}, which we compute
using~\eqref{eq:160}:
  \begin{equation}\label{eq:164}
     \begin{aligned} i_*(1) = [\und\HH]-[J\otimes L] &=(1-[L])[\und\HH] +
      2[L]\mu \\ &=\pi ^*\bigl\{ (1-[H])[\und\HH] + 2[H]\mu \bigr\}
      .\end{aligned} 
  \end{equation}
Using~\eqref{eq:118} we find 
  \begin{equation}\label{eq:165}
     [TW_0''] - 2 = i ^*\bigl\{2[H]\lambda +9[H]+1\bigr\}. 
  \end{equation}
Combining~\eqref{eq:164} and~\eqref{eq:165} with
\eqref{eq:155}--\eqref{eq:158} we calculate 
  \begin{equation}\label{eq:166}
     \begin{aligned} \pi _*i_*([TW_0'']-2) &= \pi _*\Bigl\{ \bigl(2[H]\lambda
      +9[H]+1 \bigr)\bigl((1-[H])[\und\HH] + 2[H]\mu \bigr) \Bigr\} \\ &=
      10[H]+10.\end{aligned} 
  \end{equation}
Now $j_*(1)=8(1-[H])$ is computed in Remark~\ref{thm:38}, and so 
  \begin{equation}\label{eq:167}
     j_*\bigl([H] \bigr) = [H]j_*(1)=8[H](1-[H])=-j_*(1), 
  \end{equation}
from which 
  \begin{equation}\label{eq:168}
     \gamma _*([TW_0'']-2)=j_*\pi _*i_*([TW_0'']-2)=0. 
  \end{equation}
Then $\af(W_0'')=1$ follows immediately from~\eqref{eq:146}. 
 
The $C$-field anomaly is also trivial---$\ac(W_0'')=1$---since $w_4(W_0'')=0$
and we can choose the $w_1$-twisted integer lift~$\tc$ in~\eqref{eq:83} to be
zero.

   \section{Ambiguities in the M-theory action}\label{sec:8}
 
As mentioned in the introduction, to define an M-theory action it is not
sufficient to demonstrate the cancellation of anomalies; we must also give a
trivialization of the product ~$\af\otimes \ac$, a so-called \emph{setting of
the quantum integrand}.  The ratio~$\beta $ of two trivializations is an
invertible 11-dimensional field theory.  Unitarity of M-theory requires that
$\beta $~be reflection positive.  If $\beta $~ were to depend on the metric
or the field strength of the $C$-field, then it would be detected locally.
Since the local physics is fixed by considerations other than anomaly
cancellation, we restrict $\beta $~to be a \emph{topological} field theory.
As explained in~\S\ref{sec:2} a reflection positive invertible 11-dimensional
topological field theory of $\mc$-manifolds is determined by a
homomorphism $\pi _{11}\sB\to\CZ$.  (Reflection positivity imposes a
restriction, which is satisfied here by all such homomorphisms, since they
take values in~$\pmo\subset \CZ$.)  The following conjecture describes the group
of these theories.
 
Let $\Sigma $~be the Klein bottle.  It has four $\pinp$ structures of which
two are nonbounding~\cite[Proposition~3.9]{KT2}; fix one of those.  Also, let
$S^1$~denote the circle with its nonbounding string structure; see
Remark~\ref{thm:44}.  Define the 11-manifold
  \begin{equation}\label{eq:177}
     N=\cir\times \Sigma \times B, 
  \end{equation}
where $B$~is the Bott manifold (\S\ref{subsec:5.3}).  The following is based
on computations to appear in \cite{GH}, and out of an abundance of caution we
state it here as a conjecture, a more precise version of
Conjecture~\ref{thm:2}. 

  \begin{conjecture}[]\label{thm:41}
 The group $\pi _{11}\sB$~is cyclic of order~2.  The bordism class of the
pair~$(N,0)$ represents the generator.  The mod~2 index of the $\pinp$ Dirac
operator is an isomorphism $\pi _{11}\sB\to\zt$. 
  \end{conjecture}

\noindent 
 See~\S\ref{sec:mmdimension-11} for a justification of
Conjecture~\ref{thm:41} using the Adams spectral sequence.

  \begin{remark}[]\label{thm:43}
 Index invariants of $\pinp$ $n$-manifolds correspond to index invariants of
spin $(n-1)$-manifolds; see~\cite[\S9.2.3]{FH}.  Hence the mod~2 indices of
spin manifolds in dimensions~$9,10$ correspond to mod~2 indices of $\pinp$
manifolds in dimensions~$10,11$.  Let $P$~be a $\pinp$ 10-manifold.  Then the
mod~2 index of the product~$S^1\times P^{10}$ equals the mod~2 index of~$P$.
We use product formulas analogous to Proposition~\ref{thm:30} to compute that
the mod~2 index of $\Sigma \times B$ is nonzero.
  \end{remark}

   \section{The bordism computation}\label{sec:9}
 
In this section we present the computations which prove Theorem~\ref{thm:31}
and justify Conjecture~\ref{thm:41}.  We begin in~\S\ref{subsec:9.1} by
constructing the Thom spectrum~$\twistedThom$.  In~\S\ref{sec:mmchar-class}
we discuss some characteristic classes of $\mc$-manifolds and their behavior
under transfer maps from the orientation double cover.  We compute the values
of some $\mc$-characteristic classes on two special manifolds in~
\S\ref{sec:mmtwo-examples}.  The Adams spectral sequence is introduced
in~\S\ref{sec:mmadams-spectr-sequ}.  A crucial input is the structure of the
mod~2 cohomology of $\twistedB$, which is discussed in Appendix~\ref{sec:10}.
The main work in this section occurs in~\S\ref{sec:mmcomputations}.  We
present arguments to determine the facts we need about $\mc$-bordism groups
in dimensions~11 and~12, and along the way compute low dimensional
$\mc$-bordism groups.

\subsection{The Thom spectrum}\label{subsec:9.1}
 Our aim in this section is to justify the
claim that the manifolds listed in Theorem~\ref{thm:31} generate the
bordism group of $\mathfrak m_{c}$-manifolds (Definition~\ref{thm:19}).  We
begin by identifying the relevant Thom spectrum.

Suppose that $(X,\zeta)$ is a space $X$ equipped with a stable vector
bundle $\zeta$ of virtual dimension~ $0$, which one may think of as a
map $\zeta:X\to BO$ from $X$ to the classifying space of the infinite
orthogonal group.  Write $\thom(X,\zeta)$ for the Thom spectrum of
$\zeta$.  The homotopy group $\pi_{m}\thom(X,\zeta)$ is the bordism
group of triples $(M,f,\phi)$ consisting of an $m$-manifold $M$
equipped with a map $f:M\to X$, and an isomorphism
\[
\pi:f^{\ast} \zeta \approx \R^{m}-TM
\]
of virtual vector bundles.  Put more colloquially it is the bordism
group of manifolds whose stable {\em normal} bundle has an
$(X,\zeta)$-structure.  The bordism group of manifolds whose stable
{\em tangent} bundle has an $(X,\zeta)$-structure is the homotopy
group $\pi_{n}\thom(X,-\zeta)$.

We are interested in manifolds $M$ whose 
whose stable {\em tangent} bundle has a $\pinp$ structure and which are
equipped with a $w_{1}$-twisted integer lift of $w_{4}$.   We
therefore consider the space $\twistedB$ defined by the homotopy pullback
square 
\begin{equation}
\label{eq:mm3}
\xymatrix@C=4em{
\twistedB  \ar[r]\ar[d]  &    E\Z/2\underset{\Z/2}{\times}K(\Z,4)
\ar[d] \\
\bpinplus  \ar[r]_-{(w_{1}, w_{4})}        & B\Z/2 \times
K(\Z/2,4)\mathrlap{,}
}
\end{equation}
let $\zeta$ be the virtual vector bundle classified by 
the pullback 
\[
\twistedB\to \bpinplus \to BO,
\]
and write $\twistedThom=\thom(\twistedB,-\zeta)$.  The homotopy groups
$\pi_{\ast}\twistedThom$ are
then the bordism groups of manifolds equipped with a tangential
$\pinPlus$-structure and a $w_1$-twisted integer lift of $w_{4}$.   

As in Definition~\ref{thm:19}, an \emph{$\emcee$-manifold} is a pair $(M,c)$
in which $M$ is a $\pinPlus$-manifold and $c$ is a $w_{1}$-twisted integer
lift of tangential $w_{4}$.  A $\Spin$-manifold $M$ gives rise to an
$\emcee$-manifold by taking $c$ to be the tangential characteristic class
$\lambda$.

  \begin{remark}[]\label{thm:44}
 Given two $\emcee$-manifolds $(M_{1},c_{1})$ and $(M_{1},c_{2})$ it is
tempting to imagine that the product $(M_{1}\times M_{2},
c_{1}+c_{1})$ is an $\emcee$-manifold.  While it is true that
$w_{4}(M_{1}\times M_{2})=w_{4}(M_{1})+w_{4}(M_{2})$, the sum
$c_{1}+c_{2}$ doesn't really make sense as the two summands lie in
different twisted cohomology groups.  The expression does make sense if
$w_{1}(M_{2})=0$ and $c_{2}=0$, and in particular, if  $M_{2}$ is a
$\Spin$-manifold equipped with a trivialization of $\lambda$ (a
$\String$-manifold).   The ``Bott manifold'' $B$ of \S\ref{subsec:5.3} is a
$\String$-manifold so if $(N,c)$ is an $\emcee$-manifold, then
$(N\times B,\pi_{1}^{\ast}c)$ is also an $\emcee$-manifold.
  \end{remark}

\subsection{Characteristic classes}
\label{sec:mmchar-class}

To describe  bordism invariants of $\emcee$-manifolds we will
require some cohomology classes in $\twistedB$.  First note that under
the equivalence
\[
B\!\OO_{1}\times B\!\SO \xrightarrow{\oplus}{} B\!\OO
\]
the characteristic class $w_{2}$ pulls back to $(0,w_{2})$.   Passing
to the homotopy fiber of the classifying map to $K(\Z/2,2)$ gives an
equivalence 
\[
B\!\OO_{1}\times \bspin\to \bpinplus.
\]
From this one sees that a $\pinPlus$-structure on a vector bundle
$T$ may be identified with an equivalence $T\approx L\oplus V$ of
stable vector bundles, in which $L$ is a line bundle and $V$ is a
$\Spin$-bundle.

Suppose $M$ is a $\pinPlus$ manifold and, using the above, regard the
$\pinPlus$-structure as giving a stable isomorphism $TM\approx L\oplus
V$ with $L$ a real line bundle and $V$ a $\Spin$-bundle.   Set
\begin{align*}
\alpha &= w_{1}(TM) = w_{1}(L), \\
w_{i} &=w_{i}(TM) = w_{i}(V),\mathrlap{\qquad 1 < i \le 4,}
\end{align*}
and, as in \S\ref{subsec:4.2}, write 
\[
\lambda=\lambda(V)
\]
for the
characteristic class of $\Spin$-bundles, twice which is $p_{1}$.  The
mod $2$ reduction of $\lambda$ is $w_{4}$, so every
$\pinPlus$-manifold has an untwisted integer lift of $w_{4}$.

Now suppose that $(M,c)$ is an $\emcee$-manifold.  The total space of the
orientation double cover $\pi:\widehat{M}\to M$ is a $\Spin$-manifold, and in
fact $T\widehat{M}$ is equipped with a stable isomorphism
$T\widehat{M}\approx \pi^{\ast}V$.  The class $\chat=\pi^{\ast} c$ 
is an
untwisted integer lift of $w_{4}(T\widehat{M})$.  This specifies a class
\[
\iota\in
H^{4}(\widehat{M};\Z)
\]
satisfying
\begin{equation}
\label{eq:mm4}
2 \iota = \lambda-\chat
\end{equation}

\begin{remark}
\label{rem:6} The fact that $\iota$ is specified uniquely and not just
up to elements of order two follows from the fact that both $\lambda$
and $\chat$ are {\em integer lifts} of $w_{4}$.  The integer lifts of a
fixed mod $2$ cohomology class of dimension $k$ form a torsor for
integer cohomology in dimension $k$, under the action in which an
integer cohomology class $\iota$ changes an integer lift $c$ to
$c+2\iota$.  See~\cite[\S2.5]{HS} for a more
systematic discussion of this from the point of view of cocycles.  It
is also not difficult to show (for example using~\eqref{eq:mm3}) that
the classifying space $\bspin \langle \beta w_{4} \rangle$ for $\Spin$
bundles with an integer lift of $w_{4}$ is homotopy equivalent to
$\bspin \times K(\Z,4)$ and in particular has torsion free $H^{4}$.   So
in fact the equation~\eqref{eq:mm4} specifies $\iota$ uniquely as a
cohomology class in $\bspin \langle \beta w_{4} \rangle$.
\end{remark}

\subsubsection{Transfer}
\label{sec:mmtransfer}

We will make use of the additive and multiplicative transfers 
\begin{align*}
\tr:H^{k}(\widehat{M};\Z/2) &\to H^{k}(M;\Z/2) \\
P:H^{k}(\widehat{M};\Z/2) &\to H^{2k}(M;\Z/2).
\end{align*}
Most computations of the additive transfer can be made in terms of
$\widehat{M}$: for $y\in H^{\ast}M$ one has
\[
\int_{M}\tr(x)\, y = \int_{\widehat{M}} (x\, \pi^{\ast}y).
\]
Computing the map $P$ can be a little tricky, however there are some
useful methods in special cases.  For one thing, the map $P$ is
quadratic:
\[
P(x+y) = P(x)+P(y) +\tr(x\, \tau(y)),
\]
where $\tau$ is the cohomology homomorphism induced by the involution.
One can also compute $P(x)$ in terms of characteristic classes, when
$x$ itself is a characteristic class.  The following, which is all we need,
is a special case of \cite[Theorem~1.1]{K}.

\begin{lem}
\label{thm:mm6}
Suppose that $p:\widehat X\to X$ is a double cover and $W$ is a
$\Spin$-vector bundle  on $\widehat X$ of dimension $d$.   If $x=w_{4}(W)\in H^{4}(\widehat
X;\Z/2)$ then 
\[
P(x)=w_{8}(p_{\ast}W-p_{\ast}\R^{d})\in H^{8}(X;\Z/2)
\]
where, for a vector bundle $V$, 
$p_{\ast}V=\widehat X\underset{\Z/2}{\times}(V\oplus V)$.  \qed
\end{lem}

\begin{remark}
\label{rem:1}
In the situation of Lemma~\ref{thm:mm6} if $W=p^{\ast}U$ for some vector
bundle $U$ on $X$ then $p_{\ast}W=U\oplus (U\otimes L)$ where $L=\widehat
X\underset{\Z/2}{\times}\R$, with $\Z/2$ is acting on $\R$  by the
sign representation.   In that case $P(x)= w_{8}(U+ U\otimes L -
L^{\oplus d})$. 
\end{remark}

\subsection{Two examples}
\label{sec:mmtwo-examples}

Two characteristic classes play an important role in our computation
of $\pi_{12}\twistedThom$.   They are $\tr(\iota^{3}+\iota^{2}w_{4})$ and $\alpha^{4}P(\iota)$.

\begin{eg}
\label{eg:1}
Recall from \S\ref{subsubsec:5.5.1} the pair $(W'_{0},\tilde c'_{0})$ in which 
\[
W'_{0} = S^{4}\underset{\Z/2}{\times}\hp^{2}\#\hp^{2}.
\]
The orientation double cover is 
\[
\widehat{W}'_{0} = S^{4}\times \hp^{2}\#\hp^{2}\xrightarrow{\pi}{} 
S^{4}\underset{\Z/2}{\times}\hp^{2}\#\hp^{2}
\]
and the involution of $\hp^{2}\#\hp^{2}$ exchanges the two summands,
is orientation preserving, and has two fixed points.  The $w_1$-twisted
cohomology class $\tilde c'_{0}$ satisfies
\[
\pi^{\ast}\tilde c'_{0} = 2 t + x_{1}-x_{2}.  
\]
Since $\lambda = x_{1}+x_{2}$ (see~\eqref{eq:104}) we have 
\begin{align*}
\iota &= \frac12(\lambda-\tilde \pi ^*c'_{0}) \\ &= x_{2}-t \\
\iota^{3}+\iota^{2} w_{4} & = x^{2}_{2}\,t 
\end{align*}
and so 
\[
\int_{W'_{0}} \tr(\iota^{3}+\iota^{2} w_{4} ) = 
\int_{\widehat{W}'_{0}} (\iota^{3}+\iota^{2} w_{4} ) = 
\int_{\widehat{W}'_{0}} (x_{2}^{2}\,t)=1.
\]
\end{eg}

\begin{remark}
\label{rem:2}
One can check that $\int_{W'_{0}}\alpha^{4}P(\iota)=0$, though we will
not make use of this fact. 
\end{remark}

\begin{eg}
\label{eg:3}
Consider the manifold $W''_{0}=P(K_{\R}^{\oplus 2}\oplus \R)$ described
in \S\ref{subsubsec:5.5.2}.   Since the Stiefel-Whitney classes of 
$K_{\R}^{\oplus 2}$ vanish on $S^{4}$, the projective bundle formula gives
\[
H^{\ast}(W''_{0};\Z/2) = \Z/2[t,\alpha]/(t^{2},\alpha^{9})
\]
where $t$ is the generator of $H^{4}(S^{4})$.  The orientation double
cover is $S(K_{\R}^{\oplus 2}\oplus \R)$ and the mod $2$ reduction of
\[
\iota=\frac12(\lambda-c)
\]
is $t=w_{4}(K_{\R})$.   Since the map 
\[
H^{4}(S^{4})\to H^{4}(S(K_{\R}^{\oplus 2}\oplus \R))
\]
is an isomorphism, we have 
\[
\iota^{3}+\iota^{2}w_{4}=0,
\]
and so 
\[
\int_{W''_{0}}\tr(\iota^{3}+\iota^{2}w_{4}) = 0.
\]

For the characteristic number $\alpha^{4}P(\iota)$ we first appeal to
Lemma~\ref{thm:mm6} and compute
\[
P(\iota)=w_{8}(p_{\ast}K_{\R} - p_{\ast}\R^{4}).
\]
Since $K_{\R}$ is pulled back from $W''_{0}$ we are in the situation
of Remark~\ref{rem:1}.   Writing $L_{\alpha}$ for the real line bundle with
$w_{1}(L_{\alpha})=\alpha$ we are led to the total Stiefel-Whitney
class 
\begin{align*}
w(K_{\R}+K_{\R}\otimes L_{\alpha} - 4 L_{\alpha}) &= 
(1+t)(1+t+\alpha^{4})(1+\alpha)^{-4}\\ & = 1+t^{2} + \alpha^{4}t + O[9]
= 1+\alpha^{4}t + O[9].
\end{align*}
Thus
\[
P(\iota) = \alpha^{4}t
\]
and 
\[
\int_{W''_{0}}\alpha^{8}P(\iota) = 1.
\]
\end{eg}

\subsection{The Adams spectral sequence}
\label{sec:mmadams-spectr-sequ}

Our aim is to identify generators for $\pi_{12}\twistedThom$.  Since our main
concern is the comparison of two different homomorphisms from
$\pi_{12}\twistedThom$ to a finite abelian $2$-group, it suffices to do so
after completing at $2$.  For this we can appeal to the Adams spectral
sequence, and this can be done by computer calculation.  For the purposes of
this paper the authors used {\em Mathematica} to determine the mod $2$
cohomology of $\twistedB$ as a module over the mod $2$ Steenrod algebra, and
Rob Bruner's program \cite{Br} for computing the $E_{2}$-term of the Adams
spectral sequence, as well as the map of Adams spectral sequences induced by
the map from $\mspin$ to $\twistedThom$.  The results are displayed in
Figures~\ref{fig:mm1} and~\ref{fig:mm4}.  Appendix~\ref{sec:10} contains much
more information about the cohomology of~$\twistedB$.  In~\cite{GH} a more
detailed version of this computation is described, as well as means of doing
it by hand.

The Adams spectral sequence begins with 
\[
E_{2}^{s,t} = \ext^{s,t}_{\mathbf A}(H^{\ast}\twistedThom,\Z/2)
=\ext^{s}_{\mathbf A}(H^{\ast+t}\twistedThom,\Z/2),
\]
with $\mathbf A$ the mod $2$ Steenrod algebra, and converges to the homotopy
groups
of the $2$-adic completion $\widehat\twistedThom$ of $\twistedThom$.
Since the homology groups of $\twistedB$ are each finitely generated,
the homotopy groups of $\twistedThom$ are each finitely generated and
so~$\pi_{t-s}\widehat{\twistedThom}$ is just the $2$-adic completion of
$\pi_{t-s}\twistedThom$.

\begin{remark}
\label{rem:3} 
For the remainder of this section all bordism groups
will be $2$-adically completed.  Except for the appearance of the
symbol $\Z_{2}$ for the $2$-adic numbers, this will not be indicated
in the notation.
\end{remark}

The differential $d_{r}$ of the Adams spectral sequence goes from
bidegree $(s,t)$ to bidegree $(s+r, t+r-1)$.  It is customary to
display the Adams spectral sequence with the horizontal axis numbered
by $(t-s)$ and the vertical axis $s$.  With this convention the
differential $d_{r}$ goes one square to the left and $r$-squares
upward.  The groups contributing to a given homotopy group lie in a
column.

The ``$s$'' in the Adams spectral sequence direction 
corresponds to a decreasing filtration of stable homotopy
groups known as the {\em Adams filtration}.

\begin{definition}
\label{def:1} 
A map $f:X\to Y$ of spectra has (mod $2$) {\em Adams filtration greater
than or equal to $n$} if there is a factorization
\[
X=X_{0} \xrightarrow{f_{0}}{} X_{1}\xrightarrow{f_{1}}{} \cdots \to X_{n}=Y
\]
in which each $f_{i}$ induces the zero map in mod $2$ cohomology.  
\end{definition}

The Adams filtration is natural in both variables, in the sense that
composition with a map $X'\to X$ or $Y\to Y'$ sends maps of Adams
filtration greater than or equal to $n$ to maps of Adams filtration
greater than or equal to $n$.

\begin{defin}
\label{def:2}
A map has \emph{Adams filtration $n$} if it has Adams
filtration greater than or equal to $n$ but not greater than or equal
to $(n+1)$.
\end{defin}

The maps in Adams filtration greater than or equal to $n$ appear in the Adams
charts with $s$-coordinate greater than or equal to $n$.

The Adams spectral sequence for $\pi_{\ast}X$ is a module over the
Adams spectral sequence for the stable homotopy groups
$\pi_{\ast}S^{0}$ of spheres.   The element $2\in\pi_{0}S^{0}$ has
Adams filtration $1$ and is represented by a class traditionally denoted
\[
h_{0}\in\ext^{1,1}_{\mathbf A}(\Z/2,\Z/2).
\]
Multiplication by $h_{0}$ in any Adams spectral sequence is indicated
in the chart by a vertical line.   Similarly, the Hopf map $\eta\in
\pi_{1}S^{0}$ is represented by the element 
\[
h_{1}\in\ext^{2,1}_{\mathbf A}(\Z/2,\Z/2)
\]
and multiplication by $h_{1}$ is indicated by a $(1,1)$ diagonal line.  A
little care must be used in drawing conclusions from these notations.
For example, Figure~\ref{fig:mm1} shows the Adams spectral sequence for
$\pi_{\ast}\twistedThom$.  From the chart it appears that $\pi_{4}\twistedThom =
\Z/8\oplus \Z_{2}$, with the generator of the $\Z/8$ appearing in
Adams filtration $1$.  However all that the chart implies is that $8$
times the apparent generator in filtration $1$ has Adams filtration
greater than $4$.  Some additional argument is needed to conclude that
there is an element of order $8$ in filtration $1$.  One can conclude
from the chart that $\pi_{4}$ is generated by $2$ elements and has
rank $1$.  The computation of $\pi_{4}\twistedThom$ will be given in detail
in \S\ref{sec:mmdimension-4}.

The $s=0$ line of the Adams spectral sequence consists of the groups
\[
\hom_{\mathbf A}(H^{t}X, H^{0}S^{0})\subset H_{t}X.
\]
The kernel of the higher differentials pick out the image of the
Hurewicz homomorphism in $H_{t}X$.  When $X=\thom(B,V)$ is a Thom
spectrum it is often useful to label an element $x\in
E_{2}^{0,t}=\hom_{\mathbf A}(H^{t}X, H^{0}S^{0})$ with a
cohomology class $\beta\in H^{t}(B)$ whose image under
\begin{equation}
\label{eq:mm1}
H^{t}(B) \xrightarrow{\text{Thom iso}}{} H^{t}(X)
\xrightarrow{x^{\ast}}{} H^{t}(S^{t}) = \Z/2
\end{equation}
is non-zero.  This can be a little perilous as there can be many
cohomology classes having a non-zero image under given class, and some
care must be taken to ensure that the labeled cohomology classes are
linearly independent on the image of the Hurewicz homomorphism.  
In the end it provides useful information.  If $x$ survives the Adams
spectral sequence and is represented by a manifold $M$, the image of
$\beta$ under~\eqref{eq:mm1} is
\[
\int_{M}\beta.
\]
Such labels therefore provide a means of identifying specific
manifolds as representing a basis of the image of the Hurewicz
homomorphism.  The class $\beta$ is a characteristic class of some
kind.

\subsection{Computations}
\label{sec:mmcomputations}

Armed with these spectral sequences, we first turn to the computation of
$\pi_{\ast}\twistedThom$ in low degrees and then proceed to $\pi_{12}\twistedThom$.  We remind the reader of
Remark~\ref{rem:3}, that all homotopy groups have been $2$-adically
completed.

\begin{figure}
\includegraphics[scale=2]{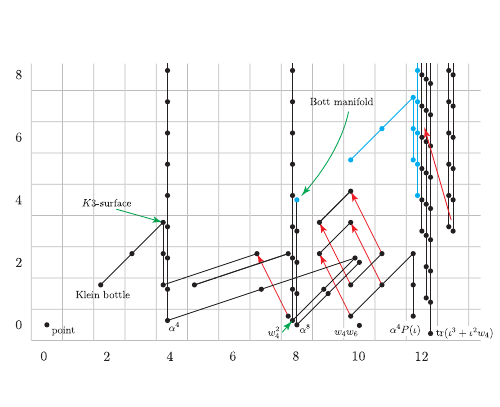}
  \vskip -1.5pc
\caption{The Adams spectral sequence for $\pi_{\ast}\twistedThom$}
\label{fig:mm1}
\end{figure}

\subsubsection{Dimension less than $4$}
\label{sec:mmdimension-less-than}

The homotopy fiber of the map $\twistedB  \to \bpinplus$ is the
Eilenberg-MacLane space $K(\Z,4)$.   It follows easily from this that
the map 
\[
\pi_{\ast}\twistedThom\to \pi_{\ast}\mtPinPlus 
\] 
is an isomorphism for $\ast<4$ and an epimorphism when $\ast=4$.
From~\cite{KT1} one concludes that 
\begin{align*}
\pi_{0}\twistedThom &= \Z/2 \mathrlap{\qquad\text{generated by a point}} \\
\pi_{1}\twistedThom &= 0 \\
\pi_{2}\twistedThom & = \Z/2\mathrlap{\qquad\text{generated by }(\Sigma,0)} \\
\intertext{where $\Sigma$ is a Klein bottle in a nonbounding $\pinPlus$-structure, and}
\pi_{3}\twistedThom &=\Z/2 \mathrlap{\qquad\text{generated by }S^{1}\times \Sigma,}
\end{align*}
where $S^{1}$ is given the non-bounding $\String$-structure (on
$S^{1}$ a $\String$ structure is equivalent to a $\Spin$-structure).

\subsubsection{Dimension $4$}
\label{sec:mmdimension-4}  
We define a homomorphism $e:\pi_{4}\twistedThom\to
\Z_{2}$ by
\[
e(M,c) = \int_{M}c.
\]
Forgetting the twisted lift of $w_{4}$ gives a map 
\[
u:\pi_{\ast}\twistedThom\to \pi_{\ast}\mtPinPlus.
\]
By~\cite{KT1}, the group $\pi_{4}\mtPinPlus $ is cyclic of order $16$, with
generator $\RP^{4}$.   Combined, these two homomorphisms give a map
\begin{equation}
\label{eq:mm2}
\pi_{4}\twistedThom \xrightarrow{(e,u)}{} \Z_{2}\oplus \Z/16.
\end{equation}

\begin{prop}
\label{thm:mm3}
The map above gives an isomorphism of $\pi_{4}\twistedThom$ with the set of
elements $(a,b)\in \Z_{2}\oplus \Z/16$ with $a\equiv b\mod 2$.  The group
$\pi_{4}\twistedThom$ is generated by $(\RP^{4},\tilde c_{\RP^{4}})$ and $(\RP^{4}\#\RP^{4},0)$.
\end{prop}

\begin{proof}
By definition, the map 
\[
\pi_{4}\twistedThom \xrightarrow{e}{} \Z_{2} \to \Z/2
\]
is given by $\int_{M}w_{4}$, and the map 
\[
\pi_{4}\twistedThom \to\pi_{4}\mtPinPlus \to \pi_{4}\mtPinPlus\otimes \;\Z/2
= \Z/2
\]
is given by $\int_{M}w_{1}^{4}$.    Since $Sq^4\:H^0(M;\Z/2)\to H^4(M;\Z/2)$
is zero, the fourth Wu class~\eqref{eq:69} vanishes.  Therefore,
$\int_{M}(w_4 + w_1w_3 + w_2^2 + w_1^4)=\int_{M}(w_4 + w_1^4)=0$.
This implies 
\[
\int_{M}w_{1}^{4} = \int_{M}w_{4},
\]
so the image of~\eqref{eq:mm2} is contained in the subgroup of elements
$(a,b)$ with $a\equiv b\mod 2$.  On the other hand the Adams spectral
sequence shows that the kernel of $e$ has order at most $8$.  The
$\twistedThom$-manifold $(\RP^{4}\#\RP^{4},0)$ is in the kernel of
$e$.  Its image in $\pi_{4}\mtPinPlus=\Z/16$ is $2[\RP^{4}]$.  Since
$[\RP^{4}]$ generates $\pi_{4}\mtPinPlus$, the image of $(\RP^{4}\#\RP^{4},0)$ actually
has order $8$.  The assertion about generators follows from the
computation
\begin{align*}
(\RP^{4},\tilde c_{1}) &\mapsto (1,1) \\
(\RP^{4}\#\RP^{4},0) &\mapsto (0,2).
\end{align*}
This completes the proof.
\end{proof}

\begin{eg}
\label{eg:2}
If $M$ is a $\Spin$-manifold of dimension $4$ then under~\eqref{eq:mm2}
one has
\[
(M,\lambda) \to (\lambda(M),\lambda(M)).
\]
It follows that 
$(M,\lambda)\equiv \lambda(M)(\RP^{4},\tilde c_{1})$.
In particular
\[
[K,\lambda] = -24[\RP^{4},\tilde c_{1}] 
\]
when $K$ is a Kummer surface.
\end{eg}

\subsubsection{Dimension $12$}
\label{sec:mmdimension-12}

Our main result in dimension $12$ is the following restatement of
Theorem~\ref{thm:31}.

\begin{prop}
\label{thm:mm4}
The group $\pi_{12}\twistedThom$ is generated (over $\Z_{2}$)
by  the six manifolds 
\begin{gather*}
(W'_{0},\tilde c'_{0}),\quad
(W''_{0},0),\quad
(W_{1},\lambda) \\
(K\times \hp^{2},\lambda),\quad
(\RP^{4},\tilde c'_{\RP^{4}})\times B, \quad
(\RP^{4}\#\RP^{4},0)\times B.
\end{gather*}\vskip-2.75pc\qed
\end{prop}

The proof makes use of the following fact about $\Spin$-bordism, which we
prove in Appendix~\ref{sec:11}.

\begin{prop}
\label{thm:mm5} The group $\pi_{12}\mspin $ is free of rank $3$, and
generated by $K\times B$, $K\times \hp^{2}$, and the manifold $W_{1}$
described in \S\ref{subsubsec:5.5.3}, sitting in the fibration sequence
\[
\hp^{2}\to W_{1}\to \CP^{1}\times \CP^{1}.
\]
\vskip-2.2pc\qed
\end{prop}

\smallskip
\begin{proof*}[Proof of Proposition~\ref{thm:mm4}]
We begin by extracting some facts from the Adams spectral sequence.
First of all, the map
\begin{equation}
\label{eq:mm5}
\pi_{12}\twistedThom\to \Z/2\oplus \Z/2
\end{equation}
with components
\[
\int_{M}\alpha^{4}P(\iota)\quad\text{and}\quad \int_{M}\tr(\iota^{3}+\iota^{2}w_{4})
\]
gives an isomorphism of the quotient of $\pi_{12}\twistedThom$ by the
elements of positive Adams filtration with a subgroup of $\Z/2\oplus
\Z/2$.  The computations of examples~\ref{eg:1} and~\ref{eg:3} show
that this subgroup is in fact all of $\Z/2\oplus \Z/2$.  The kernel of
this map contains the image of $\pi_{12}\mspin $ and the image of
multiplication by $B$.  This follows from a consideration of Adams
filtrations, but it is easily checked directly.  Indeed if $M$ is a
$\Spin$-manifold, then $M$ is oriented, $\alpha=0$, and
$\int_{M}\alpha^{4}P(\iota)=0$.  Also, since the orientation double
cover of $M$ is $M\amalg M$, one has
\begin{align*}
\int_{M}\tr(\iota^{3}+\iota^{2}w_{4}) &=  \int_{M\amalg
M}(\iota^{3}+\iota^{2}w_{4}) \\ &=
2 \int_{M}(\iota^{3}+\iota^{2}w_{4}) = 0.
\end{align*}
In the
case $(M,c)=(N_{4},c')\times B$ all of the characteristic classes to
be integrated are pulled back from $H^{8}(N_{4})=0$.   

Let $J'\subset \pi_{12}\twistedThom$ be the subgroup generated by the image
of $\pi_{12}\mspin $ and the image of multiplication by $B$, and let 
\[
C=\pi_{12}\twistedThom/J'.
\]
A portion of the Adams spectral sequence for computing the map
$\pi_{12}\mspin \to \pi_{12}\twistedThom$ is shown in
Figure~\ref{fig:mm4}.  The map, which is part of the (machine)
computation of $\ext$, can also be determined by composing with the
map $\pi_{12}\twistedThom\to \pi_{12}\mtPinPlus$.  From it one can read off
that the map~\eqref{eq:mm5} gives an isomorphism
\[
C\otimes \Z/2\approx \Z/2\oplus \Z/2.   
\]
Since $C$ is finitely generated, Nakayama's Lemma and the computations
of Examples~\ref{eg:1} and~\ref{eg:3} show that $C$ is generated by
$(W'_{0},\tilde c'_{0})$ and $(W''_{0},0)$.

\begin{figure}
\includegraphics[scale=2]{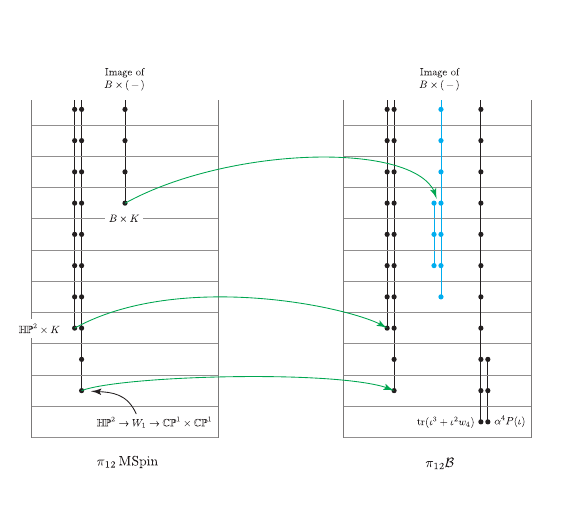}
  \vskip -2.5pc
\caption{The map from $\Spin$ bordism to $\twistedThom$-bordism in
dimension $12$}
\label{fig:mm4}
\end{figure}

Let
$J\subset\pi_{12}\twistedThom$ be the subgroup generated by 
\[
(W_{1},\lambda),\quad
(K\times \hp^{2},\lambda),\quad
(\RP^{4},\tilde c'_{\RP^{4}})\times B, \quad\text{and}\quad
(\RP^{4}\#\RP^{4},0)\times B.
\]
We are done if we can show that the manifolds
$(W'_{0},\tilde c'_{0})$ and
$(W''_{0},0)$ generate $\pi_{12}\twistedThom/J$.   
Note that by Proposition~\ref{thm:mm3}, the subgroup $J$ contains $J'$:
it contains 
image of multiplication by $B$ and by Proposition~\ref{thm:mm5} it
contains the image of $\pi_{12}\mspin$.    By the
above discussion, the manifolds $(W'_{0},\tilde c'_{0})$ and
$(W''_{0},0)$ generate $\pi_{12}\twistedThom/J'$, so they certainly generate
$\pi_{12}\twistedThom/J$.   This completes the proof.
\end{proof*}

\subsubsection{Dimension $11$}
\label{sec:mmdimension-11}

The ambiguity of the $M$ theory action has to do with the group
$\pi_{11}\twistedThom$.  In this section we offer a tentative
evaluation of this group.   Since the ambiguity involves the entire
group, we drop the convention that groups have been completed at $2$.

\begin{prop}
\label{thm:mm8}
The group $\pi_{11}\twistedThom$ is a finite abelian $2$-group.
\end{prop}

\begin{proof}
From the theory of Serre classes, one knows that
$\pi_{11}\twistedThom$ is finitely generated.   It therefore suffices
to show that 
\[
\pi_{11}\twistedThom\otimes \Z[1/2] =0.
\]
The homotopy fiber of the map $\twistedB\xrightarrow{w_{1}}{} K(\Z/2,1)$ is
$\bspin\times K(\Z,4)$, representing $\twistedB$ as the homotopy quotient of
an action of $\Z/2$ on $\bspin\times  K(\Z,4)$.   A map $T\to \bspin\times
K(\Z,4)$ 
classifies a pair consisting $(V,\iota)$ consisting of a $\Spin$
bundle $V\to T$ and an element $\iota\in H^{4}(T;\Z)$.    The pullback of the universal $w_1$-twisted integer lift $c$ of $w_{4}$ is 
\[
c=\lambda(V)-2 \iota
\]
(see~\eqref{eq:mm4}), and the generator of the $\Z/2$-action sends
$(V,\iota)$ to
\[
(V, \lambda(V)-\iota).
\]
By passing to Thom spectra this depicts
$\twistedThom$ as the homotopy quotient of the $\Z/2$ spectrum
$\mspin\wedge K(\Z,4)_{+}\wedge S^{1-\sigma}$, where $\sigma$ is the
sign representation of $\Z/2$.  This means that after inverting $2$,
the map 
\[
\pi_{\ast}\mspin \wedge K(\Z,4)_{+}\to \pi_{\ast}\twistedThom
\]
is projection to a summand.   The claim now follows from Stong's
Theorem~\cite{Sto1}, which states that $\pi_{11}\mspin \wedge
K(\Z,4)_{+}=0$. 
\end{proof}

From the discussion in the above proof it is an easy matter to compute
\[
\pi_{\ast}\twistedThom\otimes \Q.
\]
Let $J=(j_{1},\dots,)$ run through the sequences of non-negative
integers, almost all of which are $0$, and write 
\begin{align*}
|J| &= j_{1}+2 j_{2}+\cdots + n j_{n}+\cdots \\
p^{J} &= p_{1}^{j_{1}}p_{2}^{j_{2}}\cdots .
\end{align*}

\begin{prop}
\label{thm:mm2}
The group $\pi_{m}\twistedThom\otimes \Q$ is zero if $m$ is not divisible
by $4$.   The map 
\[
(M,c)\longmapsto  \int_{M} c^{2k+1}p^{J}
\]
gives an isomorphism 
\[
\pi_{4n} \twistedThom \otimes \Q \longrightarrow  \prod_{2k+1+|J| =n} \Q.
\]
\end{prop}

For example this implies that the group $\pi_{12}\twistedThom$ has  rank
$3$, corresponding to the indices $(k,J)=(0,(2,0))$, $(0,(0,1))$, $(1,(0))$.
This implies that there must be non-trivial differentials in the chart
Figure~\ref{fig:mm1} from dimension $13$ to dimension $12$.

\begin{remark}
\label{rem:4}
In~\cite{Sto1} (note after Item 6), Stong also shows that for $\ast< 12$ the groups
\[
\pi_{\ast}\mspin \wedge K(\Z,4)_{+}\otimes \Z[1/2]
\]
are torsion free.  In fact his argument for dimension $8$ can also be
adapted to dimension $12$ to establish the same conclusion for
$\ast=12$.  So the above also provides an evaluation of the groups
\[
\pi_{\ast}\twistedThom\otimes\Z[1/2], \qquad \ast \le 12.
\]
\end{remark}

Because of Proposition~\ref{thm:mm8}, the group $\pi_{11}\twistedThom$ can be
determined from the Adams spectral sequence, which is displayed in
Figure~\ref{fig:mm1}.  The $E_{2}$-term provides an upper bound and shows
that the group has order at most $8$.  In the table, there are two
$d_{3}$-differentials indicated, originating in Adams filtrations $1$ and
$2$.  These should be regarded as tentative at the moment, and will appear
in~\cite{GH}.  Assuming them, the Adams spectral sequence shows that after
$2$-completion the group $\pi_{11}\twistedThom$ is cyclic of order $2$, and
that an isomorphism is given by the mod $2$ index of the $\pinPlus$-Dirac
operator.  We state the outcome of this argument as a restatement of
Conjecture~\ref{thm:41}. 

\begin{conj}
\label{conj:1}
The map 
\[
\pi_{11}\twistedThom\to \Z/2
\]
given by the mod $2$ index of the $\pinPlus$ Dirac operator is an isomorphism.
\end{conj}

\begin{remark}
\label{rem:5} 
Let $M$ be the product of the Bott manifold, $S^{1}$ in
its non-bounding $\String$-structure, and $(\Sigma,0)$ where $\Sigma$
is the Klein bottle in a nonbounding $\pinPlus$-structure (see
\S\ref{sec:mmdimension-less-than}).  The mod $2$ index of the $\pinPlus$
Dirac operator on $M$ is $1$ and so the above conjecture implies that
$\pi_{11}\twistedThom$ is generated by $M$.
\end{remark}

\bigskip

\appendix

   \section{On the anomaly theory of a spinor field}\label{sec:7}

In an $n$-dimensional field theory~$F$ the partition function of a spinor
field on a closed $n$-dimensional Riemannian manifold is the pfaffian of a
Dirac operator, which is an element of a Pfaffian line, as reviewed
in~\S\ref{sec:3}.  The Pfaffian line is the quantum state space of the
associated \emph{anomaly theory}, which is an invertible $(n+1)$-dimensional
theory~$\alpha $, but initially truncated to manifolds of dimension~$\le n$,
since such manifolds form the domain of~ $F$.  To extend~$\alpha $ to an
$(n+1)$-dimensional theory we must define the partition function of a closed
$(n+1)$-manifold, as well as an element in the state space of the boundary of
a compact $(n+1)$-manifold with boundary, and these elements must satisfy a
gluing law.  The results in~\cite{DF} imply that an exponentiated $\eta
$-invariant works as the partition function: on a compact $(n+1)$-manifold
with boundary it takes values in the Pfaffian line.  But to define it we must
construct a Riemannian Dirac operator in $(n+1)$-dimensions from the
$n$-dimensional Lorentzian data which define the spinor field.  The
construction was given in~\cite[\S9.2.5]{FH}, but only in passing; in this
appendix we give more detail.  We discuss the base case of spin manifolds (no
time-reversal symmetry) in~\S\ref{subsec:7.1}.  In~\S\ref{subsec:7.2} we
specialize to 11~dimensions and the pin module relevant to M-theory.

  \subsection{The spin case in general dimensions}\label{subsec:7.1}

A spinor field in a relativistic quantum field theory is specified
by~\cite[\S6]{De} a real spin representation~$\SS$ of the Lorentz spin
group~$\Spin_{1,n-1}$ together with a symmetric positive\footnote{The
positivity condition is that $\Gamma (s,s)$~lie in the closure of the forward
timelike vectors in~$\RR^{1,n-1}$.}  $\Spin_{1,n-1}$-invariant map
  \begin{equation}\label{eq:3}
     \Gamma \:\SS\otimes \SS\longrightarrow \RR^{1,n-1}. 
  \end{equation}
Thus $\SS$~is an ungraded module over the even subalgebra~$\Cliff^0_{n-1,1}$
of the Clifford algebra with $n-1$ generators squaring to~$+1$ and a single
generator squaring to~$-1$.  The pair~$(\SS,\Gamma )$ Wick rotates to define
a Dirac operator on a Riemannian spin $n$-manifold~$X$ as follows.  First,
the complexification~$\SS_{\CC}$ is a module over the complex
algebra~$\Cliff^0_n(\CC)$, so restricts to a representation of the compact
spin group~$\Spin_n$.  Also, $\Gamma $~complexifies to a
$\Spin_n$-equivariant morphism
  \begin{equation}\label{eq:4}
     \hG\:(\CC^n)^*\otimes \SC\mstrut \longrightarrow \SC^*. 
  \end{equation}
Let $S\to X$ be the complex vector bundle on~$X$ associated to the $\Spin_n$
representation~$\SC$.  Then as usual the Dirac operator on~$X$ is the
composition
  \begin{equation}\label{eq:5}
     D_X = \hG\circ \nabla \:C^{\infty}(X;S)\longrightarrow
     C^{\infty}(X;S^*), 
  \end{equation}
where $\nabla $~is the covariant derivative on sections of $S\to X$.  The
operator~ $D_X$ is complex skew-adjoint.  (The metric on $S\to X$ is
constructed in the next paragraph.)  For $X$~closed this operator appears in
the Dirac form~\eqref{eq:31}, and its pfaffian is the fermionic path
integral~\eqref{eq:29}.
 
The construction of a Dirac operator on a Riemannian spin
$(n+1)$-manifold~$W$ from the data~$(\SS,\Gamma )$ uses the Clifford linear
Dirac operator~\cite[\S II.7]{LM} and Morita equivalence of Clifford
algebras.  By \cite[Corollary~6.2]{De} the data~$(\SS,\Gamma )$ determine a
unique $\zt$-graded $\Cliff_{n-1,1}$-module structure on~$\SS\oplus \SS^*$.
Let $\gamma ^0$~denote the Clifford generator with $(\gamma ^0)^2=-1$, acting
as an odd endomorphism of $\SS\oplus \SS^*$, and $\gamma ^1,\dots ,\gamma
^{n-1}$ the Clifford generators with $(\gamma ^i)^2=+1$.  Fix an inner
product on~$\SS\oplus \SS^*$ such that the finite group consisting of
products of the~$\gamma ^\mu $, $\mu =0,1,\dots ,n-1$ acts orthogonally.  (It
follows that $\Spin_n$~acts unitarily on~$\SC$, which induces the hermitian
metric on $S\to X$ used in~\eqref{eq:5}.)  Now
  \begin{equation}\label{eq:6}
     \Clp{(n+1)}\,\otimes\; (\SS\oplus \SS^*) 
  \end{equation}
is a real super vector space which carries a left action
of~$\Spin_{n+1}$---by left multiplication on $\Clp{(n+1)}$ tensored the
identity on the second factor---and a commuting left action of the real
super algebra
  \begin{equation}\label{eq:7}
     A=\Clm{(n+1)}\otimes \Cliff_{n-1,1}. 
  \end{equation}
Elements of~$\Clp{(n+1)}$ act by right multiplication on the first factor
of~\eqref{eq:6}, tensored with the identity on the second factor, which is
equivalent to a left action of $\Clm{(n+1)}$, the opposite super algebra.
Elements of~$\Cliff_{n-1,1}$ act by the identity on the first factor
of~\eqref{eq:6} tensored with the action above on the second factor.  Now the
left $\Spin_{n+1}$~action on~\eqref{eq:6} defines a bundle of real
$A$-modules over~$W$ as well as a Dirac operator on its sections which
commutes with the action of~$A$.  We claim this solves the problem of
defining an $(n+1)$-dimensional Riemannian Dirac operator from~$(\SS,\Gamma
)$ which can be used in the anomaly theory~$\alpha $.  To verify that claim
we must: (i)~define the exponentiated $\eta $-invariant of this operator, and
(ii)~identify the induced operator in $n$~dimensions with~\eqref{eq:5}.
 
For~(i) we use a (super) Morita equivalence of~$A$ with~$\Clm3$, ``canceling''
the last $n-1$~generators of $\Clm{(n+1)}$ with the $n-1$ positive generators
of~$\Cliff_{n-1,1}$.  (The cancellation identifies $\Clm{(n-1)}\otimes
\Clp{(n-1)}$ with the super algebra of endomorphisms of the vector
space~$\Clp{(n-1)}$, which is Morita trivial.)  Under the Morita isomorphism
left $A$-modules are identified with left $\Clm3$-modules, and so the
$A$-module~\eqref{eq:6}, rewritten as
  \begin{equation}\label{eq:9}
     \Clp{(n+1)}\,\otimes\; (\SS\oplus \SS^*) \cong \Clp{(n-1)} \,\otimes\,
      \Clp2\,\otimes\, (\SS\oplus \SS^*) ,
  \end{equation}
is identified with 
  \begin{equation}\label{eq:8}
     \begin{aligned} \EE &= \Hom\mstrut _{\Clp{(n-1)}\otimes
      \Clm{(n-1)}}\bigl(\Clp{(n-1)}\;,\;\Clp{(n-1)} \,\otimes\,
     \Clp2\,\otimes\, 
      (\SS\oplus \SS^*) \bigr) \\ &\cong \Clp2\,\otimes \,(\SS\oplus
      \SS^*).
  \end{aligned} 
  \end{equation}
Let $\delta ^1,\delta ^2$ be the generators of~$\Clp2$, and $\gamma ^0$~as
above the (negative) generator of the $\Cliff_{n-1,1}$-action on~$\SS\oplus
\SS^*$.  Then the~ $\Clm3$ which acts on~$\EE$ is generated by~$\gamma
^0,\delta ^1,\delta ^2$, with $\delta ^1,\delta ^2$~acting by \emph{right}
multiplication on~$\Clp2$, tensored with the identity on~$\SS\oplus \SS^*$.
Furthermore, $\Spin_{n+1}\subset \Clp{(n+1)}$ acts on~$\EE$ using left
multiplication by~$\delta ^1,\delta ^2$ on~$\Clp2$, tensored with the
identity on~$\SS\oplus \SS^*$, and by $\gamma ^1,\dots ,\gamma ^{n-1}$ acting
on~$\SS\oplus \SS^*$, tensored with the identity on~$\Clp2$.  The latter
actions determine an odd skew-adjoint Dirac operator~$D$ on sections of the
real vector bundle $E=E^0\oplus E^1\to W$ associated to the
representation~$\EE$ of~$\Spin_{n+1}$.  The operator~$D$ commutes with the
left $\Clm3$-action on~$E$.  Now $\delta ^1\delta ^2\in \Clm3$~acts as a
complex structure on~$E$ and $\gamma ^0\delta ^2\in \Clm3$~ acts as a complex
antilinear operator which squares to~$-\id_{E}$.  Thus $E$~has a quaternionic
structure.  (More simply, the ungraded algebra~$\Clm3^0$ is isomorphic to the
quaternion algebra.)  The even self-adjoint operator
  \begin{equation}\label{eq:10}
     D^0 := \gamma ^0\delta ^1\delta ^2D \:C^{\infty}(W;E^0)\to C^{\infty}(W;E^0)
  \end{equation}
commutes with this quaternionic structure.  Assume $W$~is compact without
boundary.  Then $D^0$~ is elliptic, so has a discrete spectrum and the
eigenspaces~$E^0_\lambda $ are finite dimensional quaternionic vector
spaces.  Let $a\in \RR$ be in the complement of the spectrum.  Define 
  \begin{equation}\label{eq:11}
     \eta _a(s) = \sum\limits_{\lambda \neq 0}\sign(\lambda
     -a)\;(\dim_{\CC}\Enl)\;|\lambda |^{-s}\;-\;\sign(a)\dim_{\CC}\ker
     D^0,\qquad \Re (s)>>0,  
  \end{equation}
where the sum is over the nonzero eigenvalues of~$D^0$ and $s$~is a complex
number.  According to the results of~\cite{APS1,APS2,APS3} the sum converges
to a holomorphic function of~$s$ if the real part of~$s$ is sufficiently
large, it has a meromorphic continuation to~$\CC$, and $s=0$~is a regular
point.  Then
  \begin{equation}\label{eq:12}
     \alpha (W) =\exp(2\pi i\eta _a(0)/4) 
  \end{equation}
is independent of~$a$.  It is the partition function of the anomaly theory.

We now verify~(ii) above, namely that the Dirac operator in $n$~dimensions
induced from~\eqref{eq:10} (see~\cite[(3.1)]{APS1}) can be identified
with~\eqref{eq:5}.  Let $X$~be an $n$-dimensional Riemannian spin manifold
and consider $W=\RR\times X$ with the product Riemannian metric and spin
structure.  Then $W$~has a Dirac operator~\eqref{eq:10}, which we view as
complex since it commutes with the complex structure~$\delta ^1\delta ^2$.
Choose a local orthonormal framing on~$X$, with basis numbered~$0,1,\dots
,n-1$, and a global coordinate~$t$ on~$\RR$.  Order the Clifford generators
of~$\Clp{(n+1)}$ as~$\delta ^1,\delta ^2,\gamma ^1,\dots ,\gamma ^{n-1}$, as
in the previous paragraph.  Then the Dirac operator~$D^0$ on~$W$ can be
written locally as
  \begin{equation}\label{eq:15}
     D^0 = \gamma ^0\delta ^1\delta ^2\bigl(\delta ^1\frac{\partial }{\partial
     t} + \delta ^2\nabla _0 + \gamma 
     ^i\nabla _i\bigr), 
  \end{equation}
where the sum over~$i$ runs from~$1$ to~$n-1$.  The symbol of~$D$, evaluated
on~$dt$, is induced from the algebraic operator
  \begin{equation}\label{eq:16}
     J(x\otimes \chi ) = -(-1)^{|x|}\,\delta ^1x\,\delta ^1\delta ^2\otimes
     \gamma ^0\chi ,\qquad x\in \Clp2,\quad \chi \in \SS\oplus \SS^*, 
  \end{equation}
where $|x|\in \zt$ is the parity of the homogeneous element~$x$.  Then
$J$~commutes with the complex structure~$\delta ^1\delta ^2$, anticommutes
with~$D^0$, and $J^2=-\id$.  According to~\cite[\S1]{DF} the induced Dirac
operator~$D^0_X$ on~$X$ is the operator~\eqref{eq:15} restricted to functions
on~$\RR\times X$ which are invariant under translation in~$t$, mapping the
$+\sqmo$-eigenspace of~$J$ to the $-\sqmo$-eigenspace of~$J$.  Thus, if now
$x\otimes \chi $~is a section of $E^0\to X$, we compute
  \begin{equation}\label{eq:17}
     D^0_X(x\otimes \chi )= -(-1)^{|x|} \,(\delta ^2\nabla _0 + \gamma ^i\nabla
     _i)\,(x\delta ^1\delta ^2\otimes \gamma ^0\chi ).
  \end{equation}
Write the $\pm\sqmo$ eigenbundles of~$J$ as $E^0_{\pm}\to X$, and recall the
complex vector bundle~$S\to X$ associated to~$\SC$.  Then there are
isomorphisms  
  \begin{equation}\label{eq:18}
     \begin{aligned} S&\longrightarrow E^0_+ \\ \psi \;+\;\sqmo\,\psi
     '&\longmapsto1\otimes \psi \;-\;\delta ^1\otimes \gamma ^0\psi
     \;-\;\delta ^2\otimes \gamma 
      ^0\psi \;+\;\delta ^1\delta ^2\otimes \psi '\end{aligned} 
  \end{equation}
and 
  \begin{equation}\label{eq:19}
     \begin{aligned} S^*&\longrightarrow E^0_- \\ \lambda  \;+\;\sqmo\,\lambda 
      '&\longmapsto - 1\otimes \gamma ^0\lambda'  \;+\;\delta ^1\otimes\lambda' 
      \;-\;\delta ^2\otimes\lambda  \;+\;\delta ^1\delta ^2\otimes \gamma ^0
     \lambda  \end{aligned} 
  \end{equation}
A straightforward computation demonstrates that these isomorphisms intertwine
the operators~$D^0_X$ in~\eqref{eq:17} and~$D_X$ in~\eqref{eq:5}, where the
latter is 
  \begin{equation}\label{eq:20}
     D_X = \sqmo\,\gamma ^0\nabla _0 + \gamma ^i\nabla _i 
  \end{equation}
in the local moving frame on~$X$.  The factor~$\sqmo$ comes from Wick
rotation when passing from~\eqref{eq:3} to~\eqref{eq:4}.

As a companion to~\eqref{eq:10} we have the operator 
  \begin{equation}\label{eq:32}
     D^1 := \gamma ^0\delta ^1\delta ^2D \:C^{\infty}(W;E^1)\to
     C^{\infty}(W;E^1) 
  \end{equation}
of the odd subspace of~$E=D^0\oplus E^1$.  Note that swapping~$\SS$
and~$\SS^*$ swaps the even and odd parts of~$E$; see~\eqref{eq:8}.   

  \begin{proposition}[]\label{thm:6}
 The exponentiated $\eta $-invariant formed with~$D^1$ is the reciprocal of
the exponentiated $\eta $-invariant~\eqref{eq:12} formed with~$D^0$. 
  \end{proposition}

  \begin{proof}
 From the definition of~$D$ following~\eqref{eq:8}, since $\Clm3$ graded
commutes with $\Clp{(n+1)}$ it follows that $\omega =\gamma ^0\delta ^1\delta
^2$ satisfies $\omega D=-D\omega $, and then since $D^i=\omega D$, $i=0,1$,
we deduce $\omega D^0=-D^1\omega $.  Therefore, the spectrum of~$D^1$ is the
negative of the spectrum of~$D^0$.  Then, distinguishing the $\eta
$-functions~\eqref{eq:11} for $D^0,D^1$, we have $\eta _a^0(s)=-\eta
_{-a}^1(s)$ for all~$s$.  The desired conclusion follows by analytic
continuation.
  \end{proof}

The exponentiated $\eta $-invariants are the partition functions of
invertible field theories~$\alpha ^0,\alpha ^1$, and the stronger version of
Proposition~\ref{thm:6} is that $\alpha ^0$~and $\alpha ^1$ are inverse
theories.  If both are topological, which is the case for the application to
M-theory, then the stronger assertion follows from Proposition~\ref{thm:6}
since the partition function determines the isomorphism class of the theory.
Here we will not attempt to justify the stronger assertion in the
non-topological case, nor the conjecture that $\alpha ^0\otimes \alpha ^1$
admits a \emph{canonical} trivialization.

  \subsection{The pin case in dimension~11}\label{subsec:7.2}
 
We describe the relevant pin representation and check that
\eqref{eq:15}~produces the Dirac operator in 12~dimensions which appears
in~\cite{St}.
 
We follow~\cite[\S9.2]{FH} in which the $\Pp$~case is described by a
parameter~$s=-1$.  The point is to use the embeddings~\cite[Lemma~9.25]{FH}
and \cite[(9.44)]{FH}, which specialize to
  \begin{equation}\label{eq:21}
     \begin{aligned} \Pp_{12} &\longrightarrow \Cliff^0_{12,1} \\ \gamma
      ^i&\longmapsto \gamma ^i\otimes \gamma ^-\end{aligned} 
  \end{equation}
and 
  \begin{equation}\label{eq:22}
     \begin{aligned} \Pin\mstrut _{10,1}&\longrightarrow \Cliff^0_{10,2} \\
     \gamma 
      ^i&\longrightarrow \gamma ^i\otimes \gamma ^-,\end{aligned} 
  \end{equation}
where $(\gamma ^-)^2=-1$ and $i=1,2,\dots ,12$.  These give embeddings of
groups $\Pp_{12}\hookrightarrow \Spin\mstrut _{12,1}$ and $\Pin\mstrut
_{10,1}\hookrightarrow \Spin\mstrut _{10,2}$.  The starting data is a real
representation of~$\Pin\mstrut _{10,1}$ obtained by restriction from an
ungraded real $\Cliff^0_{10,2}$-module.  There are isomorphisms
  \begin{equation}\label{eq:23}
     \Cliff_{10,2}\cong \Clp8\otimes \Cliff_{2,2}\cong \End(\MM^0\oplus
     \MM^1)\otimes \End(\Clp2) 
  \end{equation}
where $\MM^i$~is a real vector space of dimension~8.  A minimal real
$\Cliff^0_{10,2}$-module is the even subspace
  \begin{equation}\label{eq:24}
     \SS := \MM^0\otimes \Clp2^0 \;\oplus \; \MM^1\otimes \Clp2^1 
  \end{equation}
of $(\MM^0\oplus \MM^1)\otimes \Clp2$, which has real dimension~32.  (We
could as well take the odd subspace; see Proposition~\ref{thm:6}.)  The
restriction of~$\SS$ to $\Cliff^0_{10,1}\subset \Cliff^0_{10,2}$, or
equivalently to $\Spin_{10,1}\subset \Spin_{10,2}$, is irreducible.  (The
$\Clp8$ in~\eqref{eq:23} splits off and one simply checks for
$\Cliff_{2,1}\subset \Cliff_{2,2}$.)  By ~\cite[Theorem 6.1]{De} there is a
$\Spin_{10,1}$-invariant pairing~\eqref{eq:3}, unique up to a positive
scalar, and it is then automatically $\Pin_{10,1}$-invariant.  This defines
the starting data~$(\SS,\Gamma )$.
 
The Wick rotation on 12-manifolds, carried out in the second paragraph
of~\S\ref{subsec:7.1}, is modified in the first instance by
tensoring~\eqref{eq:6} with~$\Clm1$ and using the embedding~\eqref{eq:21}, of
course setting~$n=11$.  Then the commuting super algebra~\eqref{eq:7} is
$\Cliff_{1,12}\otimes \Cliff_{10,2}$, which as before is Morita equivalent
to~$\Clm3$. Then $\EE=\Clp2\otimes (\SS\oplus \SS^*)$ is as
in~\eqref{eq:8}, but is a left $\Cliff_{12,1}\otimes \Clm3$-module: the last
Clifford generator in~$\Cliff_{12,1}$ acts via the action of the last
Clifford generator on the $\Cliff_{10,1}$-module $\SS\oplus \SS^*$.  The even
subspace~$\EE^0\subset \EE$ has real dimension~128 and carries a quaternionic
structure, so is a 32-dimensional quaternionic vector space.  The resulting
representation of~$\Pp_{12}$ agrees with the one described at the end
of~\cite[\S3]{St}.  (Stolz distinguishes between two representations
of~$\Pp_{13}$, but they are isomorphic when restricted to~$\Pp_{12}$.)

   \section{Spin bordism in dimension $12$}\label{sec:11}

The purpose of this appendix is to give a proof of Proposition~\ref{thm:mm5},
which asserts that after $2$-adic completion the group
$\mspin_{4}=\pi_{12}\mspin$ is freely generated by the bordism class of the
manifolds $K\times B$, $K\times \hp^{2}$, and $W_{1} $. Here $B$ is a
Bott-manifold, $K$ is a $K3$-surface, $\hp^{2}$ is quaternionic projective
space, and $W_{1}$ is the manifold described in \S5.5.3.  This is done by
direct application of the computation of Anderson, Brown and
Peterson~\cite{ABP3}.  The material in the section owes much to conversations
with Meng Guo.

Following Anderson, Brown and
Peterson~\cite[\S4]{ABP2}, associated to an
oriented vector bundle $V$ over a space $X$ is the {\em total
$KO$-Pontrjagin class}
\[
\pi_{t}(V) = \sum_{n=0}^{\infty} \pi^{n}(V)\, t^{n} \in KO^{0}(X)\LL t \RRR.
\]
It is uniquely determined by the following properties:
\begin{enumerate}
\item\label{item:1} (Naturality)  If $f:Y\to X$ is a continuous map
then $\pi_{t}(f^{\ast}V)= f^{\ast}\pi_{t}(V)$.
\item\label{item:2} For vector bundles $V$ and $W$ one has
$\pi_{t}(V\oplus W)=\pi_{t}(V)\pi_{t}(W)$.
\item \label{item:3} If $V$ is an oriented $2$-plane bundle then
$\pi_{t}(V)= 1 + t(V-2)$.
\end{enumerate}

Because of property~(\ref{item:2}) the total $KO$-Pontrjagin
class $\pi_{t}(V)$ can be defined for virtual oriented vector bundles $V$.

In practice one often computes $\pi_{t}(V)$ by following the splitting
principle and finding a map $f:Y\to X$ for which $KO^{\ast}f$ is a
monomorphism and $f^{\ast}V$ is isomorphic to the oriented real vector
bundle underlying a sum of complex line bundles $L_{i}$.   In this case one has
\begin{align*}
\pi_{t}(f^{\ast}V) &= \prod (1+t(L_{i}^{u}-2)) \in KO^{0}(Y)\quad\text{and}\\
\pi_{t}(f^{\ast}V)\otimes \C &= \prod
(1+t(L_{i}+L_{i}^{-1}-2))\in K^{0}(Y),
\end{align*}
in which the notation $W^{u}$ is being used for the real vector bundle
underlying a complex vector bundle $W$.  

For a $\spin$-manifold $M$ of dimension $d$, and a sequence $J=(j_{1},
j_{2},\dots,)$, of non-negative integers $j_{i}$, with $j_{i}=0$ for $i\gg 0$, one defines the
$KO$-Pontrjagin number $\pi^{J}(M)\in KO_{d}$ to be the
index of the Clifford linear Dirac operator on $M$ coupled to the
virtual bundle
\[
\big(\pi^{1}(TM)\big)^{j_{1}}\, \big(\pi^{2}(TM)\big)^{j_{2}}\cdots .
\]
Anderson, Brown and
Peterson showed that the map
\begin{equation}
\label{eq:stw4} \mspin_{d}\xrightarrow{(\pi^{J}(M)\,,\, w^{N}(M))}{}
\prod_{J} KO_{d}\times \prod_{N} \Z/2
\end{equation}
is an isomorphism after completion at $2$.  Here $J$  runs over the sequences
$(j_{1},\dots)$ with $j_{1}=0$  and
\[
n(J)= (j_{1}+ 2 j_{2}+3 j_{3}+\cdots) \le \begin{cases}
d/4 \qquad &n(J)  \text{ even} \\
(d+2)/4 \qquad &n(J) \text{ odd}.
\end{cases}
\] The invariants $w^{N}(M)$ are certain Stiefel-Whitney numbers of $M$, and
don't occur in dimension less than $20$.  Both sides
of~\eqref{eq:stw4} are finitely generated abelian groups, so that the
property of being an isomorphism after $2$-adic completion is
equivalent to being an isomorphism after localization at $2$ and also
to being an isomorphism after reducing mod $2$.  
For further information
see~\cite[Corollary~1.4]{ABP1},~\cite[Theorem 2.2]{ABP3} or the Manifold
Atlas~\cite{MA}.  We alert the reader consulting these sources that our
convention associating Pontrjagin numbers to sequences $J$ differs from the
one used in these references.

For the purpose of writing down $KO$-Pontrjagin numbers it is helpful
to choose a basis of $KO_{4k}$ that is compatible with
multiplication.  We do this by identifying $KO_{4k}$ with its image in
$K_{4k}$.  Writing $v_{1}\in K_{2}$ for the Bott periodicity class, a
basis for the image of $KO_{\ast}$ consists of the elements
\[
\{v_{1}^{8k}, 2v_{1}^{8k+4} \}.
\]
For convenience we will use the same names for basis elements of
$KO_{4k}$, with the reminder that the element $2v_{1}^{8k+4}\in KO_{8k+4}$ is not
divisible by $2$.

We are interested in dimensions $d= 4,8, 12$ where the $KO$-Pontrjagin
numbers one encounters are those for which $J$ is the zero sequence,
in which case $\pi^{J}(M)$ (which we denote $\pi^{0}(M)$) is the index
of the Clifford linear Dirac operator on $M$, or the sequence $J_{i}$
whose only non-zero entry is a $1$ in the $i^{\text{th}}$ spot.  In
this latter case $\pi^{J_{i}}(M)$ is the index of the Clifford linear
Dirac operator on $M$ with coefficients in $\pi^{i}(TM)$.  We write
\[
\pi^{i}(M)= \pi^{J_{i}}(M),
\]
and note that sending $M$ to the total class 
\[
\pi_{t}(M) = \sum_{i\ge 0} \pi^{i}(M)\, t^{i}
\]
defines a ring homomorphism 
\[
\mspin_{\ast} \to KO_{\ast}\LL t\RRR.   
\]

\begin{remark}
\label{rem:1a} Hopefully our notation will not be confusing to the
reader.  The symbol $\pi_{t}(\slot)$ has a meaning that depends on the
mathematical type of the argument.  For a vector bundle $V$ over a
space $X$, $\pi_{t}(V)$ is an element of $KO^{0}(X)\LL t\RRR$.  If $M$
is a $\spin$-manifold of dimension $d$ then $\pi_{t}(M)$ is an element
of $KO_{d}\LL t\RRR$.  One has
\[
\pi_{t}(M) = f_{\ast}\pi_{t}(TM)
\]
where $f$ is the unique map from $M$ to a point, and $f_{\ast}$ is
the pushforward in $KO$-theory.
\end{remark}

By~\cite[Theorem~4.6]{ABP2}, when $M$ has dimension $d = 4 i$
the $KO$-Pontrjagin number $\pi^{i}(M)$ is the same as the ordinary
Pontrjagin number $p_{i}(M)$
\begin{equation}
\label{eq:stw1}
\pi^{i}(M) = p_{i}(M) v_{1}^{2i}.
\end{equation}
Note that in terms of our chosen basis for $KO_{4i}$, when $i$ is odd,
\eqref{eq:stw1} means that 
\[
\pi^{i}(M) = (2v_{1}^{2i})\,  \frac{p_{i}(M)}{2}.
\]

We now turn to our specific manifolds.  In dimension $4$ we have (by~\eqref{eq:stw1})
\[
\pi_{t}(M) = (2 v_{1}^{2})(-\frac{p_{1}}{48}+ \frac{p_{1}}{2}\, t).
\]
For a $K3$-surface one has $p_{1}=-48$ (see~\eqref{eq:98}), and so
\[
\pi_{t}(K) = (2 v_{1}^{2})(1 - 24 t).
\]
In the Anderson, Brown, Peterson isomorphism for $d=4$, only the zero sequence $J$
occurs, and the map
\[
\pi^{0}(M):\mspin_{4}\to KO_{4}
\]
is an isomorphism after  completing at $2$.  Since $\pi^{0}(K)= (2
v_{1}^{2})$ the composition 
\[
\Z\xrightarrow{K}{} \mspin_{4}\xrightarrow{\pi^{0}}{} KO_{4}
\]
is an isomorphism, and so 
\[
\Z\xrightarrow{K}{} \mspin_{4}
\]
is an isomorphism after completing at $2$.   Less formally, the
$2$-adic completion of $\mspin_{4}$ is freely generated ($2$-adically) by $K$.

In dimension $8$, a little calculation
gives 
\begin{equation}
\label{eq:stw2}
\pi_{t}(M) = v_{1}^{4}\left(\frac{7 p_{1}^{2}-4 p_{2}}{5760}+ 
\frac{p_{1}^{2} - 4 p_{2}}{24}\, t + p_{2}\, t^{2} \right)[M].
\end{equation}
For the Bott manifold, $p_{1}^{2}=0$ and $p_{2}=-1440$ (see~\eqref{eq:105}), and so 
\[
\pi_{t}(B)= v_{1}^{4}(1+ 240 t -1440 t^{2}).
\]
For $\hp^{2}$ one can plug the values of the Pontrjagin classes
from~\eqref{eq:103} into~\eqref{eq:stw2} to get
\[
\pi_{t}(\hp^{2}) = v_{1}^{4}(-t + 7 t^{2}),
\]
though it is a little more instructive to work through the splitting
principle approach (after all, that was what was used to determine the
Pontrjagin classes in the first place).  Let $f:\CP^{5}\to \hp^{2}$ be
the map
\[
S^{11}/U(1) \to S^{11}/SU(2).
\]
As described in~\S\ref{subsec:5.2}, the $KO$ class of $f^{\ast}T\hp^{2}$ is the
$KO$ class underlying the virtual complex bundle
\[
3(\mathcal O(1)+\mathcal O(-1)) - (\mathcal O(2)+1).   
\]
This means that
\begin{align*}
f^{\ast}\pi_{t}(T\hp^{2})&= \frac{(1+t (\mathcal
O(1)_u-2))^{3}(1+t(\mathcal O(-1)_{u}-2)))^{3}}{(1+
t(\mathcal O(2)_{u}-2))} \\
&= \frac{(1+t (\mathcal
O(1)_u-2))^{6}}{(1+
t(\mathcal O(2)_{u}-2))} 
\end{align*}
since $\mathcal O(n)_{u}=\mathcal O(-n)_{u}$.

We need to complexify the above and compute the Chern character.
Complexifying gives
\begin{align*}
f^{\ast}\pi_{t}(T\hp^{2})\otimes\C &= \frac{(1+t
(\mathcal O(1)+\mathcal O(-1)-2))^{6}}{(1+
t(\mathcal O(2)+\mathcal O(-2)-2))}\in K^{0}(\CP^{5}) \\ \intertext{and so}
f^{\ast}\ahat(T\hp^{2}) &= (y/(e^{y/2}-e^{-y/2}))^{6}(2y/(e^{y}-e^{-y}))^{-1}\\
\ch(f^{\ast}\pi_{t}(T\hp^{2})) &= (1+t (e^{y}+e^{-y}-2))^{6}/(1+ t(e^{2y}+e^{-2y}-2)) \\
f^{\ast}\ahat(T\hp^{2})\ch(\pi_{t}(T\hp^{2}))&= 1+\frac{1}{12} (-1+24 t) x+\left(-t+7 t^2\right)
x^{2}
\end{align*}
in which $y=c_{1}(\mathcal O(1))$ is the generator of $H^{2}(\CP^{5})$
and $x=y^{2}$.   From the coefficient of $x^{2}$ we recover the formula
\begin{equation}
\label{eq:stw3}
\pi_{t}(\hp^{2}) = v_{1}^{4}(-t + 7 t^{2}).
\end{equation}

In dimension $8$ the Anderson, Brown, Peterson map is
\[
(\pi^{0},\pi^{2}):\mspin_{8}\to (KO_{8})^{2}.
\]
The matrix of the composite
\[
\Z^{2} \xrightarrow{
\begin{bmatrix}B & \hp^{2} \end{bmatrix}}
\mspin_{4} \xrightarrow{
\begin{bmatrix}
\pi^{0} \\ \pi^{2}
\end{bmatrix} } (KO_{8})^{2}
\]
is 
\[
\begin{pmatrix}
1 & 0 \\
-1440 & 7
\end{pmatrix}
\]
 which is the identity modulo $2$.  This implies that after $2$-adic
completion, (or even just localization at $2$), $\mspin_{8}$ is freely
generated by $\{B,\hp^{2} \}$.  

Now for dimension $12$.   The invariants are $\pi^{0}$, $\pi^{2}$ and
$\pi^{3}$, and we need to calculate them on $K\times B$, $K\times
\hp^{2}$  and $W_{1}$.   Using the fact that $\pi_{t}(M)$ is
multiplicative in the $\spin$-manifold $M$ we find 
\begin{align*}
\pi_{t}(K\times B) &= (2 v_{1}^{6})(1-24 t)(1+240 t - 1440 t^{2})\\
&= (2 v_{1}^{6})(1 + 216 t + -7200 t^{2}+34560 t^{3})
\end{align*}
and 
\begin{align*}
\pi_{t}(K\times \hp^{2}) &= (2 v_{1}^{6})(1-24 t)(-t + 7 t^{2})\\
&= (2 v_{1}^{6})(-t + 31 t^{2}-168 t^{3}).
\end{align*}
We also know from~\eqref{eq:152} that $p_{3}(W_{1})= 6$ so that 
\[
\pi_{3}(W_{1}) =  6v_{1}^{6} = (2 v_{1}^{6})(3).
\]
The matrix of the composition
\[
\Z^{3} \xrightarrow{\begin{bmatrix}
K\times B & K\times \hp^{2} & W_{1}
\end{bmatrix} }{} \mspin_{12} \xrightarrow{\begin{bmatrix}
\pi^{0} \\ \pi^{2} \\ \pi^{3}
\end{bmatrix}}{} (KO_{12})^{3}
\]
is therefore
\[
\begin{pmatrix}
1& 0 & \ast \\
-7200& 31& \ast \\
34560&  -168& 3
\end{pmatrix}
\]
which reduces modulo $2$ to 
\[
\begin{pmatrix}
1& 0 & \ast \\
0& 1& \ast \\
0&  0& 1
\end{pmatrix}
\]
which is clearly invertible.  This shows that $\{K\times B, K\times \hp^{2},
W_{1} \}$ freely generate the $2$-adic completion of $\mspin_{12}$,
which is the assertion of Proposition~\ref{thm:mm5}.

   \section{Cohomology of $\twistedThom$}\label{sec:10}

As explained earlier, the Adams spectral sequence chart in \S9  was
determined by specifying a basis for $H^{\ast}(\twistedThom)$ and then
using a Mathematica program to encode the action of the Steenrod
operations in this basis and generate the input for Bruner's program~\cite{Br}
for computing the cohomology of the Steenrod algebra with coefficients
in a given module.  In this appendix we describe the cohomology ring
$H^{\ast}(\twistedB)$, the module $H^{\ast}(\twistedThom)$, and the
action of the Steenrod algebra.  Using this, the interested reader
should be able to reproduce the computation of the $E_{2}$-term of the
Adams spectral sequence.

Throughout this appendix all cohomology will be with coefficients in
$\Z/2$.

We begin with the cohomology of $\bspin$.  For a number $n=\sum
\epsilon_{i}2^{i}$, $\epsilon\in\{0,1 \}$ write $\alpha(n)=\sum
\epsilon_{i}$.

\begin{prop}[Thomas \cite{T}] \label{thm:mb1}
The map $\bspin\to BSO$ induces an isomorphism 
  \begin{qedequation}
\Z/2[w_{i}\mid \alpha(i-1)>1]= \Z/2[w_{4},w_{6},\dots]\approx H^{\ast}(\bspin).
  \qedhere\end{qedequation}
  \renewcommand{\qedsymbol}{}
\end{prop}

\begin{remark}
\label{rem:mb1}
The kernel of the map 
\[
H^{\ast}(BSO) = \Z/2[w_{2},w_{3},\dots] \to H^{\ast}(\bspin)
\]
is the regular ideal with generators 
\begin{equation}
\label{eq:mb1}
\sq^{2^{n}}\sq^{2^{n-1}}\cdots \sq^{1} w_{2}.
\end{equation}
It is {\em not} generated by the classes $w_{i}$ with $\alpha(i-1)=1$.
The first place these ideals differ is in dimension $17$.  In
$H^{\ast}(\bspin)$ one has
\[
w_{17}=w_7 w_{10}+w_6 w_{11}+w_4 w_{13}.
\]
\end{remark}

Consider the fibration sequence
\begin{equation}
\label{eq:mb2}
\bstring \to \bspin\xrightarrow{\lambda}{}K(\Z,4).
\end{equation}

\begin{prop}[Stong~\cite{Sto2}]
The map $\bstring\to BO$ induces an isomorphism
\[
\Z/2[w_{i}\mid \alpha(i-1)>2]=\Z/2[w_{8},w_{12},\dots] \approx H^{\ast}(\bstring).
\]
\end{prop}

\begin{remark}
\label{rem:mb2} The kernel of $H^{\ast}(\bspin)\to H^{\ast}(\bstring)$
is the regular ideal generated by the elements $\sq^{I}w_{4}$ in which
$I=(i_{1},i_{2},\dots, i_{k})$ is a sequence of non-negative integers
satisfying 
\begin{gather*}
i_{\ell}\ge 2 i_{\ell} \qquad 1\le \ell \le k \\
(i_{1}-2 i_{2})+\cdots + (i_{k-1}-2 i_{k}) + i_{k} < 4 \\
i_{k}>1.
\end{gather*}
The first condition is called {\em admissibility} and the quantity on
the left side of the second inequality is the {\em excess} of the
sequence.  The kernel is {\em not} the ideal generated by the $w_{i}$
with $\alpha(i-1)=2$.
\end{remark}

The Leray-Hirsch theorem applied to~\eqref{eq:mb2} implies

\begin{prop}
\label{thm:mb2}
The map $\bspin\to K(\Z,4)$ induces an isomorphism 
  \begin{qedequation}
H^{\ast}(K(\Z,4))[w_{i}\mid \alpha(i-1)>2]=
H^{\ast}(K(\Z,4))[w_{8},w_{12},\dots]\xrightarrow{\approx}{} H^{\ast}(\bspin).
  \qedhere\end{qedequation}
  \renewcommand{\qedsymbol}{}
\end{prop}

\begin{remark}
\label{rem:mb3}
To work out the action of the Steenrod operations under the above
isomorphism, one must use the relations setting the terms~\eqref{eq:mb1}
equal to zero, and the Wu
formulae.   This can get complicated.   For example one has 
\begin{align*}
\sq^{1}w_{16} &= w_{17} \\
&= w_7 w_{10}+w_6 w_{11}+w_4 w_{13} \\
&= (\sq^{4}\sq^{2}w_4)(\sq^{3}w_4)+
(\sq^{5}\sq^{2}w_4)(\sq^{2}w_4) \\
&\quad +
(\sq^{6}\sq^{3}w_4)(w_4) +
(\sq^{2}w_4)(\sq^{3}w_4)\,w_{4} \\
\end{align*}
\end{remark}

From~\eqref{eq:mb2} one constructs a fibration sequence
\begin{equation}
\label{eq:mb3}
\bstring \to \bspin\times K(\Z,4) \xrightarrow{(\iota, -\lambda-\iota)}{}
K(\Z,4)\times K(\Z,4).
\end{equation}
in which the rightmost map is equivariant when $K(\Z,4)\times K(\Z,4)$ is
given permutation $\Z/2$-action.

For a space $X$ write
\[
D_{2}(X) = E\Z/2\underset{\Z/2}{\times}(X^{2})
\]
in which $X^{2}=X\times X$ is equipped with the permutation action.   
Passing to homotopy orbits from~\eqref{eq:mb3} gives the fibration sequence
\[
\bstring \to \twistedB \to D_{2}(K(\Z,4)). 
\]
As before the Leray-Hirsch theorem leads to an isomorphism 
\[
H^{\ast}(D_{2}(K(\Z,4))[w_{8},w_{12},\dots] \xrightarrow{\approx}{} H^{\ast}(\twistedB).
\]

To go further we must describe the cohomology of $D_{2}(X)$.  The
computation of the cohomology ring of $D_{2}(X)$ is due to Dyer and
Lashof~\cite{DL}, and the action of the Steenrod operations was
determined by Nishida~\cite{N}.   To describe the computation
we first
recall the additive and multiplication transfers described in
\S\ref{sec:mmtransfer}.  Suppose that
\[
\pi:\widehat{M}\to M
\]
is a double cover, classified by 
\[
\alpha\in H^{1}(M;\Z/2),
\]
and write 
\[
\tau:H^{\ast}(\widehat M) \to H^{\ast}(\widehat M)
\]
for the map induced by the deck transformation.  There are additive
and multiplicative transfers
\begin{align*}
\tr:H^{k}(\widehat{M};\Z/2) &\to H^{k}(M;\Z/2) \\
P=P_{\alpha}:H^{k}(\widehat{M};\Z/2) &\to H^{2k}(M;\Z/2)
\end{align*}
which are natural in the sense that they commute with base change in
$M$.   They satisfy the following properties (for $x,y\in
H^{\ast}(\widehat{M})$, $z\in H^{\ast}(M)$)

 \begin{enumerate}
\item $\tr(x+y)= \tr(x)+\tr(y)$
\item\label{item:1a} $\tr(x)\, z = \tr(x \, \pi^{\ast}(z))$
\item $\pi^{\ast}\tr(x)= x + \tau(x)$
\item $P(x y)= P(x)P(y)$
\item $P(x+y)= P(x)+P(y)+ \tr(x\, \tau(y))$
\item $\pi^{\ast}P(x) = x \, \tau(x)$.
 \end{enumerate}
Note that property~(2) implies that
$\tr(x)\alpha=0$.   

We will refer to the additive transfer $\tr$ simply as the {\em
transfer} and the multiplicative transfer $P$ as the {\em norm}.  

The transfer map commutes with Steenrod operations 
\[
\sq^{k}\tr(x) = \tr(\sq^{k}x).
\]
Suitably interpreted, the norm also commutes with the Steenrod
operations.  Let $\beta\in
H^{1}(B\Z/2)$ be the non-zero element and for $x\in H^{n}(X)$  write
\[
\sq^{\beta}(x) = \sum \sq^{n-i}(x)\beta^{i} = x^{2}+ \beta
\sq^{n-1}(x) +\cdots + \beta^{n}x \in H^{2n}(B\Z/2\times X).
\]

As Nishida~\cite{N} observed, the values of the Steenrod
operations on the norm $P(x)$ are determined by the formula
\begin{equation}
\label{eq:mb4}
 \sq^{\beta}(P(x)) = P(\sq^{\beta}(x)).
\end{equation}

\begin{prop}[\cite{DL}, Proposition 2.2]
\label{thm:mb3}
Suppose $X$ is a space and $\{e_{i} \}$ is a basis of
$H^{\ast}(X;\Z/2)$.   The vector space
$H^{\ast}(D_{2}(X))$ has basis 
\[
\{\tr(e_{i}\otimes e_{j})\mid i< j\}\cup\{ \alpha^{j} P(e_{i}\otimes
1)\mid j\ge 0 \}.
\]
\end{prop}

To relate the elements $P(x)$ and $\tr(x)$ to other naturally
occurring elements it is useful to exploit both the covering map 
\[
\pi:E\Z/2\times X\times X \to D_{2}(X))
\]
and the diagonal map 
\[
\Delta:B\Z/2\times X\to D_{2}(X),
\]
the latter obtained by passing to homotopy orbits from the diagonal inclusion
(denoted by the same symbol)
\[
\Delta: X \to X\times X.
\]
From the definition of Steenrod operations one has 
\[
\Delta^{\ast}(P(x)) = \sq^{\alpha}(x)
\]
where $\alpha\in H^{1}(B\Z/2)$ is the non-zero element.  
From naturality, and the pullback square
\[
\xymatrix{
E\Z/2\times X  \ar[r]\ar[d]  &  E\Z/2\times X\times X
\ar[d] \\
B\Z/2\times X  \ar[r]        & D_{2}(X),
}
\]
one also sees that 
\[
\Delta^{\ast}(\tr(x)) = \tr(\Delta^{\ast}(x)) = 0.
\]

The next result also follows from~\cite[Proposition~2.2]{DL}
\begin{prop}
\label{thm:mb4}
If $\pi^{\ast}(x)=0$ and $\Delta^{\ast}(x)=0$ then $x=0$.
\end{prop}

In practice it is easier not to separate out the factor of
$D_{2}K(\Z,4)$ but rather to work directly with $\bspin\times
K(\Z,4)$.  In these terms the diagonal inclusion corresponds to the
map
\begin{equation}
\label{eq:mb1a}
\bspin\langle w_{4} \rangle \xrightarrow{(\id, -\lambda/2)}{} \bspin \times K(\Z,4)
\end{equation}
in which $\bspin\langle w_{4} \rangle$ is defined by left square in
the diagram  of homotopy pullback squares
\begin{equation}
\label{eq:mb8}
\xymatrix{
\bspin \langle w_{4} \rangle  \ar[r]^-{\lambda/2}\ar[d]  &  K(\Z,4) \ar[r]\ar[d]^-{2}   &  \ast \ar[d]\\
\bspin  \ar[r]_-{\lambda}        &  K(\Z,4) \ar[r]         &
K(\Z/2,4)\mathrlap{\ .}
}
\end{equation}
Since the right square in~\eqref{eq:mb8} is a homotopy pullback square,
the space $\bspin \langle w_{4} \rangle$ can is canonically equivalent
to homotopy fiber of the map
\[
\bspin \xrightarrow{w_{4}}{} K(\Z/2,4).   
\]
The left square in~\eqref{eq:mb8} gives a cohomology class $\lambda/2\in
H^{4}(\bspin \langle w_{4} \rangle)$ with the property that
$2(\lambda/2)$ is the restriction of $\lambda$.

One arrives at~\eqref{eq:mb1a} by considering first
\[
\xymatrix{
 & \bspin\times K(\Z,4) 
\ar[d]^{(\iota,-\lambda-\iota)}\ar[r]^-{2\iota+\lambda} & K(\Z,4)\ar@{=}[d]\\
K(\Z,4)  \ar[r]_-{\text{diag}}        &  K(\Z,4)\times
K(\Z,4)\ar[r]_-{[1,-1]} & K(\Z,4)
}
\]
to identify the homotopy pullback of the left horizontal and vertical
maps with the homotopy fiber of the map $(2\iota+\lambda)$, and the
homotopy pullback square
\[
\xymatrix{
\bspin\times K(\Z,4)  \ar[r]^-{2\iota+\lambda}\ar[d]  & K(\Z,4)
\ar[d] \\
\bspin  \ar[r]_{w_{4}}        & K(\Z/2,4),
}
\]
to identify the fiber of $(2\iota+\lambda)$ with $\bspin \langle w_{4} \rangle$.

Passing to homotopy orbits from 
\[
\bspin \langle w_{4} \rangle \to \bspin \times K(\Z,4) \to \bspin 
\]
one sees that the composition
\[
B\Z/2\times \bspin \langle w_{4} \rangle \to \twistedB \to \bpinplus
\approx B\Z/2\times \bspin 
\]
is the product of the identity map and the map $\bspin
\langle w_{4}  \rangle\to \bspin $.   It follows that the virtual
vector bundle classified by 
\[
B\Z/2\times\bspin \langle w_{4} \rangle \to \twistedB\to BO
\]
is the sum of the pullback of the tautological line bundle $L$ on
$B\Z/2$ and the pullback of the virtual vector bundle $V$ classified
by 
\[
\bspin \langle w_{4} \rangle \to \bspin \to BO.
\]
(See \S\ref{sec:mmchar-class}).  With this in hand one can appeal to the
pullback square
\[
\xymatrix{
\bspin\langle w_{4} \rangle  \ar[r]\ar[d]  & \bspin \times K(\Z,4)
\ar[d] \\
B\Z/2 \times \bspin \langle w_{4} \rangle  \ar[r]        & \twistedB
}
\]
and Proposition~\ref{thm:mb4} to conclude that the map from $H^{\ast}(\twistedB)$ to the
pullback of 
\[
\xymatrix{
 & H^{\ast}(\bspin \times K(\Z,4))
\ar[d] \\
H^{\ast}(B\Z/2)\times \bspin\langle w_{4} \rangle  \ar[r]        & H^{\ast}(\bspin \langle w_{4} \rangle)
}
\]
is a monomorphism, and identify classes like $\tr(e)$ and $P(e)$ in
terms of Stiefel-Whitney classes when possible.   

\begin{eg}
\label{eg:1a}
Consider the class $\tr(\iota)$.   The restriction of $\tr(\iota)$ to
$H^{\ast}(\bspin\times K(\Z,4))$ is 
\[
\iota+\tau(\iota) = \iota + (\iota+w_{4}) = w_{4}.
\]
This suggests that $\tr(\iota)=w_{4}$.  To verify this one need only
check that the restriction of $w_{4}$ to $B\Z/2\times \bspin \langle
w_{4} \rangle$ is zero.   For this one computes 
\[
w_{4}(L\oplus V) = w_{4}(V) + w_{1}(L)w_{3}(V) = 0.
\]
It follows that $\tr(\sq^{I}\iota) = \sq^{I}(w_{4})$.   
\end{eg}

\begin{remark}
\label{rem:mb4}
For the interested reader the map 
\[
H^{\ast}(\bspin )\to H^{\ast}(\bspin \langle w_{4} \rangle)
\]
factors uniquely as
\[
H^{\ast}(\bspin )\to H^{\ast}(\bstring ) \to H^{\ast}(\bspin \langle w_{4} \rangle).
\]
Using this, the pullback square
\[
\xymatrix{
\bspin \langle w_{4} \rangle  \ar[r]\ar[d]  & K(\Z,4)
\ar[d]^{2} \\
\bspin   \ar[r]        & K(\Z,4)
}
\]
and the Eilenberg-Moore spectral sequence implies that the map 
\begin{equation}
H^{\ast}(\bstring )\otimes H^{\ast}(K(\Z,4)) \to H^{\ast}(\bspin \langle w_{4} \rangle)
\end{equation}
is an isomorphism of algebras over the Steenrod algebra.
\end{remark}

To put this all together, let $\{e_{i} \}$ be the basis of $H^{\ast}(K\Z,4)$
consisting of monomials in the admissible Steenrod operations on
$\iota$, and $\{w_{I}\}$ the basis of monomials in the Stiefel-Whitney classes
\[
\{w_{i}\mid \alpha(i-1)>2 \}.
\]
Then a basis of $H^{\ast}(\twistedB)$ is given by 
\begin{equation}
\label{eq:mb5}
\{ w_{I}\tr(e_{i}\tau(e_{j}))\mid i<j\}\cup \{w_{I}\alpha^{j}P(e_{i}) \}.
\end{equation}
The Steenrod operations, products, and the relation with the
Stiefel-Whitney classes are as described above.   

The module $H^{\ast}(\twistedThom)$ is a free module over
$H^{\ast}(\twistedB)$ on the Thom class $U$ of the vector bundle
classified by the {\em negative} of the map 
\[
\twistedB\to \bpinplus\to BO
\]
As described in \S\ref{sec:mmchar-class} we use the equivalence
\[
B\Z/2\times \bspin\approx\bpinplus 
\]
obtained by writing a stable $\pinplus$ bundle in the form $L\oplus V$
and write $\alpha=w_{1}(L)$, $w_{i}=w_{i}(V)$.   
With the these conventions, the Thom formula the total squaring
operation on $U$ gives
\begin{equation}
\label{eq:mb6}
\sq_{t}(U) = (1+t\alpha)^{-1}(1+ t^{4}w_{4}+\cdots)^{-1} U.
\end{equation}

For the computations reported in this paper, the authors restricted to
dimensions less than $16$ and used the basis consisting of the product
of the basis elements in~\eqref{eq:mb5} with the Thom class~ $U$.

 \bigskip\bigskip
\providecommand{\bysame}{\leavevmode\hbox to3em{\hrulefill}\thinspace}
\providecommand{\MR}{\relax\ifhmode\unskip\space\fi MR }
\providecommand{\MRhref}[2]{%
  \href{http://www.ams.org/mathscinet-getitem?mr=#1}{#2}
}
\providecommand{\href}[2]{#2}

  \end{document}